\documentclass[12pt,draftclsnofoot,onecolumn,journal]{IEEEtran}
\usepackage{latexsym}
\usepackage{graphicx}
\usepackage{amsfonts,amssymb,amsmath,cite}
\usepackage[colorlinks=true,linkcolor=black,anchorcolor=black,citecolor=black,filecolor=black,menucolor=black,runcolor=black,urlcolor=black]{hyperref}
\usepackage{bbm}
\usepackage{paralist}
\usepackage[nolist]{acronym}
\usepackage{algorithm,algorithmic}
\usepackage[dvipsnames]{xcolor}
\usepackage{amsthm}
\usepackage{amsmath}
\usepackage{tikz}
\usetikzlibrary{arrows,decorations.markings,decorations.pathreplacing,}
\usepackage{pgfplots}

\def\BibTeX{{\rm B\kern-.05em{\sc i\kern-.025em b}\kern-.08em T\kern-.1667em\lower.7ex\hbox{E}\kern-.125emX}}
\newcommand{\setX}{\mathbbmss{X}}

\newcommand{\setR}{\mathbbmss{R}}

\newcommand{\setZ}{\mathbbmss{Z}}

\newcommand{\setC}{\mathbbmss{C}}

\newcommand{\setK}{\mathbbmss{K}}

\newcommand{\alg}{\mathcal{A}}

\newcommand{\sinr}{\mathrm{SINR}}
\newcommand{\esnr}{\mathrm{ESNR}}

\newcommand{\zz}{\mathrm{z}}
\newcommand{\rmq}{\mathrm{q}}

\newcommand{\rmF}{\mathrm{F}}

\newcommand{\rmr}{\mathrm{r}}

\newcommand{\rmA}{\mathrm{A}}

\newcommand{\rms}{\mathrm{s}}

\newcommand{\rmj}{\mathrm{j}}
\newcommand{\rmm}{\mathrm{m}}
\newcommand{\ssr}{\mathrm{ssr}}
\newcommand{\ee}{\mathrm{e}}

\newcommand{\her}{\mathsf{H}}

\newcommand{\maxx}{\mathrm{max}}

\newcommand{\mar}{\mathcal{R}}

\newcommand{\maP}{\mathcal{P}}

\newcommand{\mam}{\mathcal{M}}

\newcommand{\maq}{\mathcal{Q}}

\newcommand{\mal}{\mathcal{L}}

\newcommand{\bxx}{\mathbf{x}}
\newcommand{\bff}{\mathbf{f}}
\newcommand{\bq}{\mathbf{q}}
\newcommand{\mpsi}{\boldsymbol{\psi}}
\newcommand{\bpi}{\boldsymbol{\varpi}}
\newcommand{\bss}{\mathbf{s}}
\newcommand{\bvv}{\mathbf{v}}
\newcommand{\buu}{\mathbf{u}}

\newcommand{\mh}{\mathbf{h}}

\newcommand{\bx}{{\boldsymbol{x}}}

\newcommand{\vv}{\mathrm{v}}

\newcommand{\yy}{\mathrm{y}}
\newcommand{\brr}{\mathbf{r}}

\newcommand{\set}[1]{\left\lbrace#1\right\rbrace}

\newcommand{\bw}{{\mathbf{w}}}

\newcommand{\diag}{{\mathrm{diag}}}

\newcommand{\bgg}{{\mathbf{g}}}
\newcommand{\btt}{{\mathbf{t}}}

\newcommand{\rd}{\mathrm{d}}

\newcommand{\bphi}{{\boldsymbol{\phi}}}

\newcommand{\trp}{\mathsf{T}}

\newcommand{\mI}{\mathbf{I}}
\newcommand{\mone}{\mathbf{1}}

\newcommand{\mG}{\mathbf{G}}
\newcommand{\mPhi}{\mathbf{\Phi}}
\newcommand{\mQ}{\mathbf{U}}

\newcommand{\mU}{\mathbf{U}}

\newcommand{\mM}{\mathbf{M}}

\newcommand{\mb}{\mathbf{b}}

\newcommand{\mT}{\mathbf{T}}
\newcommand{\mH}{\mathbf{H}}
\newcommand{\mW}{\mathbf{W}}

\newcommand{\malpha}{\boldsymbol{\alpha}}
\newcommand{\mbeta}{\boldsymbol{\beta}}
\newcommand{\mgamma}{\boldsymbol{\gamma}}

\newcommand{\Ex}[1]{\mathbb{E} \left\lbrace #1 \right\rbrace }

\newcommand{\norm}[1]{\left\lVert #1 \right\rVert}
\newcommand{\dbc}[1]{\left[ #1 \right]}
\newcommand{\brc}[1]{\left( #1 \right)}

\newcommand{\abs}[1]{\lvert #1 \rvert}

\theoremstyle{plain}
\newtheorem{theorem}{Theorem}

\newtheorem{definition}{Definition}

\theoremstyle{remark}
\newtheorem{remark}{Remark}

\theoremstyle{Lemma}
\newtheorem{lemma}{Lemma}

\newcommand\algorithmicinput{\textbf{Input:}}
\newcommand\INPUT{\item[\algorithmicinput]}

\DeclareMathOperator*{\argmin}{argmin}
\DeclareMathOperator*{\argmax}{argmax}

\newcounter{bar}

\newcommand{\mt}{\mathrm{t}}

\begin{document}

\title{Designing IRS-Aided MIMO Systems for Secrecy Enhancement}

\author{Saba Asaad, \IEEEmembership{Member, IEEE},
	Yifei Wu,
	Ali Bereyhi \IEEEmembership{Member, IEEE},\\
	Ralf R. M\"uller, \IEEEmembership{Senior Member, IEEE}, 
	Rafael F. Schaefer, \IEEEmembership{Senior Member, IEEE},
	and H. Vincent Poor, \IEEEmembership{Fellow, IEEE}
\thanks{Saba Asaad, Yifei Wu, Ali Bereyhi and Ralf R. M\"uller are with the Institute for Digital Communications, Friedrich-Alexander Universit\"at Erlangen-N\"urnberg, Germany, \texttt{\{saba.asaad, yifei.wu, ali.bereyhi, ralf.r.mueller\}@fau.de}. Rafael F. Schaefer is with the Chair of Communications Engineering and Security, University of Siegen, Germany, \texttt{rafael.schaefer@uni-siegen.de}. H.~Vincent~Poor is with the Department of Electrical and Computer Engineering, Princeton University, NJ 08544, USA, \texttt{poor@princeton.edu}.}
\thanks{This work was supported by Deutsche Forschungsgemeinschaft (DFG) under Grant MU 3735/7-1 Project-No. 409561515, and in part by the German Federal Ministry for Education and Research (BMBF) under Grant 16KIS1242.}
}

\IEEEoverridecommandlockouts
\maketitle

\begin{abstract}
Intelligent reflecting surfaces (IRSs) enable multiple-input multiple-output (MIMO) transmitters to modify the communication channels between the transmitters and receivers. In the presence of eavesdropping terminals, this degree of freedom can be used to effectively suppress the information leakage towards such malicious terminals. This leads to significant potential secrecy gains in IRS-aided MIMO systems. This work exploits these gains via a tractable \textit{joint} design of downlink beamformers and IRS phase-shifts. In this respect, we consider a generic IRS-aided MIMO wiretap setting and invoke fractional programming and alternating optimization techniques to iteratively find the beamformers and phase-shifts that maximize the achievable weighted secrecy sum-rate. Our design concludes two low-complexity algorithms for joint beamforming and phase-shift tuning. Performance of the proposed algorithms are numerically evaluated and compared to the benchmark. The results reveal that integrating IRSs into MIMO systems not only boosts the secrecy performance of the system, but also improves the robustness against passive eavesdropping.
\end{abstract}


\section{Introduction}
Over the past few years, the unprecedented growth of data traffic due to the popularity of smart devices has created many challenges in the design of next-generation wireless networks \cite{saad2019vision,yaacoub2020key,letaief2019roadmap}. To address these challenges, various advanced technologies have been proposed in recent years. Examples of such technologies are massive \ac{mimo} systems \cite{larsson2014massive,hoydis2013massive,yang20196g}, \ac{mmwave} communications \cite{wang2018millimeter,rappaport2013millimeter,wang20206g}, and ultra-dense heterogeneous networks \cite{an2017achieving,niu2017fast,zhang2020envisioning}, just to name a few. Although employing these key technologies can significantly enhance the spectral efficiency of wireless networks, practical limits related to energy consumption, hardware cost, and transceiver complexity are still considered burdensome \cite{malkowsky2017world,huo2019enabling}. \ac{irs}-aided \ac{mimo} communications has recently attracted considerable attention as an effective solution to these implementational issues \cite{wu2020intelligent,pan2020multicell}.

Recent advances in the design of reflecting surfaces have introduced \acp{irs} as efficient components which can help overcoming various challenges of earlier enabling technologies \cite{ozdogan2019intelligent}. In fact, the ability of \acp{irs} in beamforming and reflecting the received signals with neither noise amplification nor self-interference make them different from basic relaying components, e.g., conventional repeaters, in wireless networks \cite{bjornson2019intelligentRelay}. \acp{irs} are flexibly installed on room ceiling, interior walls, and building facades \cite{hu2018beyond}. Hence, from the implementational viewpoint, \ac{irs}-aided architectures can significantly reduce implementation cost without considerable performance degradation.

A primary application of \acp{irs} in wireless communications is to employ them for realizing a massive \ac{mimo} transmitter in a distributed fashion \cite{bjornson2019massive,wu2019intelligent}: Instead of gathering a vast number of antennas at a single \ac{bs}, an \ac{irs} is placed out of the location of a \ac{bs} with a limited number of antennas. Despite conceptual differences to massive \ac{mimo},~such settings are shown to sustain the fundamental features of massive \ac{mimo} systems while enjoying low implementational complexity \cite{bjornson2020power,yue2020analysis}. \acp{irs} can further enable cost-efficient realizations of hybrid analog digital architectures \cite{jamali2019intelligent,bereyhi2019papr,bereyhi2020single,karasik2021single}. In such architectures, the analog network is implemented via a set of \ac{rf} transmitters illuminating the digitally precoded signals towards an \ac{irs}, and the analog beamforming is performed by proper phase-shifting at the \ac{irs}. As the physical analog network is replaced via a \textit{noise-free wireless channel}, these architectures significantly reduce the power-loss. Interestingly, this gain is obtained at no considerable degradation in the quality of downlink transmission \cite{jamali2019intelligent}. 

In addition to the mentioned applications, \acp{irs} potentially provide benefits in terms of secrecy \cite{dong2020enhancing,lu2020robust,hong2020artificial}, network coverage \cite{cao2019intelligent} and spectrum efficiency \cite{huang2019reconfigurable, liu2020reconfigurable}. A comprehensive overview of \ac{irs}-aided  wireless communications is provided in \cite{wu2019towards} where the authors discuss main applications, channel estimation, deployment challenges and hardware issues.

\subsection{Secrecy in IRS-aided MIMO Communications}
In an \ac{irs}-aided wireless network, the phase-shifts of the elements at the \acp{irs} can be tuned such that the reflected signals from the \acp{irs} will be
\begin{itemize}
	\item \textit{constructively} combined with the copy of the transmitted signal received through the direct links between the \ac{bs} and the \textit{legitimate} terminals, and
	\item \textit{destructively} added to the signals received by \textit{eavesdroppers} via their respective direct links.
\end{itemize}
This property enables \ac{irs}s to greatly enhance the physical layer security of \ac{mimo} settings. This potential however comes at the price of design complexity. In fact, to employ effectively an \ac{irs}, one needs to jointly perform beamforming and phase-shift tuning. This often lead to a computationally intractable optimization; see for example discussions in \cite{shen2019secrecy}. As a result, most primary studies consider the design of \ac{irs}-aided \ac{mimo} systems for simplified settings; see for instance \cite{shen2019secrecy, cui2019secure, wang2020intelligent}. 

A fundamental study on the secrecy performance of \ac{irs}-aided settings is given in \cite{shen2019secrecy}, where the authors considered a simple wiretap setting with a single legitimate receiver and a single eavesdropper. The investigations are further extended to a slightly different setting with a superior eavesdropping channel in \cite{cui2019secure}. The study in \cite{yang2020secrecy} investigates an \ac{irs}-aided  \ac{miso} transmission in a basic wiretap channel with a single-antenna eavesdropper. Using deep learning for phase-shift tuning and downlink beamforming is considered in \cite{song2020truly}, and the performance is compared with the conventional methods in \cite{shen2019secrecy} and \cite{cui2019secure}.

In addition to the passive beamforming gain achieved by \acp{irs}, the extra degrees of freedom provided by these surfaces lead to secrecy enhancements in various other respects. For instance, the studies in \cite{guan2020intelligent,xu2019resource,hong2020artificial} demonstrate that generating artificial noise and applying controlled jamming by the \ac{bs} in \ac{irs}-assisted  \ac{miso} settings result in significantly better secrecy performance than in classic settings without \acp{irs}. Another example is the robustness of \ac{irs}-assisted \ac{mimo} settings against \textit{active} eavesdropping reported recently in \cite{bereyhi2020secure}. In this study, an \ac{irs}-assisted wiretap setting with active eavesdroppers is considered that contaminate the uplink training pilots of the legitimate \acp{ut} to increase their received information leakage. The results of this study reveal that the further degrees of freedom achieved by the \ac{irs} enables us to statistically blind the active eavesdroppers. This is in contrast to classical massive \ac{mimo} settings that are known to be non-robust against active eavesdropping; see \cite{bereyhi2018robustness,bereyhi2019robustness,kapetanovic2015physical} and the references therein.

The basic studies on secrecy performance of \ac{irs}-aided \ac{mimo} settings are further extended in multiple directions; see for instance \cite{wang2020energy, chen2019intelligent, yu2020robust}. In \cite{wang2020energy}, the authors investigate the secrecy performance of a cooperative jamming strategy applied by multiple eavesdroppers employing an \ac{irs} to attack a single-antenna legitimate \ac{ut}. Secure downlink transmission in a multiuser \ac{irs}-aided \ac{miso} system with multiple legitimate and malicious \acp{ut} is further studied in \cite{chen2019intelligent}. Despite extending the basic investigations to multiuser settings, the work considers a special scenario in which all \acp{ut} are located in the same direction of the transmitter and hence their channels are highly correlated. The work is extended in \cite{yu2020robust} with respect to the channel model. The latter study is still restricted to cases with blocked direct channels from the transmitter to the \acp{irs} and eavesdroppers.

\subsection{Contributions}
As mentioned, the results in the literature imply that integrating \acp{irs} into \ac{mimo} systems can lead to significant enhancements in terms of secrecy performance. Although this finding is concluded from several recent lines of work, a design of secure active and passive beamforming in \ac{irs}-aided \ac{mimo} systems is yet to be addressed. The main goal of this work is to design a secure joint active, i.e., precoding at the transmitter, and passive beamforming, i.e., phase-shift tuning at \acp{irs}, by which the potential secrecy gains of \ac{irs}-aided settings are exploited. This goal is reached through the following contributions:
\begin{itemize}
	\item First the secrecy performance of an \ac{irs}-aided multiuser \ac{mimo} wiretap setting is described in terms of the achievable weighted secrecy sum-rate. The beamforming design is then formulated as a non-convex optimization problem. To cope with the non-convexity, we initially derive a parameterized variational problem whose solution returns the optimal design. We then invoke \ac{fp} \cite{shen2018FP,shen2018FP2} to transform the variational problem into a quadratic problem. \Ac{fp} has been widely used for throughput maximization in communication systems. The most relevant instance is the work in \cite{guo2020weighted}, in which the authors use \ac{fp} for weighted  sum-rate maximization in an \ac{irs}-aided \ac{mimo} system. Although the derivations in \cite{guo2020weighted} consider a generic \ac{irs}-aided \ac{mimo} system, extending them to a wiretap setting is not straightforward. We address this issue by dividing the variational optimization into multiple \textit{marginal} problems and by \textit{conditionally} bounding the objective of each marginal problem. The final solution is then found by following the \ac{ao} strategy.
	\item Using \ac{ao}, the final solution is approximated via an iterative algorithm which deals with multiple inner loops. This leads to high computational complexity. We hence develop two algorithms with reduced complexity; namely a \textit{two-tiers} and a \textit{single-loop} algorithm. For each of these algorithms, we show that the weighted secrecy sum-rate increases monotonically as the algorithm iterates.  
	\item To confirm the validity of our derivations, we investigate the performance of the proposed algorithms by running several numerical experiments. The results are compared to multiple reference scenarios, e.g., the scenarios with no \ac{irs} and random phase shift, as well as the state-of-the-art. Our investigations confirm that the proposed scheme exploits the potential secrecy gains of \ac{irs}-aided systems efficiently, while imposing a tractable processing load on to the system.
\end{itemize}

\subsection{Notation and Organization}
Scalars, vectors and matrices are represented with non-bold, bold lower-case, and bold upper-case letters, respectively. $\mH^{\her}$ indicates the transposed conjugate of $\mH$ and $\mI_N$ is an $N\times N$ identity matrix. The $\ell_2$-norm of $\bx$ is denoted by $\norm{\bx}$, and $\norm{\mW}_F$ represents the Frobenius norm of matrix $\mW$. $\log$ indicates logarithm to the base $2$. $\setR$ and $\setC$ refer to the real axis and the complex plane, respectively. For $z\in\setC$, $z^*$, $\Re\set{z}$, and $\Im\set{z}$ denote the complex conjugate, real part, and imaginary part of $z$, respectively. $\mathcal{CN}\brc{\eta,\sigma^2}$ represents the complex Gaussian distribution with mean $\eta$ and variance $\sigma^2$. For sake of brevity, $\set{1,\ldots,N}$ is shown by $\dbc{N}$.

The remaining parts of this manuscript are organized as follows: The system model is described in Section~\ref{Sec:SysMod}. In Section~\ref{Sec:Joint}, the design problem is formulated and a variational problem with analytic objective is derived. Section~\ref{Sec:FP} gives a quick introduction to \ac{fp} and the key analytical tools used in the paper. A class of alternative algorithms based on the \ac{ao} strategy is then derived in Section~\ref{sec:ALG_1}. Section~\ref{Sec:Red} proposes algorithms with reduced complexity and discusses their convergence. Numerical investigations are presented in Section~\ref{Sec:Num}. Finally, the manuscript is concluded in Section~\ref{Sec:Conc}.
 
\section{Problem Formulation}
\label{Sec:SysMod}
Secure downlink transmission in an \ac{irs}-aided broadcast setting is considered. For sake of brevity, we focus on a single cell in which a \ac{bs} with $M$ transmit antennas serves $K$ single-antenna legitimate users. $J$ single-antenna eavesdroppers passively overhear the channel. To improve the downlink communication links, a passive \ac{irs} unit with $N$ programmable phase shifters is deployed. The elements on this unit receive copies of the signal transmitted by the \ac{bs} and reflect them after applying phase shifts. These phase shifts are adjustable and controlled by a central control unit. A schematic diagram of the setting is represented in Fig.~\ref{fig:1}.

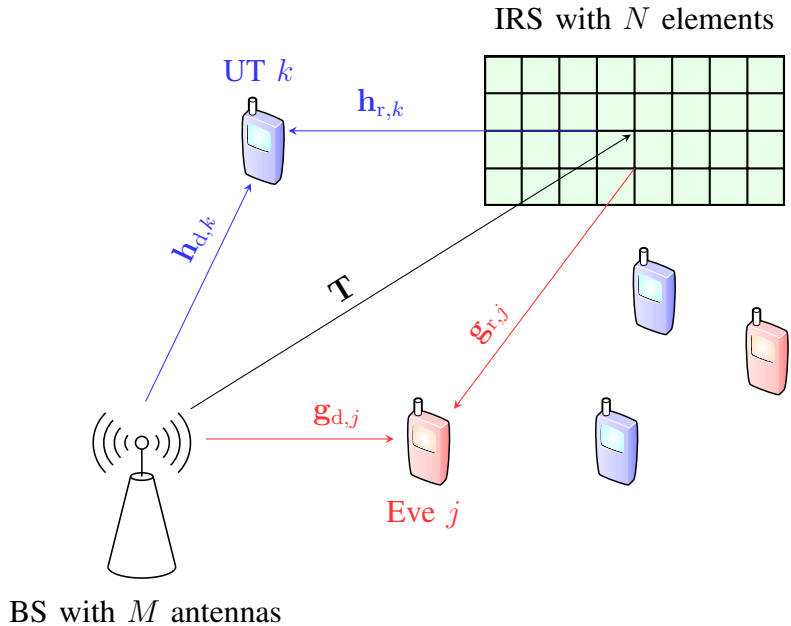
\begin{figure}
	\centering
	\begin{tikzpicture}
\tikzset{mobile phone/.pic={
		code={
			\begin{scope}[line join=round,looseness=0.25, line cap=round,scale=0.07, every node/.style={scale=0.07}]
				\begin{scope}
					\clip [preaction={left color=blue!10, right color=blue!30}] 
					(1/2,-1) to [bend left] (0,10)
					to [bend left] ++(1,1) -- ++(0,2)
					arc (180:0:3/4 and 1/2) -- ++(0,-2)
					to [bend left]  ++(5,-2) coordinate (A) to [bend left] ++(-1/2,-11)
					to [bend left] ++(-1,-1) to [bend left] cycle;
					\path [left color=blue!30, right color=blue!50]
					(A) to [bend left] ++(0,-11) to[bend left] ++(-3/2,-2)
					-- ++(0,12);
					\path [fill=blue!20, draw=white, line width=0.01cm]
					(0,10) to [bend left] ++(1,1) -- ++(0,2)
					arc (180:0:3/4 and 1/2) -- ++(0,-2)
					to [bend left]  (A) to [bend left] ++(-3/2,-5/4)
					to [bend right] cycle;
					\draw [line width=0.01cm, fill=white]
					(9/8,21/2) arc (180:360:5/8 and 3/8) --
					++(0,2.5) arc (0:180:5/8 and 3/8) -- cycle;
					\draw [line width=0.01cm, fill=white]
					(9/8,13) arc (180:360:5/8 and 3/8);
					\fill [white, shift=(225:0.5)] 
					(1,17/2) to [bend left] ++(4,-7/4)
					to [bend left] ++(0,-7/2) to [bend left] ++(-4, 6/4)
					to [bend left] cycle;
					\fill [black, shift=(225:0.25)] 
					(1,17/2) to [bend left] ++(4,-7/4)
					to [bend left] ++(0,-7/2) to [bend left] ++(-4, 6/4)
					to [bend left] cycle;
					\shade [inner color=white, outer color=cyan!20] 
					(1,17/2) to [bend left] ++(4,-7/4)
					to [bend left] ++(0,-7/2) to [bend left] ++(-4, 6/4)
					to [bend left] cycle;
					%
				\end{scope}
				\draw [line width=0.02cm] 
				(1/2,-1) to [bend left] (0,10)
				to [bend left] ++(1,1) -- ++(0,2)
				arc (180:0:3/4 and 1/2) -- ++(0,-2)
				to [bend left]  ++(5,-2) to [bend left] ++(-1/2,-11)
				to [bend left] ++(-1,-1) to [bend left] cycle;
			\end{scope}%
}}}

\tikzset{eve/.pic={
		code={
			\begin{scope}[line join=round,looseness=0.25, line cap=round,scale=0.07, every node/.style={scale=0.07}]
				\begin{scope}
					\clip [preaction={left color=red!10, right color=red!30}] 
					(1/2,-1) to [bend left] (0,10)
					to [bend left] ++(1,1) -- ++(0,2)
					arc (180:0:3/4 and 1/2) -- ++(0,-2)
					to [bend left]  ++(5,-2) coordinate (A) to [bend left] ++(-1/2,-11)
					to [bend left] ++(-1,-1) to [bend left] cycle;
					\path [left color=red!30, right color=red!50]
					(A) to [bend left] ++(0,-11) to[bend left] ++(-3/2,-2)
					-- ++(0,12);
					\path [fill=red!20, draw=white, line width=0.01cm]
					(0,10) to [bend left] ++(1,1) -- ++(0,2)
					arc (180:0:3/4 and 1/2) -- ++(0,-2)
					to [bend left]  (A) to [bend left] ++(-3/2,-5/4)
					to [bend right] cycle;
					\draw [line width=0.01cm, fill=white]
					(9/8,21/2) arc (180:360:5/8 and 3/8) --
					++(0,2.5) arc (0:180:5/8 and 3/8) -- cycle;
					\draw [line width=0.01cm, fill=white]
					(9/8,13) arc (180:360:5/8 and 3/8);
					\fill [white, shift=(225:0.5)] 
					(1,17/2) to [bend left] ++(4,-7/4)
					to [bend left] ++(0,-7/2) to [bend left] ++(-4, 6/4)
					to [bend left] cycle;
					\fill [black, shift=(225:0.25)] 
					(1,17/2) to [bend left] ++(4,-7/4)
					to [bend left] ++(0,-7/2) to [bend left] ++(-4, 6/4)
					to [bend left] cycle;
					\shade [inner color=white, outer color=orange!20] 
					(1,17/2) to [bend left] ++(4,-7/4)
					to [bend left] ++(0,-7/2) to [bend left] ++(-4, 6/4)
					to [bend left] cycle;
					%
				\end{scope}
				\draw [line width=0.02cm] 
				(1/2,-1) to [bend left] (0,10)
				to [bend left] ++(1,1) -- ++(0,2)
				arc (180:0:3/4 and 1/2) -- ++(0,-2)
				to [bend left]  ++(5,-2) to [bend left] ++(-1/2,-11)
				to [bend left] ++(-1,-1) to [bend left] cycle;
			\end{scope}%
}}}

\tikzset{radiation/.style={{decorate,decoration={expanding waves,angle=90,segment length=4pt}}},
	antenna/.pic={
		code={\tikzset{scale=3/10}
			\draw[semithick] (0,0) -- (1,4);
			\draw[semithick] (3,0) -- (2,4);
			\draw[semithick] (0,0) arc (180:0:1.5 and -0.5);
			\node[inner sep=4pt] (circ) at (1.5,5.5) {};
			\draw[semithick] (1.5,5.5) circle(8pt);
			\draw[semithick] (1.5,5.5cm-8pt) -- (1.5,4);
			\draw[semithick] (1.5,4) ellipse (0.5 and 0.166);
			\draw[semithick,radiation,decoration={angle=45}] (1.5cm+8pt,5.5) -- +(0:2);
			\draw[semithick,radiation,decoration={angle=45}] (1.5cm-8pt,5.5) -- +(180:2);
	}}
}

\tikzset{
	irs/.pic={\clip[postaction={shade,left color=green!10,right color = green!10}](0,0) rectangle (4,2);
		\draw[thick] (0,0) grid[step=0.5] (4,2);
		\draw[ultra thick](0,0) rectangle (4,2);}
}

	\path (0,-.8) pic {antenna};
	\node at (.5,-1.4) (UTK) {BS with $M$ antennas};
	
	\path (4,.5) pic {eve};
	\node[red!80] at (4.2,-.1) (EveJ) {Eve $j$};
	
	\path (8.5,1.7) pic {eve};
	
	\path (5,4) pic {irs};
	\node at (7,6.5) (UTK) {IRS with $N$ elements};
	
	\node[blue!80] at (2,5.8) (UTK) {UT $k$};
	
	\path (1.8,4.5) pic {mobile phone};
	\path (7,2.5) pic {mobile phone};
	\path (6.5,.5) pic {mobile phone};
	
	\draw [->,>=stealth] (1.1,1.3) -- (6.95,4.95) node [above, sloped,pos=.37] (d) {$\mT$};
	\draw [->,>=stealth,blue!80] (6.5,5) -- (2.4,5) node [above, sloped,pos=.7] (d) {$\mh_{\rmr, k}$};
	\draw [->,>=stealth,red!80] (1.3,.9) -- (3.8,.9) node [above, sloped,pos=.7] (d) {$\bgg_{\rd, j}$};
	\draw [->,>=stealth,red!80] (7,4.5) -- (4.6,1.3) node [above, sloped,pos=.7] (d) {$\bgg_{\rmr, j}$};
	\draw [->,>=stealth,blue!80] (.5,1.4) -- (1.9,4.3) node [above, sloped,pos=.7] (d) {$\mh_{\rd, k}$};
%

\end{tikzpicture}
	\caption{A schematic representation of the system model. In this diagram, the blue and red UTs denote the legitimate users and eavesdroppers, respectively. In this model, both the direct and reflection paths are available between the BS and a UT.}
		\label{fig:1}
\end{figure}

We consider transmission over a quasi-static slow fading channel which models either a narrow-band single-carrier system or a particular sub-channel of a wide-band multi-carrier system.
The system operates in the \ac{tdd} mode. This means that the uplink and downlink transmissions are performed at the same carrier frequency. As the result, the uplink and downlink channels are \textit{reciprocal}. Following this property, the \ac{bs} acquires the \ac{csi} in the uplink training phase and directly applies it to the downlink signal transmission. 

For channel estimation, we assume that the transmitter uses a classic channel estimation algorithm for an \ac{irs}-aided \ac{mimo} system, e.g., \cite{wang2019channel,he2019cascaded}.  It is further assumed that the eavesdroppers are registered users in the system. This means that during the uplink training phase, the \ac{bs} also acquires the \ac{csi} of the eavesdroppers. To keep the analysis tractable, we further ignore the channel estimation error and accordingly assume that the \ac{csi} of the receive terminals are perfectly available at the \ac{bs}.

\begin{remark}
It is worth mentioning that in this work we focus on \textit{passive} eavesdropping. Such a setting models broadcast scenarios with confidential messages in which the eavesdroppers are registered \acp{ut}. These \acp{ut} are supposed to receive the \textit{common} messages broadcasted in the network, but should receive no information regarding the \textit{confidential} messages; see for example studies in \cite{Liu2010multiple,Liu2013new} and references therein. The act of eavesdropping in such networks is conceptually different from overhearing the \textit{common} message. In the latter case, eavesdroppers are \textit{not} registered users and are purely employed to eavesdrop upon a message which is supposed to be received only by the registered users. In such a scenario, the eavesdroppers either do not participate in the uplink training phase\footnote{Hence, their \ac{csi} is not available at the \ac{bs}.} or perform an \textit{active} attack\footnote{And hence, they contaminate the acquired \ac{csi}.}, e.g., the active pilot attack; see for instance the system model in \cite{bereyhi2020secure}.
\end{remark}

\subsection{System Model}
To model the downlink transmission in this setting, we note that each \ac{ut} receives a superposition of two signals: 
\begin{enumerate}
\item A signal which is received through the direct path\footnote{Note that the direct path is not necessarily the \textit{line-of-sight}. It merely refers to the channel between the \ac{bs} and a \ac{ut} which may include a line-of-sight path and/or scatterings.} between the \ac{bs} and the \ac{ut}, and
\item a signal which is reflected via the \ac{irs}.
\end{enumerate}
Let $\bxx \in \setC^M $ contain the signal samples being transmitted in a particular symbol interval via the transmit antennas at the \ac{bs}. The received signal at legitimate  \ac{ut} $k\in \dbc{K}$ is hence given by
\begin{align}
	\yy_k=\mh_{\rd, k}^\her \bxx + \mh_{\rmr, k}^\her \mPhi^\her \brr + \vartheta_k. \label{eq:y_k}
\end{align}
In \eqref{eq:y_k}, $\mPhi$, $\brr$, $\vartheta_k$, $\mh_{\rd, k}$ and $\mh_{\rmr, k}$ are defined as follows:
\begin{itemize}
	\item $\mPhi\in \setC^{N\times N}$ is a diagonal matrix modeling the phase-shifts applied by the elements on the \ac{irs}-unit,
	\item $\brr$ is the signal received by the \ac{irs}-unit from the \ac{bs},
	\item $\vartheta_k$ is complex Gaussian noise with zero mean and variance $\sigma_{k}^2$, i.e., $\vartheta_k\sim \mathcal{CN}\brc{0, \sigma_{k}^2}$,
	\item $\mh_{\rd, k} \in \setC^{M}$ represents the conjugate of direct uplink channel between \ac{ut} $k$ and the \ac{bs}\footnote{Representing the uplink channel vectors via conjugated vectors is for sake of simplicity.},
	\item $\mh_{\rmr, k} \in \setC^{N}$ denotes the conjugate of uplink channel between \ac{ut} $k$ and the \ac{irs}-unit.
\end{itemize}
Following the fact that the reflecting elements are passive, we can write $\mPhi= \diag\set{\bphi} $, for some $\bphi=\dbc{\phi_1, \cdots, \phi_N}^\trp$ where $\phi_n$ is of the form
\begin{align}
	\phi_n=\beta_n e^{-\rmj \theta_n}
\end{align}
for some $\beta_n \in \dbc{ 0, 1 }$ and $\theta_n\in \left[0, 2\pi \right)$ denoting the attenuation coefficient and the phase shift applied by the $n$-th reflecting element, respectively. We assume that the attenuation coefficients of all \ac{irs} elements are identical. In other words, we consider a case in which $\beta_n$ for $n\in\dbc{N}$ are not \textit{tunable}\footnote{This means that the control unit which controls the \ac{irs} does not update the attenuation coefficients.}. Such an assumption follows from the following two facts: 
\begin{enumerate}
	\item Current technology suggests that implementationally efficient \acp{irs} consist of elements whose attenuation characteristics are not tunable \cite{PhysRevB.94.075142}. As the result, practical designs for \ac{irs}-assisted communications ignore this degree of freedom\footnote{Although it is theoretically possible to be considered.}.
	\item Due to the long distance between the \ac{bs} and \ac{irs}, the so-called \textit{tapering} effect is negligible in our setting and can be included in the path-loss model: The \textit{tapering} effect refers to the non-uniform distribution of attenuation coefficients on the elements of an \ac{irs}, due to their different distances and angles of arrival to the transmit array antenna. This effect is tangible in settings in which the \ac{irs} is a part of the transmitter architecture, e.g., the scenarios in \cite{jamali2019intelligent,bereyhi2019papr,bereyhi2020single}. This is not the case in our setting, as we assume that the \ac{irs} is a passive element located in a \textit{long} distance out of the \ac{bs} site. As the result, the variation of the attenuation coefficients from one \ac{irs} element to another is insignificant. This small tapering effect can be ignored at the \ac{irs} and be incorporated in the channel path-loss.
\end{enumerate}
Following the above discussions, we set $\beta_n=1$ for $ n \in \dbc{N}$.

The received signal $\brr$ is further given in terms of $\bxx$ as
\begin{align}
	\brr = \mT^\her \bxx,
\end{align}
where $\mT^*\in \setC^{M\times N}$ contains the coefficients of the uplink channels from the \ac{irs} elements to the \ac{bs}. Consequently, the input-output relation in \eqref{eq:y_k} can be represented as
\begin{align}
	\yy_k = \tilde{\mh}_k^\her \brc{\bphi} \bxx + \vartheta_k,
\end{align}
\begin{subequations}
where $\tilde{\mh}_k \brc{\bphi} \in\setC^{M}$ describes the \textit{effective end-to-end uplink channel} between the $k$-th legitimate \ac{ut} and the \ac{bs}, and is given by
\begin{align}
	\tilde{\mh}_k \brc{\bphi} &= \mh_{\rd, k}+\mT  \mPhi \mh_{\rmr, k}\\
	&=\mh_{\rd, k}+\mT \;  \diag\{\mh_{\rmr, k}\} \; \bphi\\
	&=\mh_{\rd, k}+\mH_k \bphi \label{eq:htrnsf}
\end{align}
\end{subequations}
with $\mH_k=\mT \;  \diag\{\mh_{\rmr, k}\}$. $\mH_k$ can be regarded as the cascaded effective uplink channel between \ac{ut} $k$ and the \ac{bs} through the \ac{irs}. We further define the following notation:
\begin{align}
	\tilde{\mH} \brc{\bphi} &=   \dbc{\tilde{\mh}_1 \brc{\bphi}, \cdots, \tilde{\mh}_K \brc{\bphi} }
\end{align}
and refer to it as the \textit{legitimate channel matrix}.

With a similar approach,  the signal received by eavesdropper $j$ is written as
\begin{align}
	\zz_j=\tilde{\bgg}_j^\her \brc{\bphi} \bxx +\xi_j
\end{align}
	\begin{subequations}
where $\xi_j \sim \mathcal{CN}\brc{0, \mu_j^2}$ is Gaussian noise, and $\tilde{\bgg}_j\brc{\bphi} \in\setC^{M}$ describes the \textit{effective end-to-end uplink channel} from eavesdropper $j$ to the \ac{bs} being given by
\begin{align}
	\tilde{\bgg}_j \brc{\bphi} &=\bgg_{\rd, j}+\mT \mPhi \bgg_{\rmr, j} \\
	&=\bgg_{\rd, j}+\mT \; \diag\{\bgg_{\rmr, j}\} \; \bphi \\
	&=\bgg_{\rd, j}+\mG_j \bphi	\label{eq:gtrnsf}.
\end{align}
\end{subequations}
In \eqref{eq:gtrnsf}, $\bgg_{\rd, j}$ and $\bgg_{\rmr, j}$ are the conjugates of the channels corresponding to the path from eavesdropper $j$ to the \ac{bs} and the \ac{irs}-unit, respectively.  $\mG_j$ further represents the effective uplink channel between the $j$-th eavesdropper and the \ac{bs} through the \ac{irs} and is defined as 
\begin{align}
	\mG_j=\mT \; \diag\set{\bgg_{\rmr, j}}
\end{align}

For sake of brevity, we further define the notations
\begin{subequations}
	\begin{align}
		\tilde{\mG} \brc{\bphi} &=   \dbc{\tilde{\bgg}_1 \brc{\bphi}, \cdots, \tilde{\bgg}_J \brc{\bphi} },\\
		\mathbf{\Upsilon}_{\ee}    &=   \diag \set{ \frac{1}{\mu_1}, \cdots, \frac{1}{\mu_J} }
	\end{align}
\end{subequations}
and refer to them as the \textit{eavesdropping channel} and \textit{precision} matrix, respectively.

\subsection{Linear Precoding at the BS}
The \ac{bs} constructs the transmit signal from the encoded information symbols of the legitimate \acp{ut}, i.e,  $s_1, \cdots, s_K$, via linear precoding. Let $\bw_k$ denote the precoding vector for \ac{ut} $k$. The transmit signal in this case is given by
 \begin{align}
	\bxx = \sum_{k=1}^K  s_k \bw_k.
\end{align}
We assume that the encoded information symbols are \ac{iid} standard complex Gaussian random variables, i.e., $\bss\sim \mathcal{CN}\brc{\boldsymbol{0},\mI_K}$, where $\bss=\dbc{s_1, \cdots, s_K}$. 

In order to restrict the transmit power to the maximum allowed power $P_\maxx$, precoding vectors are scaled, such that 
\begin{align}
\Ex{ \norm{\bxx}^2}\leq P_{\maxx}.	
\end{align}
Defining $\mW=\dbc{\bw_1, \cdots, \bw_K}$ as the \textit{precoding matrix}, the transmit signal is compactly written as $\bxx=\mW\bss$. The power constraint can then be written as 
\begin{align}
	\Ex{ \norm{\mW\bss}^2}\leq P_{\maxx}.
\end{align}
Noting that $\Ex{\bss\bss^\her}=\mI_K$, the power constraint reduces to
\begin{align}
 \sum_{k=1}^K \Vert \bw_k \Vert^2 	\leq P_{\maxx}
\end{align}
or equivalently $\norm{\mW}_F^2 \leq P_\maxx$.

\subsection{Secrecy Performance Metric}
The downlink transmission is considered to be \textit{secure}, if the \ac{bs} reliably transmits the encoded symbols to the legitimate \acp{ut} without allowing any information leakage to the malicious terminals. A rate tuple $\brc{R_1,\ldots,R_K}$ at which secure transmission is guaranteed is called an \textit{achievable} tuple of secrecy rates, and the convex hull of all achievable tuples is defined as the \textit{secrecy capacity region}. 

In general, the capacity region is not a tractable performance metric for a practical system design. An alternative metric is given by considering a classical inner bound of the secrecy capacity region: A rate tuple $\brc{R_1,\ldots,R_K}$ is achievable with precoding matrix $\mW$ and vector of phase-shifts $\bphi$, if for $k\in \dbc{K}$ we have $R_k \leq \mar_k^\rms \brc{\mW, \bphi}$ \cite{oggier2011secrecy}. The lower bound $\mar_k^\rms \brc{\mW, \bphi}$ is given by
\begin{align}
	\mar_k^\rms \brc{\mW, \bphi}=\dbc{\mar_k^\rmm \brc{\mW, \bphi} -\mar_k^\ee  \brc{\mW, \bphi} }^+ \label{eq:secrecy}
\end{align}
where $\mar_k^\rmm \brc{\mW, \bphi}$ and $\mar_k^\ee  \brc{\mW, \bphi}$ are defined as follows:
\begin{itemize}
	\item $\mar_k^\rmm \brc{\mW, \bphi}$ is a lower bound on the maximum achievable rate to legitimate \ac{ut} $k$ and is given by
	\begin{align}
		\mar_k^\rmm\brc{\mW, \bphi}=\log\brc{1+ \sinr_k \brc{\mW, \bphi} } \label{eq:rate_leg},
	\end{align}
	where $ \sinr_k\brc{\mW, \bphi}$ is the \ac{sinr} at \ac{ut} $k$ and is determined as
	\begin{align} \label{eq:sinr}
		\sinr_k\brc{\mW, \bphi}= \frac{|\tilde{\mh}_k^\her (\bphi) \bw_k|^2}{\displaystyle \sum_{i=1, i \neq k}^{K}|\tilde{\mh}_k^\her (\bphi) \bw_i|^2+ \sigma_k^2}.
	\end{align}
	\item $\mar_k^\ee \brc{\mW, \bphi}$ is an upper bound on the maximum information leakage to the eavesdroppers and is expressed as 
	\begin{align}
		\mar_k^\ee \brc{\mW, \bphi} =\log \brc{ 1+ \esnr_k \brc{\mW, \bphi} } \label{eq:rate_eav}
	\end{align} 
	with $\esnr_k \brc{\mW, \bphi}$ being defined as
	\begin{align}\label{eq:esnr}
		\esnr_k\brc{\mW, \bphi}=\norm{\mathbf{\Upsilon}_{\ee}\mathbf{\tilde{G}^\her}(\bphi)\bw_k}^2.
	\end{align}
\end{itemize}
By plugging \eqref{eq:rate_leg} and \eqref{eq:rate_eav} into \eqref{eq:secrecy}, $\mar_k^\rms \brc{\mW, \bphi}$ is compactly written as
\begin{align}
	\mar_k^\rms\brc{ \mW, \bphi } =\left[ \log\left( \dfrac{1+\sinr_k\brc{\mW, \bphi}}{1+\esnr_k\brc{\mW, \bphi} } \right) \right]^+. \label{eq:sec_rate}
\end{align}

The above inner bound is derived by considering the \textit{worst-case} scenario in which all eavesdroppers cooperate to overhear the secure transmission and are able to cancel out the interference of other legitimate \acp{ut}. It is hence clear that larger secrecy rates are also achievable in a practical setting, since such worst-case assumptions are not necessarily fulfilled.

In the sequel, we utilize this inner bound to quantify the secrecy throughput of the system. To this end, we define the \textit{weighted secrecy sum-rate} as
\begin{align}
	\mar^\ssr\brc{\mW, \bphi}=\sum_{k=1}^K \omega_k \mar_k^\rms \brc{\mW, \bphi } \label{eq:ssr}
\end{align}
for some non-negative real weights $\omega_1,\ldots,\omega_K$ corresponding to the \acp{qos} desired for the \acp{ut}.

\section{Jointly Optimal Precoding and Phase-Shifting}
\label{Sec:Joint}
Considering the problem formulation, the design parameters in this system are the precoding vectors at the \ac{bs} and phase-shifts applied by the \ac{irs}-unit. As a result, the ultimate goal of this work is to design a joint precoding and phase-shifting algorithm that maximizes the weighted secrecy sum-rate.

With respect to the given performance metric, the optimal design for the precoding matrix and phase shifts is given via the following constrained optimization $\maP_1$:
\begin{subequations}
	\label{eq:ssrm}
	\begin{align}
    & \max _{\mW, \bphi}  \left.  \mar^\ssr\brc{\mW, \bphi} \tag{$\maP_1$} \right.  \\
	&\text{subject to }  \norm{\mW}_F^2 \leq P_{\maxx}\label{eq:constraint10}\\ 
	&\phantom{\text{subject to } }     \abs{\phi_n}=1,~\forall n \in \dbc{ N }.
	\label{eq:constraint20}
	\end{align}
\end{subequations}
Here, \eqref{eq:constraint10} ensures that the maximum transmit power is kept below the \ac{bs} power budget $P_{\maxx}$, and \eqref{eq:constraint20} restricts the unit modulus constraints imposed by the physical characteristics of the \ac{irs} elements.

The optimization problem $\maP _1$ contains two analytical challenges:
\begin{itemize}
	\item[(a)] The objective function is not differentiable with respect to the optimization variables. This follows from the fact that the non-negative operator, i.e., $f\brc{x} = \dbc{x}^+$,~is~not~differentiable. 
	\item[(b)] Even by replacing the non-negative operators with their arguments in $\mar^\ssr\brc{\mW, \bphi}$, $\maP_1$ is still \textit{non-convex}. This is due to non-convexity of the objective function and the unit modulus constraints in \eqref{eq:constraint20}. 
\end{itemize}
The above challenges indicate that the global optimum of $\maP_1$ can not be computed via classical convex programming techniques. In fact, it requires a brute-force search that is not a computationally tractable option even for not-so-large dimensions. 

We address these issues as follows: First, a variational optimization problem is derived whose objective is differentiable and whose solution coincides with the solution of $\maP_1$. We then develop a two-tiers algorithm that uses \ac{bcd} and \ac{fp} to approximate the solution of the variational problem.

\subsection*{The Variational Problem}
Theorem~\ref{theorem:1} gives an alternative optimization problem with~a differentiable objective function whose solution lies on~the~solution of $\maP_1$.
\begin{theorem}[The variational problem]
	\label{theorem:1}
Let $\mar^\star$ be the maximum of the optimization problem $\maP_1$. Then, we have
	\begin{subequations}
		\begin{align}
		\mar^\star = &\max _{\mW, \bphi, \mb} \left.  \mar^\ssr_\rmq \brc{ \mW, \bphi, \mb} \right. \tag{$\maP_2$} \label{eq:ssrmq}\\
		&\text{\normalfont subject to } \norm{\mW}_F^2 \leq P_{\maxx} \label{eq:constraint12}\\ 
		&\phantom{\text{subject to }}      \abs{ \phi_n }=1,~\forall n \in \dbc{N },
		\label{eq:constraint2}\\
		&\phantom{\text{subject to }}    \mb \in \dbc{0,1}^K \label{eq:constraint3}
		\end{align}
	\end{subequations}
	where the objective function $\mar^\ssr_\rmq  \brc{\mW, \bphi, \mb} $ is given by
	\begin{align}
\hspace*{-1mm}	\mar^\ssr_\rmq  \brc{\mW, \bphi, \mb} \hspace*{-1mm}  = \hspace*{-1mm}  \sum_{k=1}^K \omega_k b_k \log\left( \dfrac{1+\sinr_k\brc{\mW, \bphi}}{1+\esnr_k\brc{ \mW, \bphi } } \right) \label{R_q}
	\end{align}
with $b_k$ denoting the $k$-th entry of $\mb$.
\end{theorem}
\begin{proof}
The equivalency of $\maP_1$ and~$\maP_2$~is~justified by showing that the maximum of $\maP_1$ bounds the maximum of $\maP_2$ from below and above. The simultaneous validity of both bounds hence concludes the equivalence of the maxima. The detailed proof is given in Appendix~\ref{app:The1}.
\end{proof}

Theorem~\ref{theorem:1} addresses the non-differentiability issue of the objective function in $\maP_1$. Nevertheless, the alternative form in $\maP_2$ is still non-convex. We address this latter issue in the sequel by developing an iterative algorithm which approximates the solution of $\maP_2$.


\section{Analytical Tools From Fractional Programming}
\label{Sec:FP}
\ac{fp} is the key analytical tool used in this work. We hence give a quick overview on the main concepts in \ac{fp} in this section. To this end, we introduce the \textit{Lagrangian dual transform} for a sum of logarithmic ratios problem and the \textit{quadratic transform} for a multiple-ratio problem  \cite{shen2018FP,shen2018FP2}. These are the key results which help us tackling the optimization $\maP_2$. 


\subsection{Lagrangian Dual and Quadratic Transforms}
The generic form of a \textit{sum of logarithmic ratios} is given by
\begin{align}
	S\brc{\bxx} = \sum_{k=1}^K \omega_k \log \left( 1+ \dfrac{A_k\brc{\bxx} }{B_k \brc{\bxx} }\right).
	\label{eq:SLR}
\end{align}
where $\omega_k$ is a non-negative real, $A_k\brc{\cdot}: \setC^N\mapsto \setR^+_0$ is a non-negative function, and $B_k\brc{\cdot}: \setC^N\mapsto \setR^+$ is a strictly positive function, for $k\in\dbc{K}$.

Consider the following optimization problem 
\begin{subequations}
	\label{eq:Lag1}
	\begin{align}
		&\max_{\bxx} \left. S\brc{\bxx} \right. \\
		&\text{subject to }  \bxx \in \setX 
	\end{align}
	for some $\setX\subseteq \setC^N$ which is non-empty. The Lagrangian dual transform gives an equivalent optimization problem whose both maximum and maximizer recover those given by \eqref{eq:Lag1}~\cite{shen2018FP2}.
\end{subequations}

\begin{definition}[Lagrangian dual objective]
	\label{def:Lag}
	For auxiliary vector $\btt = \dbc{t_1, \ldots,t_K} \in \setR_0^{+K}$, the Lagrangian dual objective of the sum of logarithmic ratios $S\brc{\bxx}$ is given by
	\begin{align}
		L\brc{\bxx,\btt} = \sum_{k=1}^K  \omega_k \left(
		\Xi \brc{t_k} + \dfrac{ \brc{1+t_k} A_k \brc{\bxx} }{ A_k \brc{\bxx} + B_k  \brc{\bxx}  } 
		\right)
	\end{align}
	where the function $\Xi\brc{\cdot}: \setR_0^+ \mapsto \setR$ is defined as
	\begin{align}
		\Xi\brc{x} = \log \brc{1+x} - x.
	\end{align}
\end{definition}

In \cite[Theorem 3]{shen2018FP2}, it is shown that the maximum value of the objective function in \eqref{eq:Lag1}, as well as the point at which the objective is maximized, are given by the following equivalent optimization problem:
\begin{subequations}
	\label{eq:Lag2}
	\begin{align}
		&  \max_{\bxx , \btt} \left.  L\brc{\bxx,\btt}  \right. \\
		&\text{ subject to }  \bxx\in \setX \text{ and }  \btt\in \setR_0^{+K}.
	\end{align}
\end{subequations}
We refer to this transformed version, as the Lagrangian dual transform of the optimization problem in \eqref{eq:Lag1}.

Similar to the Lagrangian dual transform, the quadratic transform provides an equivalent optimization problem, when a \textit{multiple-ratio} function is to be maximized \cite{shen2018FP}. To illustrate this transform, consider a multiple-ratio function 
\begin{align}
	M \brc{\bxx} = \sum_{k=1}^K   \dfrac{\abs{C_k\brc{\bxx}}^2 }{D_k \brc{\bxx} }, 
	\label{eq:SR}
\end{align}
in which $C_k\brc{\cdot}: \setC^N\mapsto \setC$, and $D_k\brc{\cdot}: \setC^N\mapsto \setR^+$ is a strictly positive function, for $k\in\dbc{K}$. We now focus on the following target optimization problem 
\begin{subequations}
	\label{eq:Lag0}
	\begin{align}
		&\max_{\bxx} \left. M\brc{\bxx} \right. \\
		&\text{subject to }  \bxx \in \setX 
	\end{align}
	for some non-empty $\setX\subseteq \setC^N$. An equivalent optimization problem in this case is given by the quadratic transformed~\cite{shen2018FP}.
\end{subequations}

\begin{definition}[Quadratic equivalent objective]
	\label{def:Quad}
	For auxiliary vector $\mbeta = \dbc{\beta_1, \ldots,\beta_K} \in \setC^{K}$, the quadratic equivalent~objective of the multiple-ratio function $M\brc{\bxx}$ is given by
	\begin{align}
		Q\brc{\bxx,\mbeta} = \sum_{k=1}^K  2 \left. \Re \set{ \beta_k^* C_k\brc\bxx }\right. - \abs{\beta_k}^2 D_k \brc{\bxx}.
	\end{align}
\end{definition}

Theorem~2 in \cite{shen2018FP} states that the solution of \eqref{eq:Lag0} is given by solving the following equivalent optimization problem:
\begin{subequations}
	\label{eq:Quad2}
	\begin{align}
		&  \max_{\bxx , \mbeta} \left.  Q\brc{\bxx,\mbeta}  \right. \\
		&\text{ subject to }  \bxx\in \setX \text{ and }  \mbeta \in \setC^{K}.
	\end{align}
\end{subequations}
We refer to \eqref{eq:Quad2} as the quadratic transform of \eqref{eq:Lag0}. 

\section{Developing an Iterative Algorithm}
\label{sec:ALG_1}
Starting from $\maP_2$, we now develop an iterative algorithm which approximates the maximum weighted secrecy sum-rate tractably. The algorithm consists of two tiers: In the first tier, the \ac{bcd} technique is used to approximate the solution of the joint optimization problem $\maP_2$ by a cyclic alternation among multiple marginal optimization problems. In the second tier, we use \ac{fp} and the \ac{mm} algorithm \cite{sun2017majorization} to address each of these marginal sub-problems.

We start the derivations with the first tier: Considering the optimization problem $\maP_2$, we group the optimization variables into three blocks; namely, $\mW$, $\bphi$, and $\mb$. Noting that these blocks are not coupled via the constraints of $\maP_2$, we use a \ac{bcd}-type algorithm to approximate the maxima by cyclically alternating among the three marginal problems \cite{bertsekas1998nonlinear}. In each of these problems, a marginal optimization with respect to one of the blocks is performed while the other two blocks are treated as fixed variables. The alternation among these marginal problems ends when all the blocks converge. 

Despite the complexity reduction achieved via the \ac{ao} technique in the first tier, there exists still a challenge: Two marginal problems deal with a \textit{non-convex} optimizations. We address this issue in the second tier, where we use the \ac{fp} and the \ac{mm} algorithm to approximate the solution of each marginal problem tractably. 

\subsection{First Marginal Problem} 
\label{sec:First_Mar}
	\begin{subequations}
The first marginal optimization finds the precoding matrix that maximizes the objective $\mar^\ssr_\rmq \brc{ \mW, \bphi, \mb}$ while treating $\bphi$ and $\mb$ as fixed variables in $\maP_2$, i.e., it finds $\mW^\star$ as
	\begin{align}
		 \mW^\star = &\argmax _{\mW}  \left. \mar^\ssr_\rmq \brc{ \mW, \bphi_0, \mb_0}\right. 
		  \tag{$\mam_1$}  \label{eq:marg1}\\
		&\text{\normalfont subject to } \norm{\mW}_F^2 \leq P_{\maxx},
	\end{align}
\end{subequations}
for some fixed $\bphi_0$ and $\mb_0$.

It is straightforwardly seen from \eqref{R_q} that the marginal problem $\mam_1$ is non-convex. Considering the logarithmic form of the objective function, we use the Lagrangian dual transform to tackle this problem. To this end, we note that the objective function can be decomposed as
\begin{align}
	 \mar^\ssr_\rmq \brc{ \mW, \bphi_0, \mb_0} = S_{1}^{ \rm m} \brc{\mW} + S_{1}^{ \rm e} \brc{\mW}
\end{align}
where the functions $S_1^{ \rm m} \brc{\mW}$ and  $S_1^{ \rm e} \brc{\mW}$ are sums of logarithmic ratios and are defined as
\begin{subequations}
\begin{align}
\label{eq:fell}
	S^{ \rm m}_1 \brc{\mW} &= \sum_{k=1}^{K} \omega_k b_{0k} \log \left( 1+\frac{ A_{1k}^{\rm m} \brc{\mW} }{ B_{1k}^{\rm m} \brc{\mW}  }\right) \\
	S_1^{ \rm e} \brc{\mW} &= 
	\sum_{k=1}^{K} \omega_k b_{0k} \log \left( \frac{1}{  B_{1k}^{\rm e} \brc{ \mW } } \right) .
\end{align}
\end{subequations}
with $	A_{1k}^{\rm m} \brc{\mW} $, $	B_{1k}^{\rm m} \brc{\mW} $ and $	B_{1k}^{\rm e} \brc{\mW} $ being
\begin{subequations}
	\begin{align}
	A_{1k}^{\rm m} \brc{\mW} &= \abs{\tilde{\mh}^\her_k \brc{\bphi_0} \bw_k}^2, \\
	B_{1k}^{\rm m} \brc{\mW} &= \sum_{i=1, i \neq k}^{K} \abs{ \tilde{\mh}^\her_k \brc{\bphi_0} \bw_i}^2+ \sigma_k^2,\\
	B_{1k}^{\rm e} \brc{\mW} &= 1+ \norm{ 
		\mathbf{\Upsilon}_{\ee} \mathbf{\tilde{G}^\her} \brc{\bphi_0} \bw_k }^2.
	\end{align}
\end{subequations}

Although $S_1^{ \rm m} \brc{\mW}$ is of the standard form given in \eqref{eq:SLR}, the function $S_1^{ \rm e} \brc{\mW}$ does not fulfill the non-negativity constraint of the nominator: One can rewrite $S_1^{ \rm e} \brc{\mW}$ as
\begin{align}
	S_1^{ \rm e} \brc{\mW} &= 
	\sum_{k=1}^{K} \omega_k b_{0k} \log \left( 1+ \frac{1 - B_{1k}^{\rm e} \brc{ \mW }}{  B_{1k}^{\rm e} \brc{ \mW } } \right) .
\end{align}
However,  in this case $1 - B_{1k}^{\rm e} \brc{ \mW } \leq 0$ which violates the necessary condition of having a non-negative nominator. To overcome this issue, we note that
	\begin{align}
	B_{1k}^{\rm e} \brc{\mW} &\leq  1+ \norm{ 
		\mathbf{\Upsilon}_{\ee} \mathbf{\tilde{G}^\her} \brc{\bphi_0}}_F^2 \norm{\bw_k }^2.
	\end{align}
Since the power constraint enforces $\norm{\bw_k}^2 \leq P_\maxx$, we could conclude that  $B_{1k}^{\rm e} \brc{\mW} \leq B_{1 U}^{\rm e}$ with\footnote{Note that the upper bound is fixed in terms of $\mW$.}
\begin{align}
	B_{1 U}^{\rm e} = 1+ \norm{ 
		\mathbf{\Upsilon}_{\ee} \mathbf{\tilde{G}^\her} \brc{\bphi_0}}_F^2 P_\maxx .
\end{align}

Using this upper bound, we can rewrite $S_1^{ \rm m} \brc{\mW}$ as
\begin{subequations}
\begin{align}
	S_1^{ \rm e} \brc{\mW} &= \sum_{k=1}^{K} \omega_k b_{0k} \log \left( \frac{1}{  B_{1k}^{\rm e} \brc{ \mW } } \right)\\
	&= \sum_{k=1}^{K} \omega_k b_{0k} \log \left( \frac{B_{1 U}^{\rm e} }{  B_{1k}^{\rm e} \brc{ \mW } } \right) - \log B_{1 U}^{\rm e}  \\
	&= \sum_{k=1}^{K} \omega_k b_{0k} \log \left( 1+ \frac{A_{1k}^{\rm e} \brc{ \mW }}{  B_{1k}^{\rm e} \brc{ \mW } } \right)  - \log B_{1 U}^{\rm e} 
\end{align}
\end{subequations}
where we define
\begin{align}
A_{1k}^{\rm e} \brc{ \mW } = B_{1 U}^{\rm e} - B_{1k}^{\rm e} \brc{ \mW }.
\end{align}
In this alternative representation, $A_{1k}^{\rm e} \brc{ \mW } \geq 0$ which satisfies the non-negativity constraint required for using  the Lagrangian dual transform.

By defining $\hat{S}_1^{ \rm e} \brc{\mW} $ as
\begin{align}
	\hat{S}_1^{ \rm e} \brc{\mW} &= \sum_{k=1}^{K} \omega_k b_{0k} \log \left( 1+ \frac{A_{1k}^{\rm e} \brc{ \mW }}{  B_{1k}^{\rm e} \brc{ \mW } } \right),
\end{align}
we could conclude that
\begin{subequations}
	\begin{align}
		\mW^\star = &\argmax _{\mW}  \left. S_1^{ \rm m} \brc{\mW}  + \hat{S}_1^{ \rm e} \brc{\mW} \right. 
		\tag{$\hat{\mam}_1$} \\
		&\text{\normalfont subject to }  \norm{\mW}_F^2 \leq P_{\maxx}.
	\end{align}
\end{subequations}
The optimization problem $\hat{\mam}_1$ is of the standard form given in \eqref{eq:Lag1}. Using the Lagrangian dual transform, we have
\begin{subequations}
	\begin{align}
		\brc{ \mW^\star , \btt^\star , \malpha^\star } = &\argmax _{\mW,\btt , \malpha }  \left. L_1^{ \rm m} \brc{\mW, \btt}  + L_1^{ \rm e} \brc{\mW,\malpha } \right. 
		\tag{$\mal_1$} \\
		&\text{\normalfont subject to }  \norm{\mW}_F^2 \leq P_{\maxx}\\
		&\phantom{\text{\normalfont subject to }} \ \btt,\malpha \in \setR_0^{+K}
	\end{align}
\end{subequations}
where $ L_1^{ \rm m} \brc{\mW, \btt} $ and $L_1^{ \rm e} \brc{\mW,\malpha } $ are the Lagrangian~dual~objectives of $S_1^{ \rm m} \brc{\mW}$ and $\hat{S}_1^{ \rm e} \brc{\mW}$, respectively, and are given~by
\begin{subequations}
\begin{align}
\hspace{-1.5mm}	L_1^{ \rm m} \brc{\mW, \btt}  \hspace{-1mm}  &= \hspace{-1.5mm} \sum_{k=1}^K  \omega_k b_{0k} \hspace{-1mm} \left(
	\Xi \brc{t_k} \hspace{-1mm} + \hspace{-1mm} \dfrac{ \brc{1+t_k} A_{1k}^{ \rm m } \brc{\mW} }{ A_{1k}^{ \rm m } \brc{\mW} \hspace{-1mm} + \hspace{-1mm} B_{1k}^{ \rm m }  \brc{\mW}  } 
	\right)\\
\hspace{-1.5mm}		L_1^{ \rm e} \brc{\mW, \malpha}  \hspace{-1mm} &= \hspace{-1.5mm} \sum_{k=1}^K  \omega_k b_{0k} \hspace{-1mm} \left(
	\Xi \brc{\alpha_k} \hspace{-1mm} + \hspace{-1mm} \dfrac{ \brc{1+ \alpha_k } A_{1k}^{ \rm e } \brc{\mW} }{ B_{1U}^{ \rm e }   } 
	\right)
\end{align}
with $t_k$ and $\alpha_k$ being the $k$-th entry of $\btt$ and $\malpha$, respectively.
\end{subequations}

The Lagrangian dual transform of the marginal problem is a joint optimization with a non-convex objective function. To tackle this issue, we note that the constraints on block variable $\mW$ and $\brc{\btt,\malpha}$ decouple in $\mal_1$. We hence utilize \ac{bcd} once again and solve the Lagrangian dual transform by cyclically alternating between the following sub-problems:
\begin{enumerate}
	\item Update $\mW$ for fixed $\btt=\bar{\btt}$ and $\malpha = \bar{\malpha}$ as
	\begin{subequations}
		\begin{align}
			\bar{\mW} = &\argmax _{\mW}  \left. L_1^{ \rm m} \brc{\mW, \bar{\btt}}  + L_1^{ \rm e} \brc{\mW,\bar{\malpha} } \right. 
			\tag{$\mal_1^{{\rm A} }$} \\
			&\text{\normalfont subject to }  \norm{\mW}_F^2 \leq P_{\maxx}.
		\end{align}
	\end{subequations}
	\item  Update $\btt$ and $\malpha$ for a fixed $\mW=\bar{\mW}$ as
\begin{subequations}
	\begin{align}
		\brc{\bar{\btt},\bar{\malpha}} = &\argmax _{\btt,\malpha}  \left. L_1^{ \rm m} \brc{\bar{\mW}, \btt}  + L_1^{ \rm e} \brc{\bar{\mW},\malpha } \right. 
		\tag{$\mal_1^{{\rm B} }$} \\
		&\text{\normalfont subject to } \btt,\malpha \in \setR_0^{+K}.
	\end{align}
\end{subequations}
\end{enumerate}

The sub-problem $\mal_1^{\rm A}$ optimized the objective with respect to $\mW$. By dropping those terms which are constant in $\mW$, we can rewrite $\mal_1^{\rm A}$ as
	\begin{subequations}
	\begin{align}
		\bar{\mW} = &\argmax _{\mW}  \left. \sum_{k=1}^K \frac{ \abs{C_{1k}^{ \rm m} \brc{\mW}}^2 }{ D_{1k}^{ \rm m} \brc{\mW} }  +  \abs{C_{1k}^{ \rm e} \brc{\mW}}^2 \right. 
		\tag{$\hat{\mal}_1^{{\rm A} }$} \\
		&\text{\normalfont subject to }  \norm{\mW}_F^2 \leq P_{\maxx}
	\end{align}
\end{subequations}
where $C_{1k}^{ \rm m} \brc{\mW}$, $D_{1k}^{ \rm m} \brc{\mW}$, and $C_{1k}^{ \rm e} \brc{\mW}$ are defined as
\begin{subequations}
	\begin{align}
		C_{1k}^{ \rm m} \brc{\mW} &= \sqrt{\omega_k b_{0k} \brc{1+\bar{t}_k}} \left. \tilde{\mh}^\her_k \brc{\bphi_0} \bw_k \right. , \\
		D_{1k}^{ \rm m} \brc{\mW} &= A_{1k}^{ \rm m } \brc{\mW} + B_{1k}^{ \rm m }  \brc{\mW} ,  \\
		C_{1k}^{ \rm e} \brc{\mW} &= \sqrt{ \frac{\omega_k b_{0k} \brc{1+\bar{\alpha}_k} }{ B_{1U}^{\rm e} } } \sqrt{A_{1k}^{ \rm e } \brc{\mW}}.
	\end{align}
\end{subequations}
The optimization problem $\hat{\mal}_1^{{\rm A} }$ is a multiple-ratio maximization whose quadratic transform is given by
\begin{subequations}
	\begin{align}
		&\max _{\mW,\mbeta,\mgamma}  \left. Q_1^{ \rm m} \brc{\mW, \mbeta}  + Q_1^{ \rm e} \brc{\mW , \mgamma } \right. 
		\tag{$\maq_1$} \\
		&\text{\normalfont subject to }  \norm{\mW}_F^2 \leq P_{\maxx}\\
		&\phantom{\text{\normalfont subject to }} \ \mbeta,\mgamma\in \setC^{K}
	\end{align}
\end{subequations}
with quadratic equivalent objectives
\begin{subequations}
\begin{align}
Q_1^{ \rm m} \brc{\mW, \mbeta} & \hspace{-1mm} = \hspace{-1mm} \sum_{k=1}^K  2 \left. \Re \set{ \beta_k^* C_{1k}^{ \rm m} \brc{\mW} }\right. - \abs{\beta_k}^2 	D_{1k}^{ \rm m} \brc{\mW}\\
Q_1^{ \rm e} \brc{\mW, \mgamma} &\hspace{-1mm} = \hspace{-1mm} \sum_{k=1}^K  2 \left. \Re \set{ \gamma_k^* C_{1k}^{ \rm e} \brc{\mW} }\right. - \abs{\gamma_k}^2.
\end{align}
\end{subequations}
Here, $\beta_k$ and $\gamma_k$ denote the $k$-th entry of $\mbeta$ and $\mgamma$, respectively.

The quadratic transform $\maq_1$ can be directly solved. Alternatively, one can use the \ac{bcd} approach once again. In this case, we consider two blocks of variables: $\brc{\mW,\mgamma}$ and $\mbeta$. The solution is then found by alternating between the following two sub-problems:
\begin{itemize}
	\item[(A)] Find optimal $\mbeta$ for $\mW= \bar{\mW}$ and $\mgamma = \bar{\mgamma}$ as
	\begin{subequations}
		\begin{align}
			\bar{\mbeta} = 
			&\argmax_{\mbeta}  \left. Q_1^{ \rm m} \brc{\bar{\mW}, \mbeta}  + Q_1^{ \rm e} \brc{\bar{\mW} , \bar{\mgamma} } \right. 
			\tag{$\maq_1^{\rm A}$} \\
			&\text{\normalfont subject to }  \mbeta\in \setC^{K}.
		\end{align}
	\end{subequations}
$\maq_1^{\rm A}$ is a standard quadratic optimization with solution
	\begin{align}
		\bar{\beta}_k &= \frac{C_{1k}^{ \rm m} \brc{\bar{\mW}} }{D_{1k}^{ \rm m} \brc{\bar{\mW}}}. \label{eq:beta_bar}
	\end{align}
\item[(B)] Find optimal $\mW$ and $\mgamma$ for $\mbeta = \bar{\mbeta}$ as
\begin{subequations}
	\begin{align}
	\hspace*{-3mm}	\brc{\bar{\mW}, \bar{\mgamma} } = 
		&\argmax _{\mW,\mgamma}  \left. Q_1^{ \rm m} \brc{\mW, \bar{\mbeta}}  + Q_1^{ \rm e} \brc{\mW , {\mgamma} } \right. 
		\tag{$\maq_1^{\rm B}$} \\
		&\text{\normalfont subject to }  \norm{\mW}_F^2 \leq P_{\maxx}  \text{ and } \mgamma\in \setC^{K}.
	\end{align}
\end{subequations}
Similar to $\maq_1^{\rm A}$, $\maq_1^{\rm B}$ is a standard quadratic optimization whose solution is given by 
\begin{subequations}
	\begin{align}
		\bar{\mW} &= \dbc{ \bar{\bw}_1, \ldots, \bar{\bw}_K }, \label{eq:w_bar}\\
		\bar{\mgamma} &= \dbc{C_{11}^{ \rm e} \brc{\bar{\mW}}, \ldots, C_{1K}^{ \rm m} \brc{\bar{\mW}} }.
	\end{align}
\end{subequations}
In \eqref{eq:w_bar}, $\bar{\bw}_k$ for $k\in\dbc{K}$ is 
\begin{align}
	\bar{\bw}_k = \rho_k \left. \mathbf{\Gamma}_k^{-1} \right. \tilde{\mh}_k \brc{\bphi_0} \label{eq:w_k_bar}
\end{align}
where $\mathbf{\Gamma}_k $ and $\rho_k$ are given by
\begin{subequations}
\begin{align}
	\hspace{-1.5mm} \mathbf{\Gamma}_k &\hspace{-.5mm}= \hspace{-.5mm}{ \tilde{\mH} \brc{\bphi_0} \mathbf{\Upsilon}_{\rm m}^2 \tilde{\mH}^\her \brc{\bphi_0} \hspace{-.75mm}+\hspace{-.75mm} \tau_k \tilde{\mG} \brc{\bphi_0} \mathbf{\Upsilon}_{\ee}^2 \tilde{\mG}^\her \brc{\bphi_0} } + \lambda_P \mI_M, \\
	 \hspace{-1.5mm} \rho_k &\hspace{-.5mm}= \hspace{-.5mm} \sqrt{\omega_k b_{0k} \brc{1+\bar{t}_k}} \bar{\beta}_k
\end{align}
\end{subequations}
with $\mathbf{\Upsilon}_{\rm m} $ and $\tau_k$ being defined as
\begin{subequations}
	\begin{align}
		\mathbf{\Upsilon}_{\rm m} &= \diag\set{\abs{\bar{\beta}_1},\ldots , \abs{\bar{\beta}_k}},\\
		\tau_k &= \frac{\omega_k b_{0k} \brc{1+\bar{\alpha}_k} }{ B_{1U}^{\rm e} },
	\end{align}
\end{subequations}
and some $\lambda_P$ being calculated according to the power constraint\footnote{$\lambda_P$ can be tuned optimally via various algorithms, e.g., \textit{bi-section search} \cite{shi2018spectral,pan2020multicell}. Alternatively, one can solve $\maq_1^{\rm B}$ directly via a convex programming solver, e.g., the CVX package in MATLAB \cite{cvx,gb08}.}.
\end{itemize}
Noting that $\bar{\mW}$ is fixed in terms of $\bar{\mgamma}$, the auxiliary vector $\bar{\mgamma}$ can be ignored in the \ac{bcd} algorithm. Hence, the solution to $\mal_1^{\rm A}$ is given by alternating between \eqref{eq:w_bar} and \eqref{eq:beta_bar}.

\begin{remark}
The precoding matrix $\bar{\mW}$ inverts the main channels while aligning the nulls of the beams with the eavesdropping channels. This is the so-called \textit{\ac{srzf} }precoding technique which has been proposed and studied in \cite{asaad2019secure}. Unlike the generic form of \ac{srzf} precoding, $\bar{\mW}$ does not need to be tuned. This is due to the fact that $\bar{\mW}$ is directly derived from secrecy sum-rate maximization.
\end{remark}

We now discuss the solution to the sub-problem $\mal_1^{\rm B}$. Following the concavity of the objective function and its decoupled form, it is straightforward to show that 
\begin{subequations}
	\label{eq:Sub_t_alph}
	\begin{align}
	\bar{t}_k &= \frac{A_{1k}^{ \rm m } \brc{\bar{\mW}}}{B_{1k}^{ \rm m } \brc{ \bar{\mW} }},\label{eq:t_bar}\\
	\bar{\alpha}_k &= \frac{A_{1k}^{ \rm e } \brc{ \bar{\mW} }}{B_{1k}^{ \rm e } \brc{ \bar{\mW} }}. \label{eq:alpha_bar}
	\end{align}
\end{subequations}
Consequently, the solution to the dual problem $\mal_1$ is given by alternating between the solution of the sub-problem $\mal_1^{\rm A}$ and the updates in \eqref{eq:Sub_t_alph}. Once the alternations converge, the solution to the first marginal problem is given by $\bar{\mW}$.

\subsection{Second Marginal Problem}
\label{sec:Second_Mar}
In the second marginal optimization, we determine a vector of phase-shifts which maximizes the objective $\mar^\ssr_\rmq \brc{ \mW, \bphi, \mb}$ assuming that $\mW$ and $\mb$ are fixed. Hence, the corresponding optimization is
\begin{subequations}
\begin{align}
	\bphi^\star = &\argmax _{\bphi}  \left. \mar^\ssr_\rmq \brc{ \mW_0, \bphi, \mb_0}\right. 
	\tag{$\mam_2$}  \label{eq:marg2}\\
	&\text{\normalfont subject to } \abs{\phi_n}=1,~\forall n \in \dbc{ N },
\end{align}
\end{subequations}
for some fixed $\mW_0$ and $\mb_0$.

In its initial form, the marginal problem $\mam_2$ is intractable due to the following issues:
\begin{enumerate}
\item The objective function is non-convex.
\item The constraint is unit-modulus.
\end{enumerate}
Considering the fractional form of the objective function, the first issue is addressed by transforming the marginal problem $\mam_2$ to a quadratic optimization using the Lagrangian dual and the quadratic transforms. Given the transformed optimization, we further overcome the second issue via the \ac{mm} method. The detailed derivations are presented in the sequel.

We start the derivations by rewriting the objective function as a sum of logarithmic ratios. Similar to the marginal problem $\mam_1$, we can write
\begin{align}
	\mar^\ssr_\rmq \brc{ \mW_0, \bphi, \mb_0} = S_{2}^{ \rm m} \brc{\bphi} + S_{2}^{ \rm e} \brc{\bphi}
\end{align}
where the functions $S_2^{ \rm m} \brc{\bphi}$ and  $S_2^{ \rm e} \brc{\bphi}$ are defined as
\begin{subequations}
	\begin{align}
		\label{eq:fell2}
		S^{ \rm m}_2 \brc{\bphi} &= \sum_{k=1}^{K} \omega_k b_{0k} \log \left( 1+\frac{ A_{2k}^{\rm m} \brc{\bphi} }{ B_{2k}^{\rm m} \brc{\bphi}  }\right) , \\
		S_2^{ \rm e} \brc{\bphi} &= 
		\sum_{k=1}^{K} \omega_k b_{0k} \log \left( \frac{1}{  B_{2k}^{\rm e} \brc{ \bphi} } \right) .
	\end{align}
\end{subequations}
Here, $	A_{2k}^{\rm m} \brc{\bphi} $, $	B_{2k}^{\rm m} \brc{\bphi} $, and $	B_{2k}^{\rm e} \brc{\bphi} $ are given by
\begin{subequations}
	\begin{align}
		A_{2k}^{\rm m} \brc{\bphi} &= \abs{\tilde{\mh}^\her_k \brc{\bphi} \bw_{0k}}^2, \\
		B_{2k}^{\rm m} \brc{\bphi} &= \sum_{i=1, i \neq k}^{K} \abs{ \tilde{\mh}^\her_k \brc{\bphi} \bw_{0i}}^2+ \sigma_k^2,\\
		B_{2k}^{\rm e} \brc{\bphi} &= 1+ \esnr_k\brc{\mW_0, \bphi }
	\end{align}
\end{subequations}
where $\bw_{0k}$ denotes the $k$-th column of $\mW_0$.

In order to use the Lagrangian dual transform, we need to convert $S_2^{ \rm e} \brc{\bphi}$ to the standard form given in \eqref{eq:SLR}. To this end, we substitute $\bgg_j\brc{\bphi} = \bgg_{\rd, j} + \mG_j \bphi$ into \eqref{eq:esnr} and write 
\begin{subequations}
\begin{align}
\esnr_k\brc{\mW, \bphi }&\hspace*{-.75mm}=\hspace*{-.75mm}\sum_{j=1}^J  \dfrac{1}{\mu_j^2}\abs{\bgg_{\rd, j}^\her \bw_k + \bphi^\her \mG_j^\her \bw_k}^2 \\
&\hspace*{-.75mm}\stackrel{\star}{\leq} \sum_{j=1}^J \hspace*{-.75mm} \dfrac{1}{\mu_j^2} \hspace*{-.75mm} \brc{
\abs{ \bgg_{\rd, j}^\her \bw_k } \hspace*{-.75mm}+\hspace*{-.75mm}
\abs{ \bphi^\her \mG_j^\her \bw_k  }
}^2 \\
&\hspace*{-.75mm}\leq \hspace*{-.75mm} \sum_{j=1}^J  \dfrac{1}{\mu_j^2} \hspace*{-.75mm} \brc{
\abs{\bgg_{\rd, j}^\her \bw_k} \hspace*{-.75mm} + \hspace*{-.75mm} \norm{\bphi}
\norm{\mG_j^\her \bw_k}
}^2 \\
&\hspace*{-.75mm}\stackrel{\dagger}{=} \hspace*{-.75mm} \sum_{j=1}^J  \dfrac{1}{\mu_j^2} \hspace*{-.75mm}  \brc{
\abs{\bgg_{\rd, j}^\her \bw_k} \hspace*{-.75mm} + \hspace*{-1mm} \sqrt{N} \norm{ \mG_j^\her \bw_k}
}^2
\end{align}
\end{subequations}
where $\star$ comes from the triangle inequality, and $\dagger$ follows from the fact that $\abs{\phi_n}=1$. Defining the upper-bound $B_{2Uk}^{\rm e}$ for $B_{2k}^{\rm e} \brc{\bphi }$ as
\begin{align}
B_{2Uk}^{\rm e} = 1+ \sum_{j=1}^J  \dfrac{1}{\mu_j^2} \brc{
	\abs{\bgg_{\rd, j}^\her \bw_{0k}} + \sqrt{N} \norm{ \mG_j^\her \bw_{0k}}
}^2,
\end{align}
we can rewrite $S_2^{ \rm e} \brc{\bphi}$ as
\begin{align}
	S_2^{\ee} \brc{\bphi} = \hat{S}_2^{ \rm e} \brc{\bphi}  - \sum_{k=1}^{K} \omega_k b_k \log B_{2Uk}^{\rm e}
\end{align}
where $\hat{S}_2^{ \rm e} \brc{\bphi}$ is given by
\begin{align}\label{eq:trnsF_e}
\hat{S}_2^{ \rm e} \brc{\bphi} = \sum_{k=1}^{K} \omega_k b_k\log\left( 1+\frac{  A_{2k}^{\rm e} \brc{\bphi } }{B_{2k}^{\rm e} \brc{\bphi} } \right)
\end{align}
with $A_{2k}^{\rm e} \brc{\bphi } = B_{2Uk}^{\rm e} - B_{2k}^{\rm e} \brc{\bphi } \geq 0$.

Given this alternative representation of the objective, the marginal problem $\mam_2$ is rewritten as
\begin{subequations}
	\begin{align}
		\bphi^\star = &\argmax _{\bphi}  \left. S_2^{\rm m} \brc{\bphi} + \hat{S}_2^{\rm e} \brc{\bphi} \right. 
		\tag{$\hat\mam_2$}  \label{eq:marg2R}\\
		&\text{\normalfont subject to } \abs{\phi_n}=1,~\forall n \in \dbc{ N }
	\end{align}
\end{subequations}
which is of the standard form in \eqref{eq:Lag1}. Using the Lagrangian dual transform, we could find $\bphi^\star$ from the transformed version
\begin{subequations}
\begin{align}
	\brc{\bphi^\star, \bq^\star , \mpsi^\star} = &\argmax _{\bphi,\bq,\mpsi}  \left. L_2^{\rm m} \brc{\bphi,\bq} + L_2^{\rm e} \brc{\bphi,\mpsi} \right. 
	\tag{$\mal_2$}  \label{eq:marg2L}\\
	&\text{\normalfont subject to } \abs{\phi_n}=1,~\forall n \in \dbc{ N }\\
	&\phantom{\text{\normalfont subject to }} \bq , \mpsi \in \setR_0^{+K}
\end{align}
\end{subequations}
where the Lagrangian dual objectives are
\begin{subequations}
	\begin{align}
		\hspace{-1.5mm}	L_2^{ \rm m} \brc{\bphi, \bq}  \hspace{-1mm}  &= \hspace{-1.5mm} \sum_{k=1}^K  \omega_k b_{0k} \hspace{-1mm} \left(
		\Xi \brc{q_k} \hspace{-1mm} + \hspace{-1mm} \dfrac{ \brc{1+q_k} A_{2k}^{ \rm m } \brc{\bphi} }{ A_{2k}^{ \rm m } \brc{\bphi} \hspace{-1mm} + \hspace{-1mm} B_{2k}^{ \rm m }  \brc{\bphi}  } 
		\right) , \\
		\hspace{-1.5mm}		L_2^{ \rm e} \brc{\bphi, \mpsi}  \hspace{-1mm} &= \hspace{-1.5mm} \sum_{k=1}^K  \omega_k b_{0k} \hspace{-1mm} \left(
		\Xi \brc{\psi_k} \hspace{-1mm} + \hspace{-1mm} \dfrac{ \brc{1+ \psi_k } A_{2k}^{ \rm e } \brc{\bphi} }{ B_{2Uk}^{ \rm e }   } 
		\right)
	\end{align}
	with $q_k$ and $\psi_k$ denoting entry $k$ of $\bq$ and $\mpsi$, respectively.
\end{subequations}

To address the Lagrangian dual problem, we follow the same \ac{bcd}-based approach as for the first marginal problem: Starting from an initial point, we alternate between the following two sub-problems:
\begin{enumerate}
\item Update $\bphi$ for fixed $\bq = \bar{\bq}$ and $\mpsi = \bar{\mpsi}$ as
\begin{subequations}
	\begin{align}
		\bar{\bphi} = &\argmax _{\bphi}  \left. L_2^{\rm m} \brc{\bphi,\bar{\bq}} + L_2^{\rm e} \brc{\bphi,\bar{\mpsi}} \right. 
		\tag{$\mal_2^{\rm A}$}  \label{eq:marg2Q1}\\
		&\text{\normalfont subject to } \abs{\phi_n}=1,~\forall n \in \dbc{ N }.
	\end{align}
\end{subequations}
\item Update $\bq$ and $\mpsi$ for fixed $\bphi = \bar{\bphi}$ as
\begin{subequations}
	\begin{align}
		\brc{\bar{\bq} , \bar{\mpsi} } = &\argmax _{\bq,\mpsi}  \left. L_2^{\rm m} \brc{\bar{\bphi},\bq} + L_2^{\rm e} \brc{\bar{\bphi},\mpsi} \right. 
		\tag{$\mal_2^{\rm B}$}  \label{eq:marg2Q2}\\
		&\text{\normalfont subject to } \bq , \mpsi \in \setR_0^{+K}.
	\end{align}
\end{subequations}
\end{enumerate}

Noting that the sub-problem $\mal_2^{\rm A}$ is to maximize a multiple-ratio term, we use the quadratic transform to find the equivalent quadratic problem. After dropping the constant terms and some standard lines of derivation the equivalent quadratic problem becomes
\begin{subequations}
	\begin{align}
		&\max _{\bphi,\bff,\bpi}  \left. Q_2^{\rm m} \brc{\bphi,\bff } + Q_2^{\rm e} \brc{\bphi, \bpi } \right. 
		\tag{$\maq_2$}  \label{eq:marg2Q1Trans}\\
		&\text{\normalfont subject to } \abs{\phi_n}=1,~\forall n \in \dbc{ N }\\
		&\phantom{\text{\normalfont subject to }} \bff,\bpi \in \setC^K
	\end{align}
\end{subequations}
for the quadratic equivalent objectives
\begin{subequations}
	\begin{align}
		Q_2^{\rm m} \brc{\bphi,\bff } & \hspace{-.5mm} = \hspace{-.5mm} \sum_{k=1}^K  2 \left. \Re \set{ f_k^* C_{2k}^{ \rm m} \brc{\bphi} }\right. - \abs{f_k}^2 	D_{2k}^{ \rm m} \brc{\bphi},\\
		 Q_2^{\rm e} \brc{\bphi, \bpi } &\hspace{-.5mm} = \hspace{-.5mm} \sum_{k=1}^K  2 \left. \Re \set{ \varpi_k^* C_{2k}^{ \rm e} \brc{\bphi} }\right. - \abs{\varpi_k}^2,
	\end{align}
\end{subequations}
where $f_k$ and $\varpi_k$ are entry $k$ of $\bff$ and $\bpi$, respectively, and $C_{2k}^{ \rm m} \brc{\bphi} $, $D_{2k}^{ \rm m} \brc{\bphi} $, and $C_{2k}^{ \rm e} \brc{\bphi} $ are defined as
\begin{subequations}
	\begin{align}
		C_{2k}^{ \rm m} \brc{\bphi} &= \sqrt{\omega_k b_{0k} \brc{1\hspace*{-.5mm}+\hspace*{-.5mm}\bar{q}_k}} \brc{ \mh^\her_{\rd, k}\bw_{0k} \hspace*{-.5mm} + \hspace*{-.5mm} \bphi^\her \mH^\her_k \bw_{0k} }, \hspace*{-1mm} \\
		D_{2k}^{ \rm m} \brc{\bphi} &= \norm{ \mh^\her_{\rd, k}\mW_{0} + \bphi^\her \mH^\her_k \mW_{0} }^2 + \sigma_k^2,  \\
		C_{2k}^{ \rm e} \brc{\bphi} &= \sqrt{ \frac{\omega_k b_{0k} \brc{1+\bar{\psi}_k} }{ B_{2Uk}^{\rm e} } } \sqrt{A_{2k}^{ \rm e } \brc{\bphi}}.
	\end{align}
\end{subequations}

As the constraints in $\maq_2$ decouple, we find the solution via alternation between the following two steps
\begin{itemize}
	\item[(A)] Find optimal $\bff$ for $\bphi= \bar{\bphi}$ and $\bpi = \bar{\bpi}$ as
	\begin{subequations}
		\begin{align}
			\bar{\bff} = &\argmax _{\bff}  \left. Q_2^{\rm m} \brc{\bar{\bphi},\bff } + Q_2^{\rm e} \brc{\bar{\bphi}, \bar{\bpi }} \right. 
			\tag{$\maq_2^\rmA$} \\
			&\text{\normalfont subject to } \bff \in \setC^K.
		\end{align}
	The solution to this marginal optimization is
		\begin{align}
		\bar{f}_k &= \frac{C_{2k}^{ \rm m} \brc{\bar{\bphi}} }{D_{2k}^{ \rm m} \brc{\bar{\bphi}}}. \label{eq:f_bar}
	\end{align}
	\end{subequations}
	\item[(B)] Find optimal $\bphi$ and $\bpi$ for $\bff= \bar{\bff}$ as
\begin{subequations}
	\begin{align}
		\brc{ \bphi^\star, \bar{\bpi} } = &\argmax _{\bphi,\bpi}  \left. Q_2^{\rm m} \brc{\bphi,\bar{\bff} } + Q_2^{\rm e} \brc{{\bphi}, {\bpi }} \right. 
		\tag{$\maq_2^{\rm B}$} \\
		&\text{\normalfont subject to } \abs{\phi_n}=1,~\forall n \in \dbc{ N }\\
		&\phantom{\text{\normalfont subject to }} \bpi \in \setC^K.
	\end{align}
\end{subequations}
It is straightforward to show that the variable $\bpi$ in $\maq_2^{\rm B}$ does not participate in the update of step (A), and hence $\maq_2^{\rm B}$ can be replaced with
\begin{subequations}
	\begin{align}
		 \bphi^\star \hspace*{-1mm}= &\argmax _{\bphi}  \left. Q_2^{\rm m} \brc{\bphi,\bar{\bff} } \hspace*{-1mm} + \hspace*{-1mm} { \frac{\omega_k b_{0k} \brc{1 \hspace*{-1mm} + \hspace*{-1mm} \bar{\psi}_k} }{ B_{2Uk}^{\rm e} } } {A_{2k}^{ \rm e } \brc{\bphi}} \right. \hspace*{-1mm}
		\tag{$\hat\maq_2^{\rm B}$} \\
		&\text{\normalfont subject to } \abs{\phi_n}=1,~\forall n \in \dbc{ N }.
	\end{align}
\end{subequations}
\end{itemize}

Despite the quadratic objective of $\hat\maq_2^{\rm B}$, the optimization is still intractable due to the unit-modulus constraint. There are various approaches by which $\bphi^\star$ can be approximated tractably; see for instance \cite{ma2010semidefinite,sun2017majorization}, and \cite{absil2009optimization}. In the following, we consider the algorithm $Q_{\rm MM} \brc{\cdot}$ which uses \ac{mm} method to determine the approximation
\begin{align}
	\bar{\bphi} = Q_{\rm MM} \brc{\bar{\bff},\bar{\bq},\bar{\mpsi},\mW_0} \label{eq:phi_bar}
\end{align}
iteratively.  The algorithm is illustrated in Algorithm~\ref{alg:MM}. In this algorithm, $\mQ$ and $\bvv$ are determined in terms of $\bar{\bff}$ as 
\begin{subequations}
	\label{eq:Uv}
	\begin{align}
		\mQ  &= \sum_{k=1}^K  \abs{\bar{f}_k}^2 \mH_k^\her \mW_{0} \mW_0^\her  \mH_k + \sum_{k=1}^K \left. \sum_{j=1}^J  \frac{\kappa_k}{\mu_j^2} \mG_j^\her \bw_{0k} \bw_{0k}^\her \mG_j \right. \label{eq:mQ}, \\
		\bvv &= \sum_{k=1}^K  \abs{\bar{f}_k}^2 \mH_k^\her \mW_{0} \mW_0^\her \mh_{\rd, k} - \eta_k \mH_k^\her \bw_{0k} + \sum_{k=1}^K \left. \sum_{j=1}^J  \frac{\kappa_k}{\mu_j^2} \mG_j^\her \bw_{0k} \bw_{0k}^\her \bgg_{\rd,j} \right. , \label{eq:bv}
	\end{align}
\end{subequations}
where $\kappa_k$ and $\eta_k$ are given by
\begin{subequations}
	\begin{align}
		\kappa_k &= \frac{\omega_k b_{0k} \brc{1  + \bar{\psi}_k} }{ B_{2Uk}^{\rm e} }, \\
		\eta_k &= \bar{f}_k^* \sqrt{ \omega_k b_{0k} \brc{1 + \bar{q}_k} }.
	\end{align}
\end{subequations}
Furthermore, $\buu^\her_n$ and $\vv_n$ denote the $n$-th row of $\mQ$ and $\bvv$, respectively, and $\lambda_{\max}$ represents the maximum eigenvalue of $\mQ$. The detailed derivations are given in Appendix~\ref{app:MM}. Following the classical approach, e.g., the approach in \cite{song2015optimization}, the convergence of Algorithm~\ref{alg:MM} can be shown. This means that Algorithm~\ref{alg:MM} in each iteration updates the phase-shift, such that the objective in $\hat{\maq}_2^{\rm B}$ evolves in a non-decreasing fashion\footnote{More precisely, the objective at the end of each iteration is larger or equal to the objective value in the previous iteration. Noting that the objective has a bounded maximum, this concludes that the algorithm converges to a fixed-point.}. Due to the similarity, we drop the convergence analysis, and refer the interested reader to \cite{song2015optimization} and the references therein.
	\begin{algorithm} 
	\caption{The MM-based algorithm $Q_{\rm MM} \brc{\cdot}$}
	\label{alg:MM}
	\begin{algorithmic}[1]
		\INPUT{$\bar{\bff}$, $\bar{\bq}$, $\bar{\mpsi}$ and $\mW_0$}
		\REQUIRE Set a feasible initial point ${\bphi}$, and define
		\begin{align*}
			\mar_{\rm M}^\ssr \brc{\bphi} =  Q_2^{\rm m} \brc{\bphi,\bar{\bff} }  + { \frac{\omega_k b_{0k} \brc{1 +  \bar{\psi}_k} }{ B_{2Uk}^{\rm e} } } {A_{2k}^{ \rm e } \brc{\bphi}} .
		\end{align*}
	Calculate $\mU$ and $\bvv$ from \eqref{eq:Uv}, and let $\lambda_{\max}$ be the maximum eigenvalue of $\mU$.
		\IF{$\mar_{\rm M}^\ssr \brc{\bphi} $ has not converged}{
		\STATE  Set $\hat{\phi}_n  = {\buu_n^\her \bphi  +\vv_n - \lambda_{\max} \phi_{n} }$ for $n\in \dbc{N}$
		\STATE Update $\displaystyle {\phi}_n = - \frac{\hat{\phi}_n }{\abs{\hat{\phi}_n}}$ for $n\in \dbc{N}$, and go back to line 1
	}
\ENDIF
	\end{algorithmic} 
\end{algorithm}

\begin{remark}
It is worth mentioning that $Q_{\rm MM} \brc{\cdot}$ does not necessarily results in the best approximation possible to be determined tractably. One can hence also consider other possible approaches for approximating $\bphi^\star$. An example of an alternative algorithm based on the \ac{bcd} method is given in Appendix~\ref{app:BCDPhi}.
\end{remark}

By alternating between \eqref{eq:phi_bar} and \eqref{eq:f_bar}, the solution to the sub-problem $\mal_2^{\rm A}$ is determined. It is further straightforward to show that the solution to $\mal_2^{\rm B}$, when $\bphi = \bar{\bphi}$, is given by
\begin{subequations}
	\label{eq:Sub_q_alph}
	\begin{align}
		\bar{q}_k &= \frac{A_{2k}^{ \rm m } \brc{\bar{\bphi}}}{B_{2k}^{ \rm m } \brc{ \bar{\bphi} }},\label{eq:q_bar}\\
		\bar{\psi}_k &= \frac{A_{2k}^{ \rm e } \brc{ \bar{\bphi} }}{B_{2k}^{ \rm e } \brc{ \bar{\bphi} }}. \label{eq:psi_bar}
	\end{align}
\end{subequations}

\subsection{Third Marginal Problem}
The third marginal optimization determines the auxiliary vector $\mb$ which maximizes $\mar^\ssr_\rmq \brc{ \mW, \bphi, \mb}$ for fixed $\mW$ and $\mb$. The corresponding optimization is thus given by
\begin{subequations}
	\begin{align}
		\mb^\star = &\argmax _{\mb}  \left. \mar^\ssr_\rmq \brc{ \mW_0, \bphi_0, \mb}\right. 
		\tag{$\mam_3$}  \label{eq:marg3}\\
		&\text{\normalfont subject to } \mb\in\dbc{0,1}^K
	\end{align}
\end{subequations}
for some fixed $\mW_0$ and $\bphi_0$.

$\mam_3$ is a linear program whose solution is given by
\begin{align} \label{eq:mb}
\bar{b}_{k}=
\begin{cases}
1 &\sinr_k\brc{\mW_0, \bphi_0} > \esnr_k \brc{\mW_0 , \bphi_0 }  \\
0 &\text{Otherwise}
\end{cases}.
\end{align}
This means that in each iteration, the objective includes only those \acp{ut} whose achievable secrecy rates are non-zero. 

\section{Iterative Algorithms with Reduced Complexity}
\label{Sec:Red}
The \ac{bcd} method developed in Section~\ref{sec:ALG_1} approximates the jointly optimal design using a multi-tier \ac{bcd}-based algorithm: The algorithm alternates among three marginal problems. The solution to each marginal problem is further determined via the \ac{bcd} method in which some inner loops alternate among multiple sub-problems. Some of these sub-problems are further solved via alternation among multiple other sub-problems. For instance, the solution to the first marginal problem is given by alternating between the two sub-problems $\mal_1^\rmA$ and $\mal_1^{\rm B}$, and the solution of $\mal_1^{\rm A}$, is given by alternating between the update rules in \eqref{eq:w_bar} and \eqref{eq:beta_bar}. Consequently, in each iteration of the first inner loop\footnote{By the first inner loop, we mean alternating between $\mal_1^\rmA$ and $\mal_1^{\rm B}$.}, the second inner loop\footnote{By the second inner loop, we mean alternating between \eqref{eq:w_bar} and \eqref{eq:beta_bar}.} should run for several iterations, such that the inner \ac{bcd}-based algorithm converges.

The multi-tier nature of the proposed algorithm yields high computational complexity. In this section, we address this issue by proposing two iterative algorithms with reduced complexity. These algorithms reduce the computational complexity by merging multiple inner loops of different \ac{bcd} tiers. We show that these algorithms converge to their fixed-points in a non-decreasing fashion, meaning that the updated objective in each iteration is necessarily greater than or equal to the objective calculated in the previous iteration.

\subsection{Two-Tiers Algorithm} 
\label{sec:algTwoTier}
The first iterative algorithm comprises of two tiers. The inner tier corresponds to the marginal problems and contains the merged version of inner loops. The algorithm uses the decoupling of sub-problems and perform the updates in parallel. The inner loop of iterations in the second marginal problem, by which the phase-shifts are updated, is however not merged. This leads to a \textit{moderate} computational complexity. 

\subsubsection{Statement of the Algorithm}
The algorithm is represented in Algorithm~\ref{alg:2Tier}. This algorithm follows the standard \ac{bcd} approach by alternating among the marginal problems. Here, $\alg_1\brc{\cdot}$ and $\alg_2\brc{\cdot}$ refer to inner loops which approximate the solution of first and second marginal problems, respectively. These loops are presented in Algorithms~\ref{Alg:1} and \ref{Alg:2}.

\begin{algorithm} 
\caption{Two-Tiers Algorithm for Precoding and Phase-Shift Tuning}
\label{alg:2Tier}
\begin{algorithmic}[1]
	\REQUIRE{Set $\mW$ and $\bphi$ to some initial values.} 
	\IF{ $\mar^\ssr \brc{\mW,\bphi}$ has not converged }{
		\STATE Update $\mW$ as $\mW = \alg_1 \brc{{\bphi}, {\mb}}$
		\STATE Update $\bphi$ as $\bphi = \alg_2 \brc{\mW, \mb}$
		\STATE Update $b_k = \mone \set{ \sinr_k\brc{\mW, \bphi} > \esnr_k \brc{\mW , \bphi } }$ for $k\in\dbc{K}$
		\STATE Go back to line 1
	}
	\ENDIF
\end{algorithmic} 
\end{algorithm}

Algorithm~\ref{Alg:1}  illustrates the inner loop for the first marginal problem. In this algorithm, the inner loops of the sub-problems are merged following the intuitive decoupling of these problems. The convergence of this algorithm is discussed in Lemma~\ref{lem:Conv_1} in Section~\ref{sec:Convergence_1}.
\begin{algorithm} 
	\caption{Inner Loop for the First Marginal Problem $\alg_1\brc{\cdot}$}
	\label{Alg:1}
	\begin{algorithmic}[1]
		\INPUT{$\bphi$ and $\mb$}
		\REQUIRE{Set $\bar{\mW}$ to some initial value and define marginal $\mar_{\rm M}^{\ssr} \brc{\mW} = \mar_\rmq^{\ssr} \brc{ \mW , {\bphi} , {\mb} }$.} 
		\IF{ $\mar_{\rm M}^{\ssr} \brc{\mW} $ has not converged }{
			\STATE Update $\bar{\beta}_k$ via \eqref{eq:beta_bar} for $k\in\dbc{K}$
			\STATE Update $\bar{t}_k$ via \eqref{eq:t_bar} for $k\in\dbc{K}$
			\STATE Update $\bar{\alpha}_k$ via \eqref{eq:alpha_bar} for $k\in\dbc{K}$
			\STATE Update $\bar{\bw}_k$ via \eqref{eq:w_k_bar} for $k\in\dbc{K}$
			\STATE Set ${\mW}_{\rm new} = \dbc{ \bar{\bw}_1, \ldots, \bar{\bw}_K }$
			\STATE Update $\bar{\mW} = \mW_{\rm new}$ and go back to line 1
		}
		\ENDIF
	\end{algorithmic} 
\end{algorithm}

The inner loop $\alg_2\brc{\cdot}$ is presented in Algorithm~\ref{Alg:2}. Here, all the \ac{bcd} loops are merged into a single loop. The \ac{mm} loop, by which the phase-shifts are updated is however not merged. This means that $\alg_2\brc{\cdot}$ contains another inner loop, i.e., $Q_{\rm MM} \brc{\cdot}$. As the result, Algorithm~\ref{alg:2Tier} yields moderate computational complexity $Q_{\rm MM} \brc{\cdot}$ requires larger number of iterations to converge. The convergence of Algorithm~\ref{Alg:2} is discussed in Lemma~\ref{lem:Conv_2} in Section~\ref{sec:Convergence_1}.

\begin{algorithm} 
	\caption{Inner Loop for the Second Marginal Problem $\alg_2\brc{\cdot}$}
	\label{Alg:2}
	\begin{algorithmic}[1]
	\INPUT{$\mW$ and $\mb$}
	\REQUIRE{Set $\bar{\bphi}$ to some initial value and define marginal $\mar_{\rm M}^{\ssr} \brc{\bphi} = \mar_\rmq^{\ssr} \brc{ \mW , {\bphi} , {\mb} }$} 
		\IF{ $\mar_{\rm M}^{\ssr} \brc{\bphi} $ has not converged }{
			\STATE Update $\bar{f}_k$ via \eqref{eq:f_bar} for $k\in\dbc{K}$
			\STATE Update $\bar{q}_k$ via \eqref{eq:q_bar} for $k\in\dbc{K}$
			\STATE Update $\bar{\psi}_k$ via \eqref{eq:psi_bar} for $k\in\dbc{K}$
			\STATE Set ${\bphi}_{\rm new} = Q_{\rm MM} \brc{\bar{\bff},\bar{\bq},\bar{\mpsi},\mW}$
			\STATE Update $\bar{\bphi} = \bphi_{\rm new}$ and go back to line 1
		}
		\ENDIF
	\end{algorithmic} 
\end{algorithm}

\subsubsection{Convergence Analysis}
\label{sec:Convergence_1}
We now discuss the convergence of Algorithm~\ref{alg:2Tier}. We show that using this two-tiers algorithm, the original objective, i.e., $\mar^\ssr\brc{\mW, \bphi}$, evolve non-decreasingly. This means that the value of the objective function after each iteration would be greater or equal to its value in the previous iteration. Since the objective has a bounded maximum, this result indicates that this algorithm converges to a fixed-point after a certain number of iterations\footnote{This fixed-point is however not necessarily the global or a local maximum and is only an approximation.}.

We start our analysis by showing the convergence of the inner loops $\alg_1\brc\cdot$ and $\alg_2\brc\cdot$ in the following lemmas:

\begin{lemma}
	\label{lem:Conv_1}
	For the inner loop $\alg_1\brc{\bar{\bphi} , \bar{\mb}}$, we have
	\begin{align}
	\mar^\ssr_\rmq  \brc{\bar{\mW}, \bar{\bphi}, \bar{\mb}} \leq \mar^\ssr_\rmq  \brc{\mW_{\rm new}, \bar{\bphi}, \bar{\mb}}.
	\end{align}
\end{lemma}
\begin{proof}
	The proof is given in Appendix~\ref{app:Proof_Conv_1}.
\end{proof}
Lemma~\ref{lem:Conv_1} indicates that in each iteration, the updated precoding matrix returned a larger value of the transformed objective function $\mar^\ssr_\rmq  \brc\cdot$. As the result, one can conclude that after a certain number of iterations, $\alg_1\brc\cdot$ must converge. This result is further stated for $\alg_2\brc\cdot$ in Lemma~\ref{lem:Conv_2}.

\begin{lemma}
	\label{lem:Conv_2}
	Assume that in iteration $t\geq 1$ of the inner loop $\alg_2\brc{\bar{\mW} , \bar{\mb}}$, $Q_{\rm MM} \brc{\cdot}$ iterates for $T_Q\brc{t}$ iterations. For any positive integer sequence of $T_Q\brc{t}$, we have
	\begin{align}
		\mar^\ssr_\rmq  \brc{\bar{\mW}, \bar{\bphi}, \bar{\mb}} \leq \mar^\ssr_\rmq  \brc{\bar{\mW}, \bphi_{\rm new}, \bar{\mb}}.
	\end{align}
This inequality further holds for any alternative of $Q_{\rm MM} \brc\cdot$ which converges to its fixed-point in a non-decreasing fashion.
\end{lemma}
\begin{proof}
	The proof is given in Appendix~\ref{app:Proof_Conv_2}.
\end{proof}

Using Lemmas~\ref{lem:Conv_1} and \ref{lem:Conv_2}, we state the following result on the convergence of Algorithm~\ref{alg:2Tier}:

\begin{theorem}
	\label{theorem:2}
Let $\mW^{\brc{t}}$ and $\bphi^{\brc{t}}$ be the precoding matrix and vector of phase-shifts updated at the end of iteration $t$ in Algorithm~\ref{alg:2Tier}. Assume that $\alg_1\brc\cdot$ and $\alg_2\brc\cdot$ iterate for $T_1\brc{t}$ and $T_2\brc{t}$, respectively, in iteration $t$ of Algorithm~\ref{alg:2Tier}. For any positive integer sequence of $T_1\brc{t}$ and $T_2\brc{t}$, we have
\begin{align*}
	\mar^\ssr\brc{\mW^{\brc{t}}, \bphi^{\brc{t}}} \leq 	\mar^\ssr\brc{\mW^{\brc{t+1}}, \bphi^{\brc{t+1}}}.
\end{align*}
The inequality holds for any alternative of $Q_{\rm MM} \brc\cdot$ in the inner loop $\alg_2\brc\cdot$ which converges to its fixed-point in a non-decreasing fashion and iterates in each iteration of $\alg_2\brc\cdot$ for $T\geq1$ iterations.
\end{theorem}

\begin{proof}
	Let $\mb^{\brc{t}}$ denote the updated auxiliary variable in iteration $t$. Considering Algorithm~\ref{alg:2Tier}, we can write
	\begin{align} 
	b_{k}^{\brc{t}}=
		\begin{cases}
			1 &\sinr_k\brc{  \mW^{\brc{t}}, \bphi^{\brc{t}} } > \esnr_k \brc{\mW^{\brc{t}} , \bphi^{\brc{t}} }  \\
			0 &\text{Otherwise}
		\end{cases}. \label{update:b}
	\end{align}
This concludes that
\begin{subequations}
	\begin{align}
	\mar^\ssr_\rmq  \brc{ \mW^{\brc{t}}, \bphi^{\brc{t}}, \mb^{\brc{t}} } &=  \sum_{k=1}^K \omega_k b^{\brc{t}}_k \log\left( \dfrac{1+\sinr_k\brc{\mW^{\brc{t}}, \bphi^{\brc{t}}}}{1+\esnr_k\brc{ \mW^{\brc{t}}, \bphi^{\brc{t}} } } \right) \\
	&=  \sum_{k=1}^K \omega_k  \dbc{ \log\left( \dfrac{1+\sinr_k\brc{\mW^{\brc{t}}, \bphi^{\brc{t}}}}{1+\esnr_k\brc{ \mW^{\brc{t}}, \bphi^{\brc{t}} } } \right) }^+\\
	&= \mar^\ssr  \brc{ \mW^{\brc{t}}, \bphi^{\brc{t}} } .
\end{align}
\end{subequations}

As the result, we can write
\begin{subequations}
	\begin{align}
		\mar^\ssr  \brc{ \mW^{\brc{t}}, \bphi^{\brc{t}} }  &= \mar^\ssr_\rmq  \brc{ \mW^{\brc{t}}, \bphi^{\brc{t}}, \mb^{\brc{t}} } \\
		&\stackrel{\dagger}{\leq}  \mar^\ssr_\rmq  \brc{ \mW^{\brc{t+1}}, \bphi^{\brc{t}}, \mb^{\brc{t}} } \\
		&\stackrel{\star}{\leq}  \mar^\ssr_\rmq  \brc{ \mW^{\brc{t+1}}, \bphi^{\brc{t+1}}, \mb^{\brc{t}} } \label{eq:proof_Th2_1}
	\end{align}
\end{subequations}
where $\dagger$ and $\star$ follow from Lemmas~\ref{lem:Conv_1} and \ref{lem:Conv_2}, respectively. Note that for $T_1\brc{t+1}=1$ and $T_2\brc{t+1}=1$, $\mW^{\brc{t}} = \mW^{\brc{t+1}}$ and $\bphi^{\brc{t}}=\bphi^{\brc{t+1}}$. Moreover, $T_2\brc{t}\geq 1$, Lemma~\ref{lem:Conv_2} indicates that $\star$ holds for any alternative of $Q_{\rm MM} \brc\cdot$ in $\alg_2\brc\cdot$ which converges in a non-decreasing fashion and iterates for $T\geq1$ iterations. Therefore, the inequality holds for any  $T_1\brc{t}=1$ and $T_2\brc{t}\geq1$, and for any alternative of $Q_{\rm MM} \brc\cdot$ with positive integer sequence of iterations.

Let us now define
the expression 
\begin{align}
	R_k^{\brc{t+1}} = \log\left( \dfrac{1+\sinr_k\brc{\mW^{\brc{t+1}}, \bphi^{\brc{t+1}}}}{1+\esnr_k\brc{ \mW^{\brc{t+1}}, \bphi^{\brc{t+1}} } } \right) .
\end{align}
From \eqref{update:b}, we conclude that for $k\in \dbc{K}$ at which ${b}_{k}^{\brc{t}} \neq {b}_{k}^{\brc{t+1}}$, we have
\begin{align}
	\begin{cases}
			R_k^{\brc{t+1}} > 0 & \text{if } {b}_{k}^{\brc{t+1}} = 1\\
			R_k^{\brc{t+1}} \leq 0 & \text{if } {b}_{k}^{\brc{t+1}} = 0\\
	\end{cases}.
\end{align}
This indicates that
\begin{subequations}
	\begin{align}
		\mar^\ssr_\rmq  \brc{ \mW^{\brc{t+1}}, \bphi^{\brc{t+1}}, \mb^{\brc{t}} } &=  \sum_{k=1}^K \omega_k b^{\brc{t}}_k 	R_k^{\brc{t+1}}\\
		&\leq  \sum_{k=1}^K \omega_k b^{\brc{t+1}}_k 	R_k^{\brc{t+1}}\\
		&= \mar^\ssr_\rmq  \brc{ \mW^{\brc{t+1}}, \bphi^{\brc{t+1}}, \mb^{\brc{t+1}} }\\
		&= \mar^\ssr  \brc{ \mW^{\brc{t+1}}, \bphi^{\brc{t+1}}}.
	\end{align}
\end{subequations}
We can hence conclude the proof from \eqref{eq:proof_Th2_1}.
\end{proof}

\subsection{Single-Loop Algorithm}
\label{sec:algSingleLoop}
The complexity of the two-tiers algorithm can be further reduced by merging the second tier also in the main loop. This complexity reduction intuitively degrades the performance of the algorithm. However, it can boost the speed of the algorithm significantly. The numerical simulations given in the next section show that despite the existence of such a trade-off, the performance degradation in comparison to the speed gain is essentially negligible.

\subsubsection{Statement of the Algorithm}
The algorithm is represented in Algorithm~\ref{alg:Single_Loop}. In this algorithm, all inner loops of the two-tiers algorithm, i.e., $\alg_1\brc\cdot$, $\alg_2\brc\cdot$ and the \ac{mm} loop in Algorithm~\ref{alg:MM}, are merged with the outer loop. In this respect, this algorithm can be observed as the special case of the two-tiers algorithm in which all the inner loops iterate only of a single iteration.

\begin{algorithm} 
	\caption{Single-Loop Algorithm for Precoding and Phase-Shift Tuning}
	\label{alg:Single_Loop}
	\begin{algorithmic}[1]
		\REQUIRE{Set $\mW$ and $\bphi$ to some initial values.} 
			\IF{ $\mar^\ssr \brc{\mW,\bphi}$ has not converged }{
			\STATE Calculate $\bar{\beta}_k$, $\bar{t}_k$, $\bar{\alpha}_k$ and $\bar{\bw}_k$ via \eqref{eq:beta_bar}, \eqref{eq:t_bar}, \eqref{eq:alpha_bar} and \eqref{eq:w_k_bar}, respectively, for $k\in\dbc{K}$
			\STATE Update ${\mW} = \dbc{ \bar{\bw}_1, \ldots, \bar{\bw}_K }$
			\STATE Calculate $\bar{f}_k$, $\bar{q}_k$ and $\bar{\psi}_k$ via \eqref{eq:f_bar}, \eqref{eq:q_bar} and \eqref{eq:psi_bar}, respectively, for $k\in\dbc{K}$
			\STATE Calculate $\hat{\phi}_n  = {\buu_n^\her \bphi_0  +\vv_n - \lambda_{\max} \phi_{0n} }$ with $\buu_n$ and $\vv_n$ given in \eqref{eq:Uv}
			\STATE Update for $n\in \dbc{N}$
			\begin{align*}
			\displaystyle {\phi}_n = - \frac{\hat{\phi}_n }{\abs{\hat{\phi}_n}}	
			\end{align*}
			\STATE Update $b_k = \mone \set{ \sinr_k\brc{\mW, \bphi} > \esnr_k \brc{\mW , \bphi } }$ for $k\in\dbc{K}$
			\STATE Go back to line 1
		}
		\ENDIF
	\end{algorithmic} 
\end{algorithm}

\subsubsection{Convergence Analysis}
The convergence of the single-loop algorithm is directly concluded from Theorem~\ref{theorem:2}. In fact, Algorithm~\ref{alg:Single_Loop} is a special case of Algorithm~\ref{alg:2Tier} in which $\alg_1\brc\cdot$, $\alg_2\brc\cdot$, and $Q_{\rm MM}\brc\cdot$ are run for a single iteration. Noting that Theorem~\ref{theorem:2} holds for any number of inner loop iterations, we conclude that Algorithm~\ref{alg:Single_Loop} converges in a non-decreasing fashion to its fixed-point with respect to the original objective $\mar^\ssr\brc{\mW, \bphi}$.

\section{Numerical Investigations}
\label{Sec:Num}
We investigate the performance of the proposed approach by performing several numerical experiments. A schematic view of the setting considered in simulations is shown in Fig.~\ref{fig:setting} for $K=4$ legitimate \acp{ut} and $J=6$ eavesdroppers. In this setting, an \ac{irs} is located at distance $D=25$ m from the \ac{bs}. $K$ legitimate \acp{ut} are located uniformly and randomly around the \ac{irs} in a circle of radius $r_{\rm IRS} = 10$ m. There are $J$ eavesdroppers with uniformly generated random distances located in a circle of radius $r_{\rm BS} = 10$ m around the \ac{bs}. To calculate the noise power, a typical GSM channel with bandwidth $200$ kHz has been considered. Assuming the noise spectral density to be $\log N_0 = -147$ dB/Hz, the noise power at the receiving terminals is set to $\log \sigma_k^2 = \log \mu_j^2 = -147$ dB.

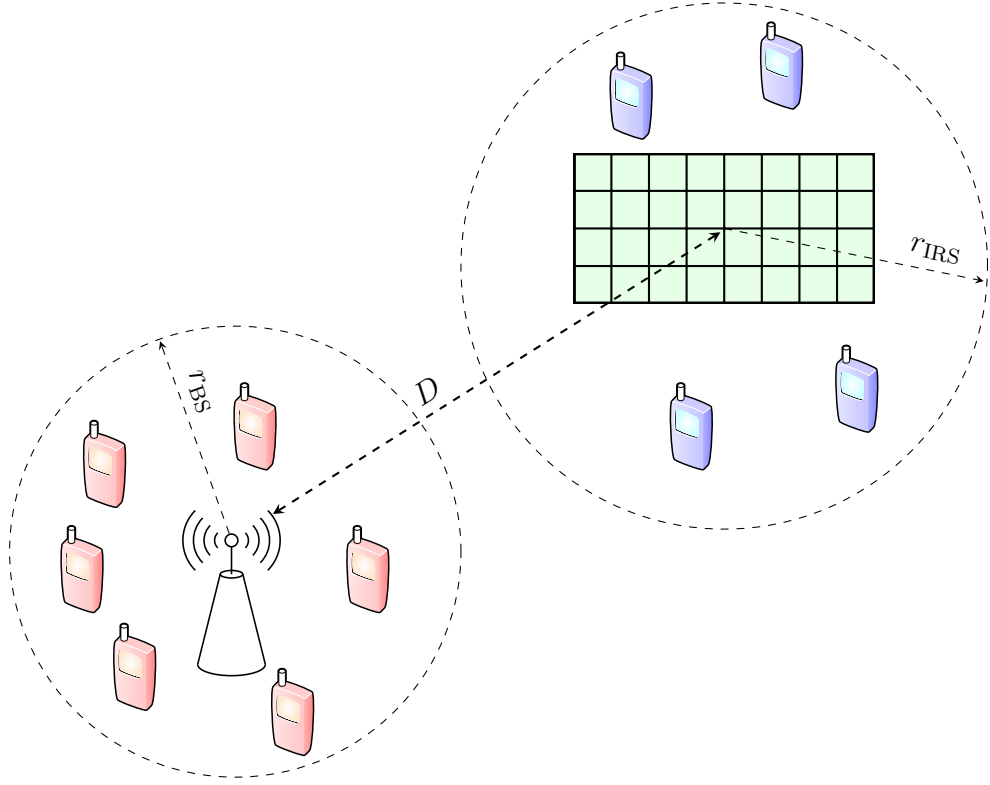
\begin{figure}
	\centering
	\begin{tikzpicture}
\tikzset{mobile phone/.pic={
code={
\begin{scope}[line join=round,looseness=0.25, line cap=round,scale=0.07, every node/.style={scale=0.07}]
\begin{scope}
\clip [preaction={left color=blue!10, right color=blue!30}] 
  (1/2,-1) to [bend left] (0,10)
  to [bend left] ++(1,1) -- ++(0,2)
  arc (180:0:3/4 and 1/2) -- ++(0,-2)
  to [bend left]  ++(5,-2) coordinate (A) to [bend left] ++(-1/2,-11)
  to [bend left] ++(-1,-1) to [bend left] cycle;
\path [left color=blue!30, right color=blue!50]
  (A) to [bend left] ++(0,-11) to[bend left] ++(-3/2,-2)
  -- ++(0,12);
\path [fill=blue!20, draw=white, line width=0.01cm]
  (0,10) to [bend left] ++(1,1) -- ++(0,2)
  arc (180:0:3/4 and 1/2) -- ++(0,-2)
  to [bend left]  (A) to [bend left] ++(-3/2,-5/4)
  to [bend right] cycle;
\draw [line width=0.01cm, fill=white]
  (9/8,21/2) arc (180:360:5/8 and 3/8) --
  ++(0,2.5) arc (0:180:5/8 and 3/8) -- cycle;
\draw [line width=0.01cm, fill=white]
  (9/8,13) arc (180:360:5/8 and 3/8);
\fill [white, shift=(225:0.5)] 
  (1,17/2) to [bend left] ++(4,-7/4)
  to [bend left] ++(0,-7/2) to [bend left] ++(-4, 6/4)
  to [bend left] cycle;
\fill [black, shift=(225:0.25)] 
  (1,17/2) to [bend left] ++(4,-7/4)
  to [bend left] ++(0,-7/2) to [bend left] ++(-4, 6/4)
  to [bend left] cycle;
\shade [inner color=white, outer color=cyan!20] 
  (1,17/2) to [bend left] ++(4,-7/4)
  to [bend left] ++(0,-7/2) to [bend left] ++(-4, 6/4)
  to [bend left] cycle;
%
\end{scope}
\draw [line width=0.02cm] 
  (1/2,-1) to [bend left] (0,10)
  to [bend left] ++(1,1) -- ++(0,2)
  arc (180:0:3/4 and 1/2) -- ++(0,-2)
  to [bend left]  ++(5,-2) to [bend left] ++(-1/2,-11)
  to [bend left] ++(-1,-1) to [bend left] cycle;
\end{scope}%
}}}

\tikzset{eve/.pic={
		code={
			\begin{scope}[line join=round,looseness=0.25, line cap=round,scale=0.07, every node/.style={scale=0.07}]
				\begin{scope}
					\clip [preaction={left color=red!10, right color=red!30}] 
					(1/2,-1) to [bend left] (0,10)
					to [bend left] ++(1,1) -- ++(0,2)
					arc (180:0:3/4 and 1/2) -- ++(0,-2)
					to [bend left]  ++(5,-2) coordinate (A) to [bend left] ++(-1/2,-11)
					to [bend left] ++(-1,-1) to [bend left] cycle;
					\path [left color=red!30, right color=red!50]
					(A) to [bend left] ++(0,-11) to[bend left] ++(-3/2,-2)
					-- ++(0,12);
					\path [fill=red!20, draw=white, line width=0.01cm]
					(0,10) to [bend left] ++(1,1) -- ++(0,2)
					arc (180:0:3/4 and 1/2) -- ++(0,-2)
					to [bend left]  (A) to [bend left] ++(-3/2,-5/4)
					to [bend right] cycle;
					\draw [line width=0.01cm, fill=white]
					(9/8,21/2) arc (180:360:5/8 and 3/8) --
					++(0,2.5) arc (0:180:5/8 and 3/8) -- cycle;
					\draw [line width=0.01cm, fill=white]
					(9/8,13) arc (180:360:5/8 and 3/8);
					\fill [white, shift=(225:0.5)] 
					(1,17/2) to [bend left] ++(4,-7/4)
					to [bend left] ++(0,-7/2) to [bend left] ++(-4, 6/4)
					to [bend left] cycle;
					\fill [black, shift=(225:0.25)] 
					(1,17/2) to [bend left] ++(4,-7/4)
					to [bend left] ++(0,-7/2) to [bend left] ++(-4, 6/4)
					to [bend left] cycle;
					\shade [inner color=white, outer color=orange!20] 
					(1,17/2) to [bend left] ++(4,-7/4)
					to [bend left] ++(0,-7/2) to [bend left] ++(-4, 6/4)
					to [bend left] cycle;
					%
				\end{scope}
				\draw [line width=0.02cm] 
				(1/2,-1) to [bend left] (0,10)
				to [bend left] ++(1,1) -- ++(0,2)
				arc (180:0:3/4 and 1/2) -- ++(0,-2)
				to [bend left]  ++(5,-2) to [bend left] ++(-1/2,-11)
				to [bend left] ++(-1,-1) to [bend left] cycle;
			\end{scope}%
}}}

\tikzset{radiation/.style={{decorate,decoration={expanding waves,angle=90,segment length=4pt}}},
	antenna/.pic={
		code={\tikzset{scale=3/10}
			\draw[semithick] (0,0) -- (1,4);
			\draw[semithick] (3,0) -- (2,4);
			\draw[semithick] (0,0) arc (180:0:1.5 and -0.5);
			\node[inner sep=4pt] (circ) at (1.5,5.5) {};
			\draw[semithick] (1.5,5.5) circle(8pt);
			\draw[semithick] (1.5,5.5cm-8pt) -- (1.5,4);
			\draw[semithick] (1.5,4) ellipse (0.5 and 0.166);
			\draw[semithick,radiation,decoration={angle=45}] (1.5cm+8pt,5.5) -- +(0:2);
			\draw[semithick,radiation,decoration={angle=45}] (1.5cm-8pt,5.5) -- +(180:2);
	}}
}

\tikzset{
irs/.pic={\clip[postaction={shade,left color=green!10,right color = green!10}](0,0) rectangle (4,2);
	\draw[thick] (0,0) grid[step=0.5] (4,2);
	\draw[ultra thick](0,0) rectangle (4,2);}
}

	\path (0,-.8) pic {antenna};
	\draw[dashed] (0.5,0.7) circle[radius=3cm];
	
	\path (01,-1.8) pic {eve};
	\path (-1.1,-1.2) pic {eve};
	\path (-1.5,1.5) pic {eve};
	\path (2,.1) pic {eve};
	\path (.5,2) pic {eve};
	\path (-1.8,.1) pic {eve};
	
	\path (5,4) pic {irs};
	\draw[dashed] (7,4.5) circle[radius=3.5cm];
	
	\path (6.3,2) pic {mobile phone};
	\path (8.5,2.5) pic {mobile phone};
	\path (5.5,6.4) pic {mobile phone};
	\path (7.5,6.8) pic {mobile phone};
	
	\draw [<->,>=stealth,dashed,thick] (1,1.2) -- (6.95,4.95) node [above, sloped,pos=.37] (d) {$D$};
	\draw [->,>=stealth,dashed] (7,5) -- (10.45,4.3) node [above, sloped,pos=.8] (d) {$r_{\rm IRS}$};
	\draw [->,>=stealth,dashed] (.42,0.94) -- (-.5,3.5) node [above, sloped,pos=.7] (d) {$r_{\rm BS}$};

\end{tikzpicture}
	\caption{A schematic representation of the simulation setting for $K=4$ legitimate \acp{ut} and $J=6$ eavesdroppers.}
	\label{fig:setting}
\end{figure}

The channel coefficients for a given \ac{ut} $k$ and eavesdropper $j$ are generated according to
\begin{align}
	\mh_{i,k} &= \varrho^{\rm L}_{i,k} \mh^0_{i,k}, \\
	\bgg_{i,j} &= \varrho^{\rm E}_{i,j} \bgg^0_{i,j} ,
\end{align}
respectively, where the index $i \in \set{ \mathrm{d} , \mathrm{r} }$ refers to the direct and reflection paths. Here, $\varrho^{\rm L}_{i,k}$ and $\varrho^{\rm E}_{i,k}$ model path-loss and $\mh^0_{i,k}$ and $\bgg^0_{i,j}$ take into account the impact of small-scale fading. Similarly, the channel between the \ac{bs} and the \ac{irs} is written as
\begin{align}
	\mT &= \varrho_{\rm BI} \mT^0
\end{align}
where $\varrho_{\rm BI}$ and $\mT^0$ model path-loss and small-scale fading effects, respectively. Throughout the simulations, we consider the standard Rayleigh fading model. This means that the entries of $\mh^0_{i,k}$ and $\bgg^0_{i,j}$ for $i\in\dbc{K}$ and $j\in\dbc{J}$, as well as the entries of $\mT^0$ are generated \ac{iid} according to a zero-mean unit-variance Gaussian distribution. 

The path-loss coefficients are generated according to the following model
\begin{align}
 \varrho\brc{ d,\varsigma} = \frac{\varrho_{\rm ref} }{d^\varsigma}.
\end{align}
In this model, $\varrho_{\rm ref}$ denotes the path-loss at the reference distance $d=1$ m which depends on the operating wave-length and antenna gains. Throughout the simulations, we set it to $\log \varrho_{\rm ref}=-30$ dB. As the result, $\varrho^{\rm L}_{i,k} = \varrho\brc{ d^{\rm L}_{i,k},\varsigma^{\rm L}_{i,k}}$ where $d^{\rm L}_{i,k}$ for $i= \rm d$ and $i= \rm r$ denotes the distance from legitimate \ac{ut} $k$ to the \ac{bs} and the distance from legitimate \ac{ut} $k$ to the \ac{irs}, respectively.  $\varsigma^{\rm L}_{i ,k}$ further denotes the path-loss exponent for the link specified by index $i$. Similarly, we set 
\begin{itemize}
	\item $\varrho^{\rm E}_{i,k} = \varrho\brc{ d^{\rm E}_{i,j} , \varsigma^{\rm E}_{i,j} }$, where $d^{\rm E}_{i,j}$ and $\varsigma^{\rm E}_{i,j} $ denote the distance to the eavesdropper $j$ and its corresponding path-loss exponent for the link specified with index $i$, respectively.
	\item $\varrho_{\rm BI} = \varrho\brc{ D,\varsigma_{\rm BI} }$, where $D$ and $\varsigma_{\rm BI}$ represent the distance between the \ac{bs} and the \ac{irs} and the corresponding path-loss exponent respectively.
\end{itemize}
Throughout the simulations, the distances are calculated from the realization of the randomized setting, and the path-loss exponents are set to $\varsigma^{\rm L}_{\mathrm{d},k} = \varsigma^{\rm E}_{\mathrm{d},k} = 3.5$ and $\varsigma^{\rm L}_{\mathrm{r},k} = \varsigma^{\rm E}_{\mathrm{r},k} = \varsigma_{\rm BI}  = 2.3$. The weighted secrecy sum-rate is evaluated for multiple realizations of the channel and then averaged numerically. To indicate this point, we denote the averaged secrecy sum-rate\footnote{This is in fact the \textit{ergodic} secrecy sum-rate.} with $\bar{\mar}^\ssr$. The weights are further set to be all equal to one, i.e., $\omega_k=1$ for $k\in\dbc{K}$.

\subsection{Reference Scenarios and Benchmark}
For each experiment, we evaluate the performance for the both iterative algorithms proposed in Sections~\ref{sec:algTwoTier} and \ref{sec:algSingleLoop}. To illustrate the performance gain obtained by employing \acp{irs}, we consider a reference scenario in which the \ac{irs} is set off. This means that $\beta_n=0$ for all $n\in \dbc{N}$. The \ac{bs} in this case determines the precoding vectors via the \ac{srzf} precoding scheme proposed in \cite{asaad2019secure}. We refer to this scenario as \textit{Ref. 1}.

To compare the performance of the proposed algorithms with a reference point, we further consider the reference scenario \textit{Ref. 2} in which phase-shifts at the \ac{irs} are set randomly and uniformly. The linear precoding vectors are further found by applying \ac{srzf} precoding to the end-to-end equivalent channel \cite{asaad2019secure}.  As the state-of-the-art performance, we also compare the proposed phase-tuning algorithms with the algorithm proposed in \cite{dong2020enhancing}. This latter algorithm is refereed to as the \textit{benchmark} throughout the investigations.

We further consider an \textit{enhanced} form of the setting in which the magnitudes of the signals reflected by \ac{irs} elements are also allowed to be  modified. Here, we assume that the $\beta_n$ is a tunable parameter whose value is taken from $\beta_n\in\dbc{0,1}$ and is updated at the same rate as the phase\footnote{This means that the constraint $\abs{\phi_{n}} = 1$ is replaced by $\abs{\phi_{n}} \leq 1$ in the optimization problems.} $\theta_{n}$. As a result, tuning of the \ac{irs} reduces to a convex program which is solved tractably via a convex programming algorithm. It is worth mentioning that this is only an \textit{enhanced} form of the system which we simulate for sake of comparison. In fact, due to the larger degrees of freedom, this enhanced form is expected to outperform the proposed algorithms, since in the system, amplitude modifications are not performed at the \ac{irs}. For this enhanced setting, we simulate two algorithms: The first algorithm performs alternating optimization, such that the \ac{irs} in each iteration is tuned via convex programming. This algorithm is referred to as \textit{C-Ref.~1} and can be seen as the enhanced form of the two-tiers algorithm. The second algorithm follows the single-loop algorithm with this minor modification that in each iteration the \ac{irs} tuning is performed directly via convex programming. This algorithm is referred to as \textit{C-Ref. 2} and is considered as the enhanced version of the single-loop algorithm. The convex programs in these algorithms are solved via the CVX toolbox in MATLAB \cite{cvx}.

\subsection{Numerical Simulations}
\label{sec:Numerical}
We start the investigations by setting $K=4$, and $J=6$ in the setting\footnote{Similar to what is shown in Fig.~\ref{fig:setting}.}. The number of  transmit antennas at the \ac{bs} is set to $M=8$, and the \ac{irs} is assumed to be equipped with $N=128$ elements. Fig.~\ref{fig:WSRvsPower} shows the weighted secrecy sum-rate against the transmit power $P_{\max}$. The figure shows significant gains achieved by the proposed algorithms, when the performance is compared with the reference scenarios Ref. 1 and Ref. 2. From the figure, it is further observed that the two-tiers algorithm considerably outperform the benchmark. This is however not the case with the single-loop algorithm. Although this algorithm always outperform the benchmark, its gain becomes negligible, as the transmit power increases.

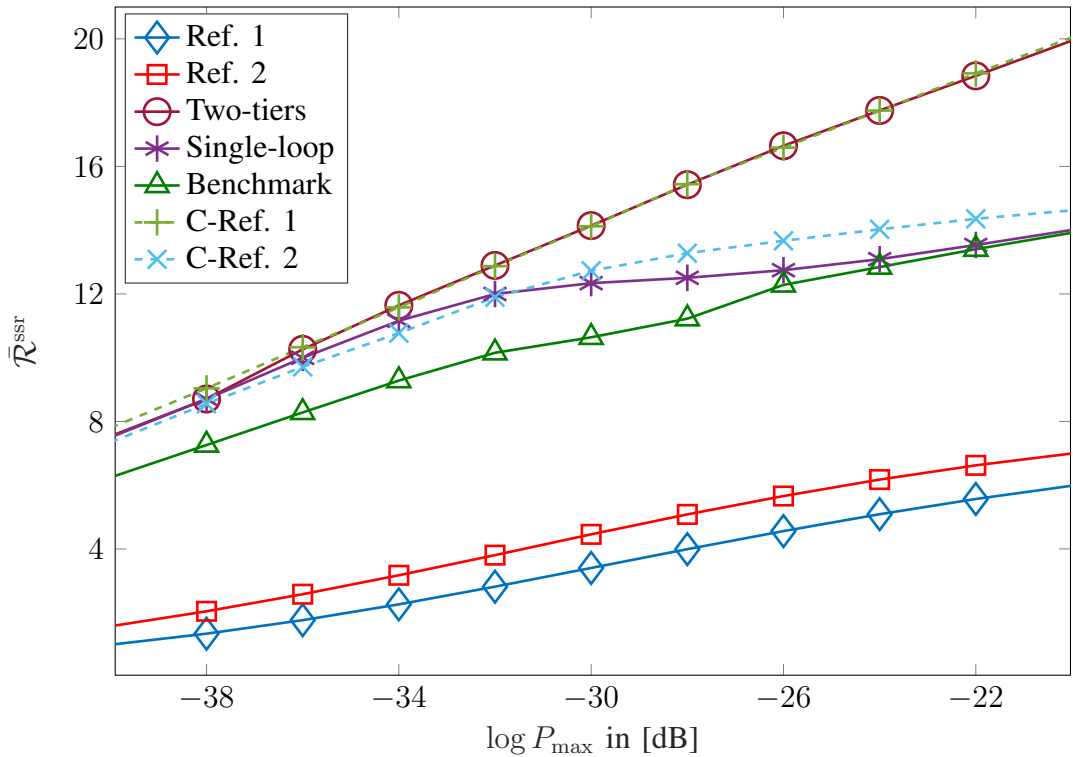
\begin{figure}
	\centering
%
%
\definecolor{mycolor1}{rgb}{0.00000,0.44700,0.74100}%
\definecolor{mycolor2}{rgb}{0.63529,0.07843,0.18431}%
\definecolor{mycolor3}{rgb}{0.49400,0.18400,0.55600}%
\definecolor{mycolor4}{rgb}{0.00000,0.49804,0.00000}%
\definecolor{mycolor5}{rgb}{0.46600,0.67400,0.18800}%
\definecolor{mycolor6}{rgb}{0.30100,0.74500,0.93300}%
\begin{tikzpicture}

\begin{axis}[%
width=5in,
height=3.5in,
at={(1.23in,0.788in)},
scale only axis,
xmin=-39.9,
xmax=-20.0307219662058,
xlabel style={font=\color{white!15!black}},
xlabel={$\log P_{\max}$ in [dB]},
xtick={-38,-34,-30,-26,-22},
xticklabels={{$-38$},{$-34$},{$-30$},{$-26$},{$-22$}},
ymin=0.0389521861914477,
ymax=21,
ylabel style={font=\color{white!15!black}},
ylabel={$\bar{\mar}^\ssr$},
ytick={4,8,12,16,20},
yticklabels={{$4$},{$8$},{$12$},{$16$},{$20$}},
axis background/.style={fill=white},
legend style={at={(axis cs: -39.7,20.8)}, anchor=north west, legend cell align=left, align=left, draw=white!15!black}
]
\addplot [color=mycolor1, line width=1.0pt, mark size=6pt, mark=diamond, mark options={solid, mycolor1}]
table[row sep=crcr]{%
	-40	0.994137332211496\\
	-38	1.34174175517578\\
	-36	1.76555176160576\\
	-34	2.26197851408818\\
	-32	2.8159232969404\\
	-30	3.40328953890991\\
	-28	3.99270930590264\\
	-26	4.55965055767643\\
	-24	5.08904891434104\\
	-22	5.5666877968582\\
	-20	5.98388270451835\\
};
\addlegendentry{Ref. 1}

\addplot [color=red, line width=1.0pt, mark size=3.5pt, mark=square, mark options={solid, red}]
table[row sep=crcr]{%
	-40	1.57344152025599\\
	-38	2.04352764495599\\
	-36	2.57976124708396\\
	-34	3.17391828255195\\
	-32	3.80550486880665\\
	-30	4.4576897627951\\
	-28	5.08360627475987\\
	-26	5.65916623671945\\
	-24	6.17596201401609\\
	-22	6.62231154855901\\
	-20	6.99487607191236\\
};
\addlegendentry{Ref. 2}

\addplot [color=mycolor2, line width=1.0pt, mark size=5.0pt, mark=o, mark options={solid, mycolor2}]
table[row sep=crcr]{%
	-40	7.53614887914735\\
	-38	8.70395736386562\\
	-36	10.2609941335268\\
	-34	11.640839485947\\
	-32	12.8928387770552\\
	-30	14.1417495909804\\
	-28	15.4259278621438\\
	-26	16.6464246674387\\
	-24	17.7542249023608\\
	-22	18.8406674764854\\
	-20	19.945932984474\\
};
\addlegendentry{Two-tiers}

\addplot [color=mycolor3, line width=1.0pt, mark size=5.0pt, mark=asterisk, mark options={solid, mycolor3}]
table[row sep=crcr]{%
	-40	7.4977408158802\\
	-38	8.70141956453199\\
	-36	9.99541192034194\\
	-34	11.1565535979883\\
	-32	11.9974639613919\\
	-30	12.3369970079572\\
	-28	12.5046283458591\\
	-26	12.7437906707772\\
	-24	13.0893488665038\\
	-22	13.5344867282548\\
	-20	14.007100759704\\
};
\addlegendentry{Single-loop}

\addplot [color=mycolor4, line width=1.0pt, mark size=5.0pt, mark=triangle, mark options={solid, mycolor4}]
table[row sep=crcr]{%
	-40	6.24083894494675\\
	-38	7.25995588431797\\
	-36	8.28235695208496\\
	-34	9.28065759516915\\
	-32	10.1527827492813\\
	-30	10.6383765491205\\
	-28	11.222669766905\\
	-26	12.2820410556018\\
	-24	12.8428108830516\\
	-22	13.4034578828239\\
	-20	13.9188918802995\\
};
\addlegendentry{Benchmark}

\addplot [color=mycolor5, dashed, line width=1.0pt, mark size=5.0pt, mark=+, mark options={solid, mycolor5}]
table[row sep=crcr]{%
	-40	7.79587962427265\\
	-38	9.03731224908068\\
	-36	10.3257151797226\\
	-34	11.5772433825062\\
	-32	12.8715007639566\\
	-30	14.1202889831594\\
	-28	15.4419303817982\\
	-26	16.588048351661\\
	-24	17.7509537679283\\
	-22	18.918356352522\\
	-20	20.044536551357\\
};
\addlegendentry{C-Ref. 1}

\addplot [color=mycolor6, dashed, line width=1.0pt, mark size=5.0pt, mark=x, mark options={solid, mycolor6}]
table[row sep=crcr]{%
	-40	7.34051770221018\\
	-38	8.55545313077229\\
	-36	9.71026844287349\\
	-34	10.7701068265168\\
	-32	11.8959062361078\\
	-30	12.7271247223293\\
	-28	13.278819222568\\
	-26	13.6672084077134\\
	-24	14.0264643213076\\
	-22	14.3555659666085\\
	-20	14.6216607190758\\
};
\addlegendentry{C-Ref. 2}

\end{axis}
\end{tikzpicture}%
	\caption{Weighted secrecy sum-rate against the transmit power $P_{\max}$.}
	\label{fig:WSRvsPower}
\end{figure}

An interesting conclusion is further given by comparing the performance of convex reference algorithms C-Ref.~1 and C-Ref.~2 with the proposed algorithms. As we observe, the performance improvement achieved by replacing the unit-modulus constraints with enhanced convex ones is negligible. This is very significant for the two-tiers algorithm in which the algorithm performs almost identical in both the enhanced and the original settings\footnote{Note that for the single-loop setting, there is no guarantee that the algorithm outperforms in the enhanced setting. In fact, by merging multiple loops the algorithm may stick to a local minimum at very first iterations. This is also observed in the figure for small values of $P_{\max}$ where the algorithm in the original setting outperforms slightly the enhanced setting.}. These tight tracks of performance in both settings confirm the efficiency of the proposed algorithms.

We now set the transmit power to $\log P_{\max} = -30$ dB, and let the number of transmit antennas at the \ac{bs} vary between $4$ and $12$. The results are plotted in Fig.~\ref{fig:WSRvsTxAntenna128} where the weighted secrecy sum-rate is sketched against the number of \ac{bs} antennas. As the figure demonstrates, the deployment of an \ac{irs} significantly boosts the secrecy performance of the system. For instance, the secrecy sum-rate achieved in the reference scenario without an \ac{irs}, i.e., Ref. 1, with $M=12$ transmit antennas is achieved in the \ac{irs}-aided scenario with only $M=7$ antennas. Comparing the performance of the two tuning algorithm, it is observed from Fig.~\ref{fig:WSRvsTxAntenna128}, that the single-loop algorithm performs slightly degraded compared to the two-tiers algorithm. This degradation is due to the inner loop merging, and is the cost we pay to reduce the computational complexity.

\begin{figure}
	\centering
%
%
\definecolor{mycolor1}{rgb}{0.00000,0.44700,0.74100}%
\definecolor{mycolor2}{rgb}{0.63529,0.07843,0.18431}%
\definecolor{mycolor3}{rgb}{0.49400,0.18400,0.55600}%
\definecolor{mycolor4}{rgb}{0.00000,0.49804,0.00000}%
\begin{tikzpicture}
	
	\begin{axis}[%
		width=5in,
		height=3.5in,
		at={(1.23in,0.788in)},
		scale only axis,
		xmin=3.7,
		xmax=12.3,
		xlabel style={font=\color{white!15!black}},
		xlabel={$M$},
		xtick={4,6,8,10,12},
		xticklabels={{$4$},{$6$},{$8$},{$10$},{$12$}},
		ymin=-1,
		ymax=23.8,
		ylabel style={font=\color{white!15!black}},
		ylabel={$\bar{\mar}^\ssr$},
		ytick={0,5,10,15,20},
		yticklabels={{$0$},{$5$},{$10$},{$15$},{$20$}},
		axis background/.style={fill=white},
		legend style={at={(axis cs: 3.8,23.6)}, anchor=north west, legend cell align=left, align=left, draw=white!15!black}
		]
\addplot [color=mycolor1, line width=1.0pt, mark size=6pt, mark=diamond, mark options={solid, mycolor1}]
table[row sep=crcr]{%
	4	0\\
	5	0\\
	6	0.293183121969683\\
	7	2.18720341427167\\
	8	3.40328953890991\\
	9	4.61727975109414\\
	10	5.77059124566955\\
	11	7.07495357721171\\
	12	8.31538448429015\\
};
\addlegendentry{Ref. 1}

\addplot [color=red, line width=1.0pt, mark size=3.5pt, mark=square, mark options={solid, red}]
table[row sep=crcr]{%
	4	0\\
	5	0.123439922276303\\
	6	0.265996899934642\\
	7	2.35314532246794\\
	8	4.4576897627951\\
	9	6.47050081808669\\
	10	7.15438793279009\\
	11	8.67213981710908\\
	12	11.0674177158959\\
};
\addlegendentry{Ref. 2}

\addplot [color=mycolor2, line width=1.0pt, mark size=5.0pt, mark=o, mark options={solid, mycolor2}]
table[row sep=crcr]{%
	4	2.41959002976293\\
	5	3.08064983070215\\
	6	6.43740199909503\\
	7	10.2236178267266\\
	8	14.1417495909804\\
	9	17.4809872215278\\
	10	20.6355250668015\\
	11	22.0098962448412\\
	12	22.7857340906216\\
};
\addlegendentry{Two-tiers}

\addplot [color=mycolor3, line width=1.0pt, mark size=5.0pt, mark=asterisk, mark options={solid, mycolor3}]
table[row sep=crcr]{%
	4	2.11049766676484\\
	5	2.81783873213062\\
	6	5.72843598125373\\
	7	7.84087550412345\\
	8	12.3369970079572\\
	9	16.5234003552268\\
	10	19.9692048435472\\
	11	21.2304277383639\\
	12	21.9149418017554\\
};
\addlegendentry{Single-loop}
		
	\end{axis}
\end{tikzpicture}%
	\caption{Weighted secrecy sum-rate against the number of transmit antennas $M$.}
	\label{fig:WSRvsTxAntenna128}
\end{figure}
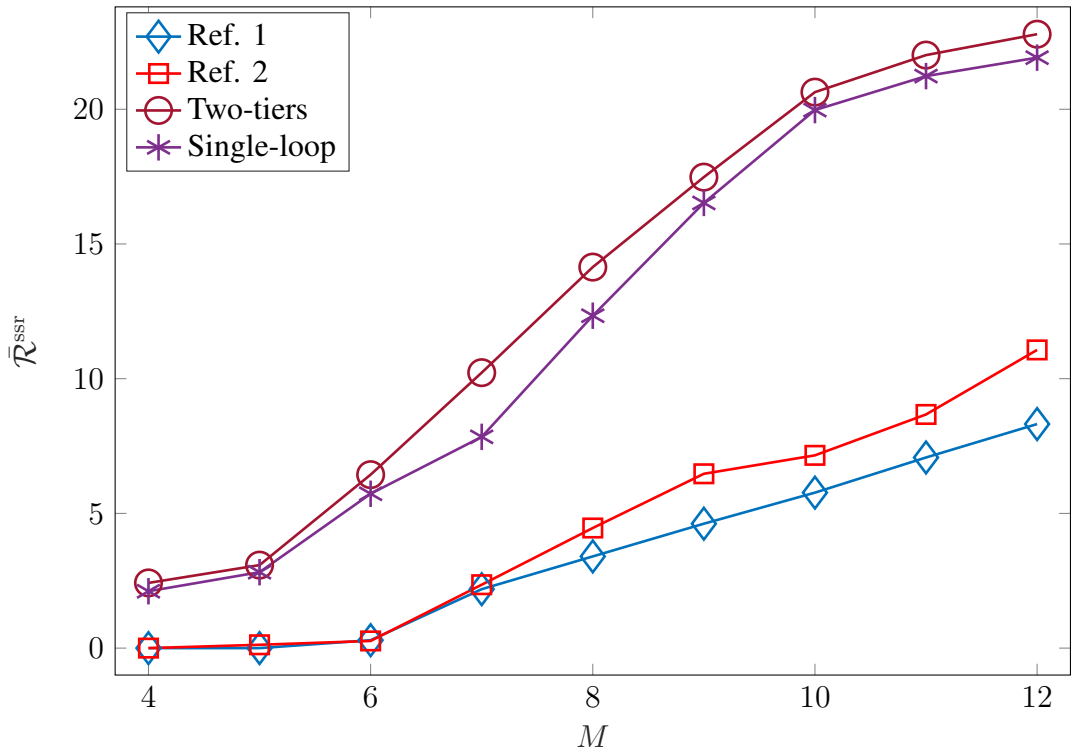

As the next experiment, we consider the setting in Fig.~\ref{fig:WSRvsTxAntenna128} and set the number of transmit antennas to $M=8$. We now vary the number of reflecting elements at the \ac{irs} between $N=16$ and $N=256$. Fig.~\ref{fig:WSRvsIRS} shows the secrecy sum-rate against the number of elements on the \ac{irs}. As the reference point, we have further plotted the result for Ref. 2 in which the \ac{irs} elements are tuned randomly. From the figure, one can observe that using either of the proposed algorithms, the achievable secrecy sum-rate increases \textit{linearly} in terms of the number of reflecting elements with the same slope of growth. For the case with random phase-shifts, the increase in the number of \ac{irs} elements does not lead to a perceptible growth in rate.

\begin{figure}
	\centering
%
%
\definecolor{mycolor1}{rgb}{0.00000,0.44700,0.74100}%
\definecolor{mycolor2}{rgb}{0.63529,0.07843,0.18431}%
\definecolor{mycolor3}{rgb}{0.49400,0.18400,0.55600}%
\definecolor{mycolor4}{rgb}{0.00000,0.49804,0.00000}%
\begin{tikzpicture}

\begin{axis}[%
	width=5in,
height=3.5in,
at={(1.23in,0.788in)},
scale only axis,
xmin=-3,
xmax=264,
xlabel style={font=\color{white!15!black}},
xlabel={$N$},
xtick={50,100,150,200,250},
xticklabels={{$50$},{$100$},{$150$},{$200$},{$250$}},
ymin=2.5,
ymax=21,
ylabel style={font=\color{white!15!black}},
ylabel={$\bar{\mar}^\ssr$},
ytick={5,10,15,20,25},
yticklabels={{$5$},{$10$},{$15$},{$20$},{$25$}},
axis background/.style={fill=white},
legend style={at={(axis cs: 8.5,20.5)}, anchor=north west, legend cell align=left, align=left, draw=white!15!black}
]

\addplot [color=red, line width=1.0pt, mark size=3.5pt, mark=o, mark options={solid, red}]
  table[row sep=crcr]{%
4	3.53319831362751\\
8	3.29011221902965\\
16	3.58960402950956\\
32	3.71346997931474\\
64	3.77431093139302\\
128	4.09386741944088\\
256	5.22189639271632\\
};
\addlegendentry{Ref. 2}

\addplot [ color=mycolor2, line width=1.0pt, mark=diamond,mark size=3.5pt, mark options={solid, mycolor2}]
table[row sep=crcr]{%
	4	4.89171411640431\\
	8	5.31240669582561\\
	16	6.46726043144662\\
	32	8.5626624220462\\
	64	10.8050140309566\\
	128	13.9208247469215\\
	256	16.9382582684522\\
};
\addlegendentry{Two-tiers}

\addplot [color=mycolor1, line width=1.0pt, mark=square,mark size=3.5pt, mark options={solid, mycolor1}]
table[row sep=crcr]{%
	4	4.70016104762992\\
	8	5.28834340610488\\
	16	6.34535553773002\\
	32	7.92130113208545\\
	64	8.5153956017558\\
	128	9.3144667612696\\
	256	9.96559609876578\\
};
\addlegendentry{Single-loop}

\end{axis}
\end{tikzpicture}%
	\caption{Weighted secrecy sum-rate against the number of IRS elements $N$.}
	\label{fig:WSRvsIRS}
\end{figure}
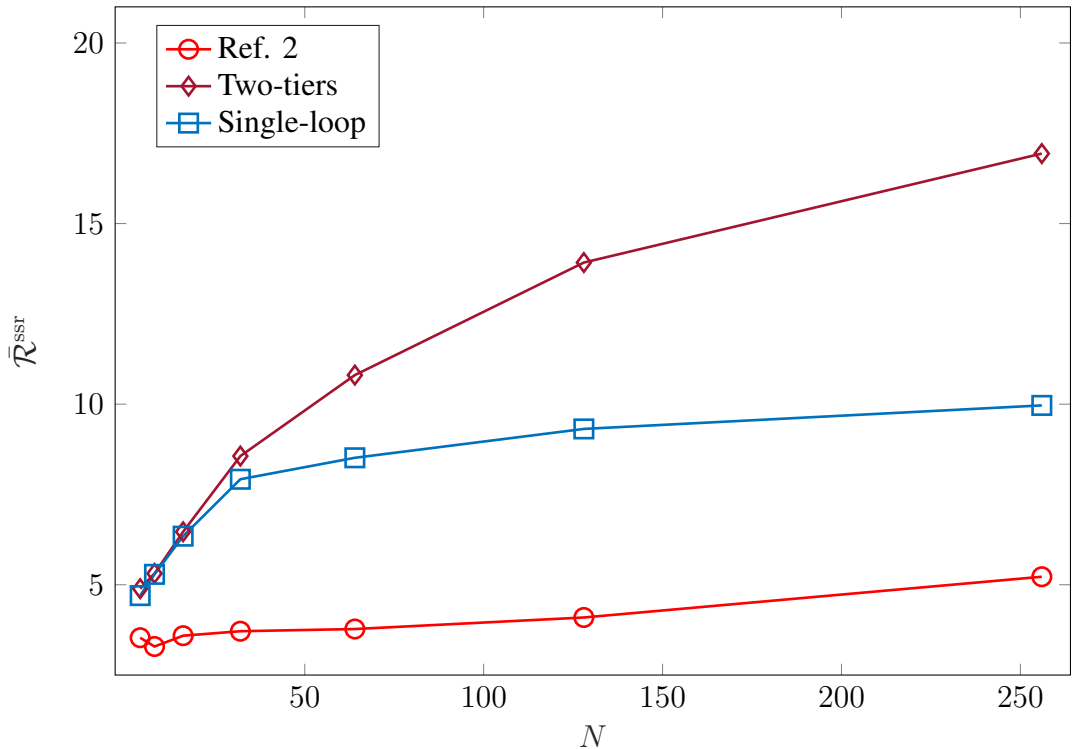

We now investigate the resistance of the proposed algorithms against the number of eavesdroppers in the system. To this end, we keep the settings as in Fig.~\ref{fig:WSRvsIRS} and set the number of elements on the \ac{irs} to $N=128$. The number of eavesdroppers is then varied from $J=1$ to $J=8$, and the secrecy sum-rate is plotted against $J$ in Fig.~\ref{fig:WSRvsEve}. For sake of comparison, we further plot the results for \ac{srzf} precoding, i.e., the reference scenarios. It is worth mentioning that from the literature, \ac{srzf} precoding is known to be highly robust against passive eavesdropping, due to the fact that it imposes explicitly the leakage suppression as a constraint in the precoder design \cite{asaad2019secure}. From Fig.~\ref{fig:WSRvsEve}, it is observed that the secrecy sum-rate drops for the both proposed algorithms in terms of $J$ similar to the \ac{srzf} precoding scheme. This identical behavior follows the fact that the proposed algorithms consider the secrecy rate as the objective function. This objective function implicitly constrains the information leakage in the design and leads to a behavior identical to \ac{srzf} precoding. The performance gain compared to the reference scenarios, which is observed in Fig.~\ref{fig:WSRvsEve}, is due to the efficient phase-shift tuning. 

\begin{figure}
	\centering
%
%
\definecolor{mycolor1}{rgb}{0.00000,0.44700,0.74100}%
\definecolor{mycolor2}{rgb}{0.85000,0.32500,0.09800}%
\definecolor{mycolor3}{rgb}{0.30100,0.74500,0.93300}%
\definecolor{mycolor4}{rgb}{0.92900,0.69400,0.12500}%
\pgfplotsset{custom/.style={thick}}
\begin{tikzpicture}

\begin{axis}[%
	width=5in,
height=3.5in,
at={(1.23in,0.788in)},
scale only axis,
xmin=-3,
xmax=264,
xlabel style={font=\color{white!15!black}},
xlabel={$N$},
xtick={50,100,150,200,250},
xticklabels={{$50$},{$100$},{$150$},{$200$},{$250$}},
ymin=2.5,
ymax=21,
ylabel style={font=\color{white!15!black}},
ylabel={$\bar{\mar}^\ssr$},
ytick={5,10,15,20,25},
yticklabels={{$5$},{$10$},{$15$},{$20$},{$25$}},
axis background/.style={fill=white},
legend style={at={(axis cs: 8.5,24.5)}, anchor=north west, legend cell align=left, align=left, draw=white!15!black}
]

\addplot [custom, color=mycolor2, line width=1.0pt, mark=diamond,mark size=3.5pt, mark options={solid, mycolor2}]
  table[row sep=crcr]{%
4	4.89171411640431\\
8	5.31240669582561\\
16	6.46726043144662\\
32	8.5626624220462\\
64	10.8050140309566\\
128	13.9208247469215\\
256	16.9382582684522\\
};
\label{t0}

\addplot [custom, color=mycolor3, dotted, line width=1.0pt, mark=diamond,mark size=3.5pt, mark options={solid, mycolor3}]
  table[row sep=crcr]{%
4	4.80101183375362\\
8	5.08534389390284\\
16	6.05316041723363\\
32	7.86661038461484\\
64	9.84585900136279\\
128	12.0994466219592\\
256	14.1625179123103\\
};
\label{t2}

\addplot [custom, color=mycolor1, dashdotted, line width=1.0pt, mark=diamond,mark size=3.5pt, mark options={solid, mycolor1}]
  table[row sep=crcr]{%
4	4.86551689075764\\
8	5.20828466682845\\
16	6.37617956393411\\
32	8.37105568926779\\
64	10.4458061966761\\
128	13.3809320650253\\
256	15.9000543384852\\
};
\label{t3}

\addplot [custom, color=mycolor4, dashed, line width=1.0pt, mark=diamond,mark size=3.5pt, mark options={solid, mycolor4}]
table[row sep=crcr]{%
	4	4.88187414014296\\
	8	5.2467034410378\\
	16	6.44168199410387\\
	32	8.5199461907813\\
	64	10.7172672349266\\
	128	13.8097268293104\\
	256	16.7089616619422\\
};
\label{t4}

\addplot [custom, color=mycolor1, line width=1.0pt, mark=square,mark size=3.5pt, mark options={solid, mycolor1}]
table[row sep=crcr]{%
	4	4.70016104762992\\
	8	5.28834340610488\\
	16	6.34535553773002\\
	32	7.92130113208545\\
	64	8.5153956017558\\
	128	9.3144667612696\\
	256	9.96559609876578\\
};
\label{single}

\addplot [custom, color=mycolor2, dashed, line width=1.0pt, mark=square,mark size=3.5pt, mark options={solid, mycolor2}]
  table[row sep=crcr]{%
4	4.68916285860309\\
8	5.22297308056513\\
16	6.32183995806275\\
32	7.74514012937501\\
64	8.51526873079045\\
128	9.10666173888474\\
256	9.40672638004299\\
};
\label{single4}


\end{axis}

\node [draw,fill=none] at ( 4.8,9.7) {\shortstack[l]{
		\ref{t0} Two-tiers \\
		\ref{t2} $B=2$ \\
		\ref{t3} $B=3$\\
	\ref{t4} $B=4$} };

\node [draw,fill=none] at (14,2.75) {\shortstack[l]{
		\ref{single} Single-loop \\
		\ref{single4} $B=4$}};

\end{tikzpicture}%
	\caption{Weighted secrecy sum-rate against $N$ for the ideal case without phase-shift quantization at the IRS, as well as cases with quantized phase-shifts.}
	\label{fig:Quantization}
\end{figure}
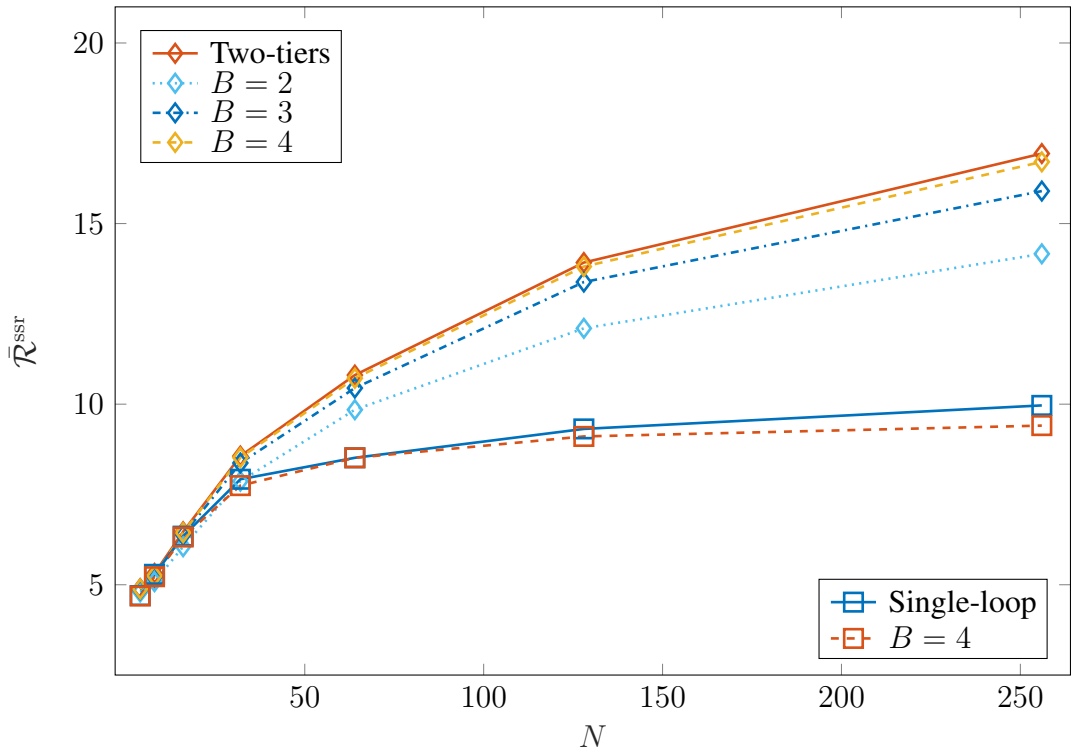

The proposed implementations for \acp{irs} suggest that the phase-shifts are realized discretely by a finite resolution. This means that the phase-shifts determined by the proposed algorithms are in practice quantized with a certain number of bits. Doing so, the proposed algorithms perform degraded due to the distortion introduced after phase-shift quantization. To investigate this performance degradation, we further plot the secrecy sum-rate against the number of \ac{irs} elements for the scenario of Fig.~\ref{fig:WSRvsIRS} in Fig.~\ref{fig:Quantization} considering phase-shift quantization with $B$ bits. As the figure shows, with a resolution of $B=4$, both algorithms perform close to the ideal case without quantization and the performance degradation is negligible.

\begin{figure}
	\centering
%
%
\definecolor{mycolor1}{rgb}{0.00000,0.44700,0.74100}%
\definecolor{mycolor2}{rgb}{0.63529,0.07843,0.18431}%
\definecolor{mycolor3}{rgb}{0.49400,0.18400,0.55600}%
\definecolor{mycolor4}{rgb}{0.00000,0.49804,0.00000}%
\begin{tikzpicture}

\begin{axis}[%
width=5in,
height=3.5in,
at={(1.23in,0.788in)},
scale only axis,
xmin=.7,
xmax=8.4,
xlabel style={font=\color{white!15!black}},
xlabel={$J$},
xtick={2,4,6,8},
xticklabels={{$2$},{$4$},{$6$},{$8$}},
ymin=-.5,
ymax=24,
ylabel style={font=\color{white!15!black}},
ylabel={$\bar{\mar}^\ssr$},
ytick={0,5,10,15,20},
yticklabels={{$0$},{$5$},{$10$},{$15$},{$20$}},
axis background/.style={fill=white},
legend style={legend cell align=left, align=left, draw=white!15!black}
]
\addplot [color=mycolor1, line width=1.0pt, mark size=6pt, mark=diamond, mark options={solid, mycolor1}]
  table[row sep=crcr]{%
1	8.84210765980666\\
2	7.74568978766127\\
3	6.60692399892306\\
4	5.27314031342619\\
5	4.67803152099868\\
6	3.31501912294984\\
7	1.86966904660648\\
8	0\\
};
\addlegendentry{Ref. 1}

\addplot [color=red, line width=1.0pt, mark size=3.5pt, mark=square, mark options={solid, red}]
  table[row sep=crcr]{%
1	11.8681762371373\\
2	11.147565261871\\
3	8.87017274897225\\
4	6.81661029508596\\
5	6.05136781894519\\
6	4.09386741944088\\
7	2.30921242952604\\
8	0\\
};
\addlegendentry{Ref. 2}

\addplot [color=mycolor2, line width=1.0pt, mark size=5.0pt, mark=o, mark options={solid, mycolor2}]
  table[row sep=crcr]{%
1	23.427031023959\\
2	22.6645706335714\\
3	21.830352418971\\
4	20.8293170748025\\
5	17.5977544142937\\
6	14.0321933801936\\
7	10.6882769329692\\
8	3.89766961261527\\
};
\addlegendentry{Two-tiers}

\addplot [color=mycolor3, line width=1.0pt, mark size=5.0pt, mark=asterisk, mark options={solid, mycolor3}]
  table[row sep=crcr]{%
1	18.2666353276817\\
2	17.3226214189362\\
3	16.6656897685536\\
4	15.4416122890159\\
5	12.6599744871086\\
6	9.19526456581823\\
7	5.55856183528224\\
8	1.31451067773366\\
};
\addlegendentry{Single-loop}

\end{axis}
\end{tikzpicture}%
	\caption{Weighted secrecy sum-rate against the number of eavesdroppers $J$.}
	\label{fig:WSRvsEve}
\end{figure}
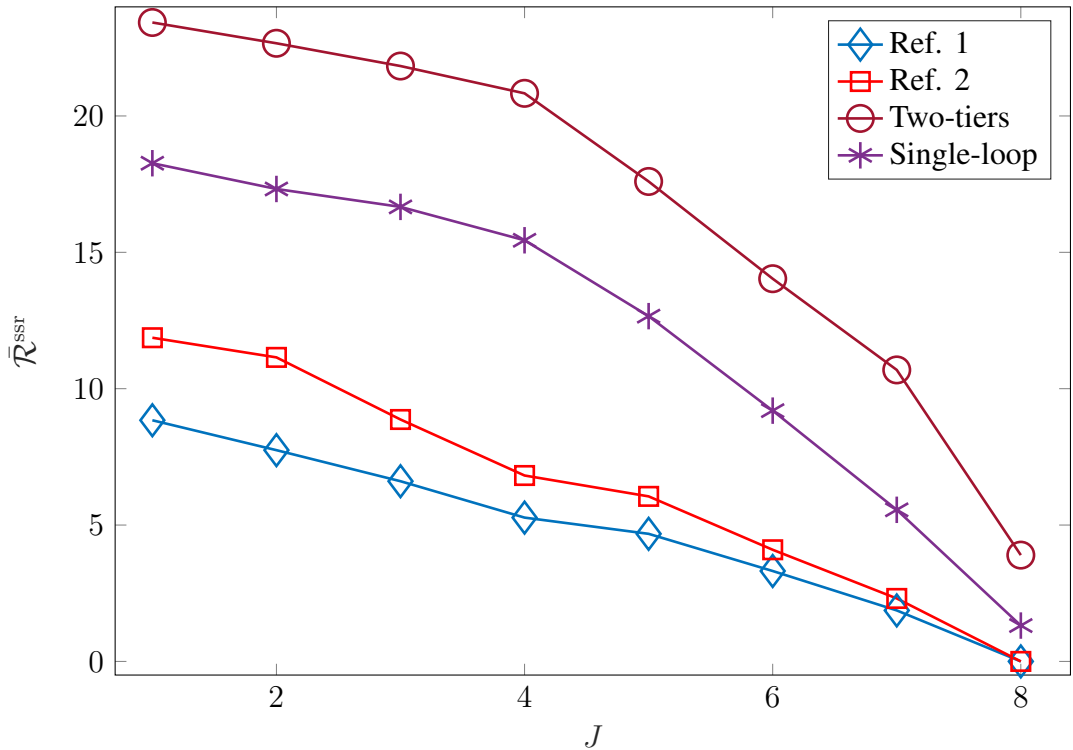

For the last experiment, we consider the setting in Fig.~\ref{fig:WSRvsEve}, and set $J=6$. The number of legitimate terminals is then increased from $K=1$ to $K=8$ and the secrecy sum-rate is plotted against $K$ in Fig.~\ref{fig:WSRvsUser}. As the figure shows, the secrecy sum-rate grows in terms of $K$. The order of growth however decreases at large choices of $K$ which is due to the higher multiuser interference in the system. Similar behavior is observed for the reference scenarios with a certain gap in the achievable rate which is due to the efficiency of phase-shift tuning in the proposed algorithms.

\begin{figure}
	\centering
%
%
\definecolor{mycolor1}{rgb}{0.00000,0.44700,0.74100}%
\definecolor{mycolor2}{rgb}{0.63529,0.07843,0.18431}%
\definecolor{mycolor3}{rgb}{0.49400,0.18400,0.55600}%
\definecolor{mycolor4}{rgb}{0.00000,0.49804,0.00000}%
\begin{tikzpicture}
	
	\begin{axis}[%
		width=5in,
		height=3.5in,
		at={(1.23in,0.788in)},
		scale only axis,
		xmin=.7,
		xmax=8.4,
		xlabel style={font=\color{white!15!black}},
		xlabel={$K$},
		xtick={2,4,6,8},
		xticklabels={{$2$},{$4$},{$6$},{$8$}},
		ymin=2.7,
		ymax=17,
		ylabel style={font=\color{white!15!black}},
		ylabel={$\bar{\mar}^\ssr$},
		ytick={0,4,8,12,16},
		yticklabels={{$0$},{$4$},{$8$},{$12$},{$16$}},
		axis background/.style={fill=white},
		legend style={at={(axis cs:6.5 ,13)}, anchor=north west,legend cell align=left, align=left, draw=white!15!black}
		]
		\addplot [color=mycolor1, line width=1.0pt, mark size=6pt, mark=diamond, mark options={solid, mycolor1}]
		table[row sep=crcr]{%
1	2.89035040144733\\
2	3.20456833368361\\
3	3.43667569499122\\
4	3.31501912294984\\
5	3.2458649722088\\
6	3.10829042185904\\
7	3.04620911199471\\
8	2.91323470402914\\
		};
		\addlegendentry{Ref. 1}
		
		\addplot [color=red, line width=1.0pt, mark size=3.5pt, mark=square, mark options={solid, red}]
		table[row sep=crcr]{%
1	3.60209604180118\\
2	4.00941442848999\\
3	4.19874627492817\\
4	4.09386741944088\\
5	4.07067238558674\\
6	3.76656907540057\\
7	3.63118370369795\\
8	3.54020670883184\\
		};
		\addlegendentry{Ref. 2}
		
		\addplot [color=mycolor2, line width=1.0pt, mark size=5.0pt, mark=o, mark options={solid, mycolor2}]
		table[row sep=crcr]{%
1	8.82092927884862\\
2	13.7689913613391\\
3	14.0059336947249\\
4	14.0321933801936\\
5	14.134105102991\\
6	14.0934356565752\\
7	14.2161523774518\\
8	14.2394858048155\\
		};
		\addlegendentry{Two-tiers}
		
		\addplot [color=mycolor3, line width=1.0pt, mark size=5.0pt, mark=asterisk, mark options={solid, mycolor3}]
		table[row sep=crcr]{%
1	5.47142215997541\\
2	9.03248975951454\\
3	9.17941060057993\\
4	9.19526456581823\\
5	8.61109492925788\\
6	8.87010191374552\\
7	8.4018949638208\\
8	8.24553970137482\\
		};
		\addlegendentry{Single-loop}
		
	\end{axis}
\end{tikzpicture}%
	\caption{Weighted secrecy sum-rate against the number of legitimate UTs $K$.}
	\label{fig:WSRvsUser}
\end{figure}
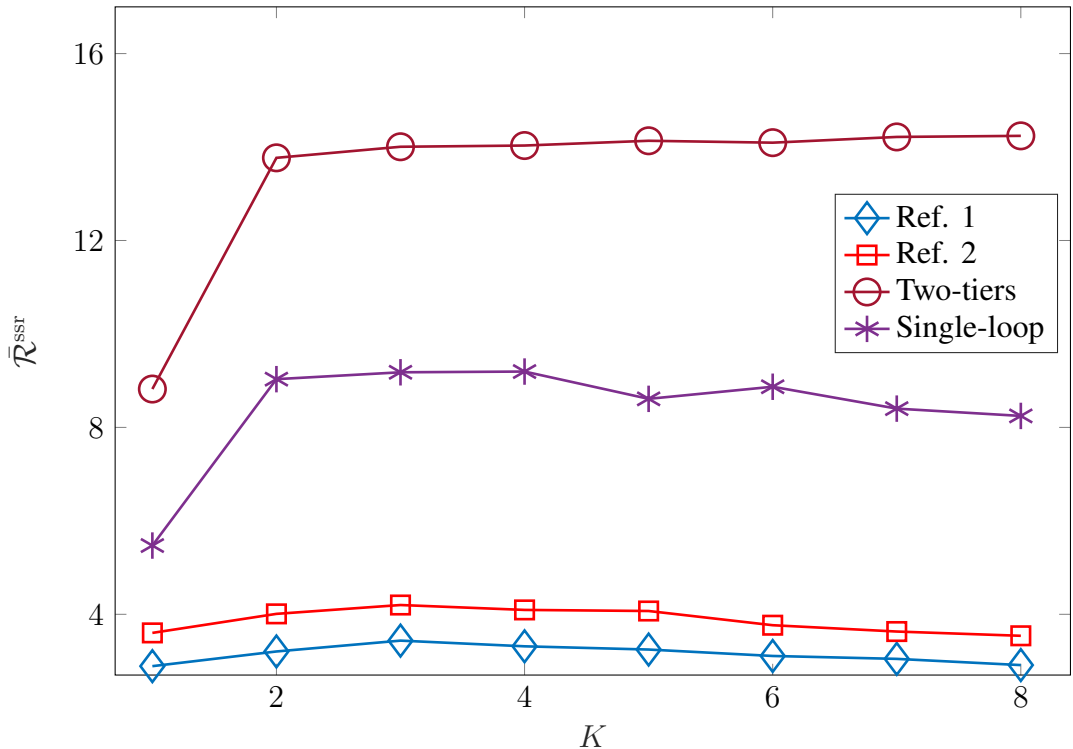

\subsection{Computational Complexity}
The computational complexity of both proposed algorithms shows quadratic growth in terms of the \ac{irs} dimension, i.e., $N$: In both algorithms, the dominant computational task is the update of \ac{irs} phase-shifts, i.e., $\bar{\bphi} = Q_{\rm MM} \brc{\bar{\bff},\bar{\bq},\bar{\mpsi},\mW_0} $. The computational complexity of this task for each $\phi_n$ is $\mathcal{O}\brc{N}$. As the update is performed entry-wise, we need $N$ iterations to determine the updated $\bar{\bphi}$. This means that the computational complexity of both algorithms grows $\mathcal{O}\brc{N^2}$ in $N$.

Fig.~\ref{fig:Complexity} plots the average runtime against the number of~\ac{irs} elements $N$ for the proposed algorithms considering~the~same scenario as the one considered in Figs.~\ref{fig:WSRvsIRS} and \ref{fig:Quantization}. As it shows, despite having the same order of complexity, the single-loop algorithm runs significantly faster than the~two-tiers~algorithm. This follows from the fact that both inner and outer loops of the two-tiers algorithm are merged in the~single-loop~algorithm. To have a quantitative comparison, we further use curve-fitting to model the runtime by the quadratic polynomial function 
\begin{align}
	Q_2\brc{N\vert a,b} = aN^2 + b.
\end{align}
 The fitted curves are further plotted in the figure. From the figure, it is observed that the single-loop approach approximately reduces the computational complexity by a factor of $100$. This factor is however fixed and does not scale with the system dimensions. Comparing Figs.~\ref{fig:Complexity} and \ref{fig:WSRvsIRS}, one can observe that this complexity reduction comes at the cost of performance degradation.

\begin{figure}
	\centering
	\input{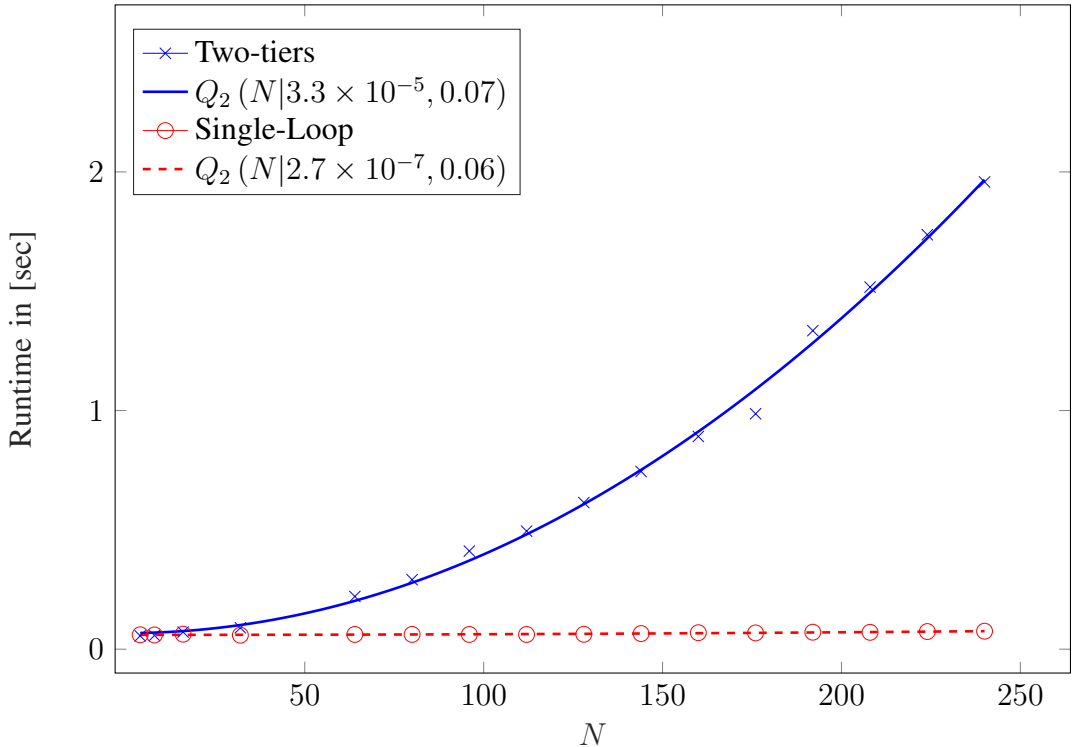}
	\caption{Runtime against the number of IRS elements $N$.}
	\label{fig:Complexity}
\end{figure}

\section{Conclusion}
\label{Sec:Conc}
In this work, two major iterative algorithms for secure precoding and phase-shift tuning in \ac{irs}-aided \ac{mimo} systems were proposed, and their applications were investigated through multiple numerical experiments. The proposed algorithms determine the system parameters in terms of the \ac{csi}. This means that the proposed algorithms are required to be run, once per a coherence time interval. From the computational viewpoint, this is a tractable task in many practical scenarios.

Since the proposed algorithms consider the weighted secrecy sum-rate as the objective, the information leakage suppression is implicitly applied. This is confirmed by the numerical investigations which show similar behavior as the one observed in \ac{srzf} precoding \cite{asaad2019secure}. The results further indicate that \acp{irs} can boost the secrecy performance of the system significantly. These two findings, along with the cost-efficiency of integrating \acp{irs}, suggest that \ac{irs} ia a good candidate for establishing secure communications in \ac{mimo} systems.

The results of this study are given under several idealistic assumptions and can be extended in various aspects. A natural direction is to extend the derivations to more realistic scenarios with imperfect \ac{csi} acquisition. In this respect, one can adopt the proposed scheme to more realistic settings, e.g., by considering the study in \cite{you2021reconfigurable}. Robustness of the proposed scheme against the unavailability of eavesdroppers' \ac{csi} is also an interesting direction for future work. Another line of work is to develop the proposed approach for settings in which the eavesdroppers increase their obtained leakage rate by actively contaminating the uplink pilots.

\appendices
\section{Proof of Theorem~\ref{theorem:1}}
\label{app:The1}
The proof follows two steps: In the first step, we show that for $\mb\in\dbc{0,1}^K$
\begin{align}
\max _{\mW, \bphi}  \mar^\ssr \brc{\mW, \bphi } \geq \max _{\mW, \bphi, \mb}\mar^\ssr_\rmq \brc{ \mW, \bphi, \mb }. \label{ineq1}
\end{align}
To show the validity of this inequality, we define the function 
\begin{align}
\hspace*{-1mm}	\mar^\ssr_\mt \brc{\mW, \bphi, \mb} \hspace*{-1mm}= \hspace*{-1.5mm} \sum_{k=1}^K \omega_k b_k \textbf{} \left[ \log\hspace*{-.5mm}\left( \dfrac{1\hspace*{-1mm}+\hspace*{-1mm}\sinr_k\brc{\mW, \bphi}}{1\hspace*{-1mm}+\hspace*{-1mm}\esnr_k\brc{\mW, \bphi}}\right)\right] ^+ \hspace*{-2mm}. \hspace*{-1mm} 
\end{align}
Since for $x\in \setR$ $\log x\leq \left[ \log x\right] ^+$, it is straightforward to write
\begin{align}
	\mar^\ssr_\rmq\brc{\mW, \bphi, \mb} \leq \mar^\ssr_\mt \brc{\mW, \bphi, \mb}. \label{eq:ineq1}
\end{align}

Now, we define the difference term 
\begin{subequations}
\begin{align}
	\Delta &= \mar^\ssr\brc{\mW, \bphi} - \mar^\ssr_\mt\brc{ \mW, \bphi, \mb}\\
	&= \sum_{k=1}^K \omega_k \brc{1-b_k} \mar_k^\rms \brc{\mW, \bphi}.
\end{align}
\end{subequations}
Noting that $\omega_k\geq 0$ and $b_k \in \dbc{0, 1} $, we further can write
\begin{align}
	\label{eq:ineq2}
	\Delta  \geq 0.
\end{align}

From \eqref{eq:ineq1} and \eqref{eq:ineq2}, we can write
\begin{align}
	\mar^\ssr_\rmq\brc{ \mW, \bphi, \mb} \leq \mar^\ssr_\mt \brc{ \mW, \bphi, \mb} \leq\mar^\ssr\brc{ \mW, \bphi},
\end{align}
which concludes that
\begin{align}
	\max _{\mW, \bphi, \mb} \mar^\ssr_\rmq\brc{ \mW, \bphi, \mb} \leq \max _{\mW, \bphi} \mar^\ssr \brc{\mW, \bphi}.
\end{align}

In the second step, we show that for $\mb\in\dbc{0,1}^K$
\begin{align}
	\max _{\mW, \bphi}  \mar^\ssr\brc{\mW, \bphi}\leq \max _{\mW, \bphi, \mb}  \mar^\ssr_\rmq\brc{ \mW, \bphi, \mb}. \label{ineq2}
\end{align}
To this end, let us define the $\setK \subseteq \dbc{K}$ as the subset of \acp{ut} for which we have
\begin{align}
	\sinr_k\brc{\mW, \bphi} \geq \esnr_k \brc{ \mW, \bphi}
\end{align}
when $\mW$ and $\bphi$ are set to the optimal values. We hence have
\begin{align}
	\max _{\mW, \bphi}  \mar^\ssr\brc{\mW, \bphi} 
	= \max_{\mW, \bphi} \sum_{k \in \setK} w_k \mar_k^\rms(\mW, \bphi).
\end{align}
Noting that for $ k \in \setK$, 
\begin{align}
\log\left( \dfrac{1+\sinr_k\brc{\mW, \bphi} }{ 1+\esnr_k \brc{ \mW, \bphi }}\right) \geq  0, 
\end{align}
we can write 
\begin{align}
\mar_k^\rms(\mW, \bphi) = \log\left( \dfrac{1+\sinr_k\brc{\mW, \bphi} }{ 1+\esnr_k \brc{ \mW, \bphi }}\right) 
\end{align}
and hence conclude that
\begin{align}
		\max _{\mW, \bphi}  \mar^\ssr\brc{\mW, \bphi} = 
	\max _{\mW, \bphi}  \mar^\ssr_\rmq\brc{ \mW, \bphi, \mb_\setK} \label{eq:eqrssm}
\end{align}
where $\mb_\setK = \dbc{b_{\setK 1}, \ldots, b_{\setK K} }^\trp$ with
\begin{align}
b_{\setK k} = \begin{cases}
	1 & k\in\setK\\
	0 & k\notin\setK
	\end{cases}.
\end{align}
Noting that  $\mb_{\setK k}\in \dbc{0,1}^K$, we have
\begin{align}
	\max _{\mW, \bphi, \mb}  \mar^\ssr_\rmq\brc{ \mW, \bphi, \mb}  \geq 
	\max _{\mW, \bphi}  \mar^\ssr_\rmq\brc{ \mW, \bphi, \mb_\setK}.
\end{align}
This concludes \eqref{ineq2}. Considering the inequalities in \eqref{ineq1} and \eqref{ineq2}, Theorem~\ref{theorem:1} is proved.

\section{Derivation of the MM-Based Algorithm}
\label{app:MM}
The \ac{mm} method approximates the solution with the limit of a sequence of feasible points. This sequence is derived by optimizing sequentially an objective function which \textit{majorizes}\footnote{The concept of majorization is defined in the sequel.} the original objective at the feasible point derived in the previous optimization. We illustrate the method through the derivations. More details on the \ac{mm} method can be followed in \cite{hunter2004tutorial,stoica2004cyclic,razaviyayn2013unified,song2015optimization,sun2017majorization}.

We start the derivations by rewriting the optimization problem $\hat{\maq}_2^{\rm B}$. Let $Q_2\brc{\bphi}$ denote the negative objective of $\hat{\maq}_2^{\rm B}$, i.e.,
\begin{align}
	Q_2\brc{\bphi} =  - Q_2^{\rm m} \brc{\bphi,\bar{\bff} } -  { \frac{\omega_k b_{0k} \brc{1  + \bar{\psi}_k} }{ B_{2Uk}^{\rm e} } } {A_{2k}^{ \rm e } \brc{\bphi}}.
\end{align}
By standard derivations, one can rewrite $Q_2\brc{\bphi}$ as 
\begin{align}
	Q_2\brc{\bphi} =  \bphi^\her \mQ \bphi + 2 \Re\set{ \bphi^\her \bvv } + C \label{eq:Q_2}
\end{align}
for some $C$ which is fixed in $\bphi$ and $\mQ$ and $\bvv$ as defined in Algorithm~\ref{alg:MM}. As the result, $\hat{\maq}_2^{\rm B}$ is rewritten as
\begin{subequations}
	\label{eq:Q222}
	\begin{align}
		&\min _{\bphi}  \bphi^\her \mQ \bphi + 2 \Re\set{ \bphi^\her \bvv } \\
		&\text{subject to } \abs{\phi_n} = 1, \forall n \in \dbc{ N }.
	\end{align}
\end{subequations}
In order to address this problem via the \ac{mm} method, we need to follow two steps:
\begin{enumerate}
\item Define a sequence of \textit{majorizations} for the objective in \eqref{eq:Q222}.
\item Determine the sequence of the solutions to the majorizations series.
\end{enumerate}
Before we start with the first step, let us define the concept of \textit{majorization}.
\begin{definition}[Majorization]
	Consider $f\brc{\cdot}: \setX \to \setR$. The function $m\brc{\cdot \vert \bxx_0}: \hat{\setX} \to \setR$ with $\bxx_0\in \setX$ and $\setX\subseteq \hat{\setX}$ majorizes $f\brc{\cdot}$ at $\bxx_0$, if the following constraints are satisfied:
	\begin{subequations}
		\begin{align}
			f\brc{\bxx} &\leq m\brc{\bxx \vert \bxx_0}, \forall \bxx \in \setX,\\
			f\brc{\bxx_0} &= m\brc{\bxx_0 \vert \bxx_0}.
		\end{align}
	\end{subequations}
\end{definition}

To find majorization functions of the objective in \eqref{eq:Q222}, we invoke the following lemma whose proof can be followed in \cite[Lemma~1]{song2015optimization}.
\begin{lemma}[Lemma~1 of \cite{song2015optimization}]
	\label{Lemma-Major}
	Let $\mQ \in \setC^{N \times N}$ and $\mM\in \setC^{N \times N}$ be Hermitian matrices satisfying $\mM\succeq \mQ$.  At any point $\bxx_0 \in \setC^{N}$, the quadratic function $f\brc{\bxx} = \bxx^\her \mQ \bxx$ is majorized by
	\begin{align}
		\Omega\brc{\bxx\vert\bxx_0} = \bxx^\her \mM \bxx +  2\left. \Re\set{ \bxx^\her \brc{ \mQ-\mM } \bxx_0 } \right. + \bxx_0^\her \brc{ \mM-\mQ } \bxx_0
	\end{align}
\end{lemma}

Let $\mM=\lambda_{\max} \mI_N $ in Lemma \ref{Lemma-Major}, where $\lambda_{\max}$ denotes the maximum eigenvalue of $\mQ$. Hence, at a feasible point $\bphi_0$, we can write
\begin{align}
	\bphi^\her \mQ \bphi + 2 \Re\set{ \bphi^\her \bvv } \leq \hat\Omega\brc{\bphi\vert \bphi_0},
\end{align}
where $\hat\Omega\brc{\bphi\vert \bphi_0} $ is given by
\begin{align}
	\hat\Omega\brc{\bphi\vert \bphi_0}  = \lambda_{\max} \brc{\norm{\bphi}^2 + \norm{ \bphi_0 }^2} +  2\left. \Re\set{ \bphi^\her  \brc{\mQ \bphi_0+\bvv - \lambda_{\max} \bphi_0 } } \right.  - \bphi_0^\her \mQ \bphi_0.
\end{align}
Noting that for any feasible point $\bphi$, we have $\norm{\bphi}^2 = N$, we can further write
\begin{align}
	\hat\Omega\brc{\bphi\vert \bphi_0}  = 2N \lambda_{\max} +  2\left. \Re\set{ \bphi^\her \brc{\mQ \bphi_0 +\bvv - \lambda_{\max} \bphi_0 } } \right.  - \bphi_0^\her \mQ \bphi_0.
\end{align}

Given the majorization function $\hat\Omega\brc{\bphi\vert \bphi_0} $, the \ac{mm} method starts from a feasible point ${\bphi^{ \brc{0} }}$ and constructs the sequence $\{\bphi^{ \brc{t} }\}$ for $t\in\setZ^+$ iteratively as 
\begin{subequations}
	\label{eq:phi_t1}
	\begin{align}
		\bphi^{ \brc{t+1} }  = &\argmin _{\bphi} \left. \hat\Omega\brc{\bphi\vert \bphi^{ \brc{t} } }\right. \\
		&\text{subject to }  \abs{\phi_n } =1,~\forall n \in \dbc{N}.
	\end{align}
\end{subequations}
The limit of this sequence, as the algorithm converges, is considered to be the solution. Unlike the original problem, the solution to \eqref{eq:phi_t1} is easily given by
\begin{align}
	\phi_n^{ \brc{t+1} }  = - \frac{\buu_n^\her \bphi^{ \brc{t} }  +\vv_n - \lambda_{\max} \phi_n^{ \brc{t} } }{ \abs{ \buu_n^\her \bphi^{ \brc{t} }  +\vv_n - \lambda_{\max} \phi_n^{ \brc{t} } } }
\end{align}
with $\buu_n^\her$ denoting the $n$-th row of the matrix $\mQ$. This concludes the derivation of Algorithm~\ref{alg:MM}.

\section{A BCD-Type Phase-Shift Update Algorithm}
\label{app:BCDPhi}
In this appendix, we tackle the non-convex unit modulus optimization $\hat{\maq}_2^{\rm B}$ via the \ac{bcd} method. Starting from an initial vector of phase-shifts, i.e., $\bphi^{\brc{0}}$, we alternately update each phase-shift in iteration $t$, i.e., each entry of $\bphi^{\brc{t}}$, by marginally optimizing it while treating the other phase-shifts as fixed variables whose values are calculated in iteration $t-1$. This means that in the $t$-th iteration, $\phi_{n}^{\brc{t}}$ is found as
\begin{align}
	\phi_{n}^{\brc{t}} = 	\argmin_{\abs{\varphi} = 1 } Q_2\brc{ \bxx^{\brc{t}}_n \brc{\varphi} } \label{BCS_phi}
\end{align}
where $Q_2\brc{\cdot}$ is given in \eqref{eq:Q_2}, and $\bxx^{\brc{t}}_n \brc{\cdot} $ is defined as
\begin{align}
	\bxx^{\brc{t}}_n \brc{\varphi} = \dbc{ \phi_{1}^{\brc{t-1}}, \ldots, \phi_{n-1}^{\brc{t-1}} , \varphi, \phi_{n+1}^{\brc{t-1}}, \ldots , \phi_{N}^{\brc{t-1}} }^\trp 
\end{align}
with $\phi_{n}^{\brc{t-1}}$ denoting the $n$-th entry of $\bphi^{\brc{t-1}}$. 

Defining the function $f^{\brc{t}} \brc{\varphi}= Q_2\brc{ \bxx^{\brc{t}}_n \brc{\varphi} }$,  we can use the Hermitian symmetry of $\mQ$ to show that
	\begin{align}
	f^{\brc{t}} \brc{\varphi} = \dbc{\mQ}_{nn} \abs{\varphi}^2 + 2\left. \Re \set{ a_n^{\brc{t}} \varphi^* } \right. +c_n^{\brc{t}} 
	\end{align}
	where $a_1^{\brc{t}}$ and $a_2^{\brc{t}}$ are given by
		\begin{align}
			a_n^{\brc{t}} &= \vv_n + \tilde{\buu}_n^\her \tilde{\bphi}_n^{\brc{t-1}} \label{eq:a_2}
		\end{align}
	with $\tilde{\buu}_n , \tilde{\bphi}_n^{\brc{t-1}} \in \setC^{N-1}$ being constructed from ${\buu}_n$ and ${\bphi}^{\brc{t-1}}$ by excluding their $n$-th entries, respectively, and $a_n^{\brc{t}}$ is a constant in terms of $\varphi$. Noting that $\abs{\varphi}=1$, \eqref{BCS_phi} reduces to
\begin{align}
	\phi_{n}^{\brc{t}} = 	\argmin_{\abs{\varphi} = 1 } \left. \Re \set{ a_n^{\brc{t}} \varphi^* } \right.
\end{align}
whose solution is given in a closed form as 
\begin{align}
	\phi_{n}^{\brc{t}} = - \frac{a_n^{\brc{t}} }{ \abs{a_n^{\brc{t}}} }. \label{eq:phi_t}
\end{align}
After multiple iterations, the solution of $\hat{\maq}_2^{\rm B}$ is approximated with the converged phase-shifts. The final algorithm is summarized in Algorithm~\ref{Alg:BCD-Method}.
	\begin{algorithm} 
	\caption{Phase-Shift Update via the \ac{bcd} Method}
\label{Alg:BCD-Method}
	\begin{algorithmic}[1]
		\INPUT{$\bar{\bff}$, $\bar{\bq}$, $\bar{\mpsi}$ and $\mW_0$}
\REQUIRE Set a feasible initial point $\bphi = \bphi^{(0)}$, and define
\begin{align*}
	\mar_{\rm M}^\ssr \brc{\bphi} =  Q_2^{\rm m} \brc{\bphi,\bar{\bff} }  + { \frac{\omega_k b_{0k} \brc{1 +  \bar{\psi}_k} }{ B_{2Uk}^{\rm e} } } {A_{2k}^{ \rm e } \brc{\bphi}} 
\end{align*}
\IF{$\mar_{\rm M}^\ssr \brc{\bphi} $ has not converged}{
	\STATE Calculate $\mU$ and $\bvv$ from \eqref{eq:Uv}
			\STATE Update $a_n^{\brc{t}}$ via \eqref{eq:a_2} for $n\in \dbc{N}$
			\STATE Update ${\phi}_n^{\brc{t}}$ via \eqref{eq:phi_t} for $n\in \dbc{N}$
			\STATE Set $\bphi = \bphi^{(t)}$, $t \leftarrow t+1$ and go back to line 1
		}
		\ENDIF
	\end{algorithmic} 
\end{algorithm}

Algorithm~\ref{Alg:BCD-Method} is compared with Algorithm~\ref{alg:MM} by performing a simple numerical experiment. This experiment considers a setting consistent to the one presented in Fig.~\ref{fig:setting}. In this setting, we set the number of eavesdroppers to $J=6$, the number of legitimate \acp{ut} to $K=4$, the number of \ac{bs} antennas to $M=8$ and the number of \ac{irs} elements to $N=16$. The transmit power is further set to $\log P_{\max} = 0$ dB and the channel parameters are set as described in Section~\ref{sec:Numerical}. For this setting, we apply the two-tiers algorithm, i.e., Algorithm~\ref{alg:2Tier} two times: once the inner loop, i.e., Algorithm~\ref{Alg:2}, is performed via the \ac{mm} algorithm and once with Algorithm~\ref{Alg:BCD-Method}. Fig.~\ref{fig:MMvsBCD1} shows the achievable secrecy sum-rate against the number of iterations in the main loop of Algorithm~\ref{alg:2Tier}. As the figure shows, the algorithms converge to the target rate almost identical with respect to the number of iterations. To compare the computational complexity of these algorithms, we further plot the achievable secrecy sum-rate against the runtime for these two algorithms. The result is shown in Fig.~\ref{fig:MMvsBCD2}. As the figure shows, using Algorithm~\ref{Alg:BCD-Method} leads  to a slightly faster convergence compared to the \ac{mm} algorithm.

 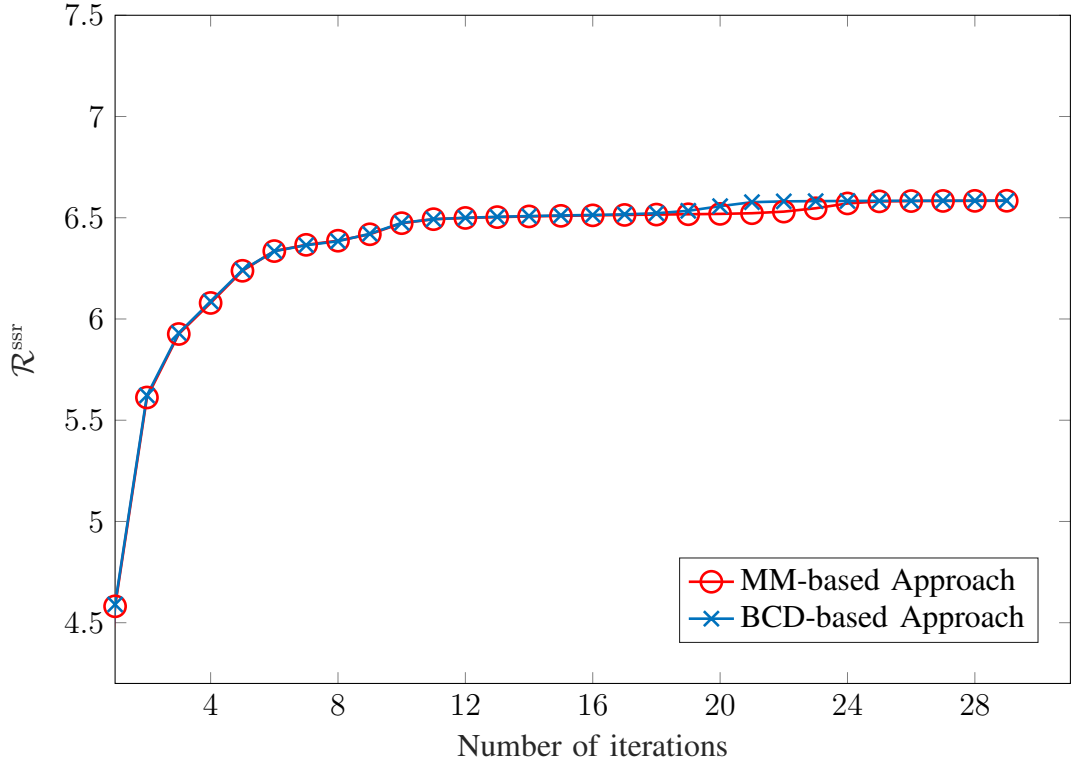
\begin{figure}
 	\centering
%
%
\definecolor{mycolor1}{rgb}{0.00000,0.44700,0.74100}%
\definecolor{mycolor2}{rgb}{0.85000,0.32500,0.09800}%
\definecolor{mycolor3}{rgb}{0.92900,0.69400,0.12500}%
\definecolor{mycolor4}{rgb}{0.49400,0.18400,0.55600}%
\begin{tikzpicture}
	
	\begin{axis}[%
		width=5in,
		height=3.5in,
		at={(1.23in,0.788in)},
		scale only axis,
xmin=1,
xmax=31,
xtick={{4},{8},{12},{16},{20},{24},{28}},
xticklabels={{$4$},{$8$},{$12$},{$16$},{$20$},{$24$},{$28$}},
xlabel style={font=\color{white!15!black}},
xlabel={{Number of iterations}},
ymin=4.2,
ymax=7.5,
ylabel style={font=\color{white!15!black}},
ylabel={$\mar^\ssr$},
axis background/.style={fill=white},
title style={font=\bfseries},
legend style={at={(axis cs:30,4.4)},anchor=south east,legend cell align=left, align=left, draw=white!15!black}
]
	\addplot [color=red, line width=1.0pt, mark size=4.0pt, mark=o, mark options={solid, red}]
	table[row sep=crcr]{%
		0	2.11331130475329\\
		1	4.58095335784947\\
		2	5.61231832863635\\
		3	5.92577059939003\\
		4	6.07914215638967\\
		5	6.23801324522637\\
		6	6.3353707629698\\
		7	6.36621517488945\\
		8	6.38653639741482\\
		9	6.41832095290117\\
		10	6.47315513725057\\
		11	6.49293724551668\\
		12	6.49900476619802\\
		13	6.50313203498932\\
		14	6.50650562537983\\
		15	6.50934668383735\\
		16	6.51177273394998\\
		17	6.51386076190504\\
		18	6.51566626850017\\
		19	6.51723163904855\\
		20	6.51901611283648\\
		21	6.5222998341499\\
		22	6.53021531226277\\
		23	6.54680802124206\\
		24	6.57064721217899\\
		25	6.5808360081241\\
		26	6.58233730325819\\
		27	6.58298666600155\\
		28	6.58344586770941\\
		29	6.58379931077801\\
	};
	\addlegendentry{MM-based Approach}
	
	\addplot [color=mycolor1, line width=1.0pt, mark size=4.0pt, mark=x, mark options={solid, mycolor1}]
	table[row sep=crcr]{%
		0	2.11331130475329\\
		1	4.59020075632412\\
		2	5.62255397101956\\
		3	5.93065399742702\\
		4	6.08672663995888\\
		5	6.24149196728257\\
		6	6.33479348171546\\
		7	6.36440741302403\\
		8	6.38410292145224\\
		9	6.42081651694658\\
		10	6.47491483709608\\
		11	6.49272629951664\\
		12	6.49971649364238\\
		13	6.50458220459505\\
		14	6.50828689865525\\
		15	6.51124537470364\\
		16	6.51370399425111\\
		17	6.51668016500105\\
		18	6.52225220941517\\
		19	6.53419706528051\\
		20	6.55712730941939\\
		21	6.57689087019316\\
		22	6.58085130929768\\
		23	6.58201490655863\\
		24	6.58279701488112\\
		25	6.58338843544051\\
		26	6.58385021384363\\
		27	6.58421691572448\\
		28	6.58451426107853\\
		29	6.58476255684789\\
	};
	\addlegendentry{BCD-based Approach}
	\end{axis}
\end{tikzpicture}%
 	\caption{Comparing the BCD-based phase tuning approach to the one which performs the \ac{mm} algorithm.}
 	\label{fig:MMvsBCD1}
 \end{figure}

 \begin{figure}
	\centering
%
%
\definecolor{mycolor1}{rgb}{0.00000,0.44700,0.74100}%
\definecolor{mycolor2}{rgb}{0.85000,0.32500,0.09800}%
\definecolor{mycolor3}{rgb}{0.92900,0.69400,0.12500}%
\definecolor{mycolor4}{rgb}{0.49400,0.18400,0.55600}%
\begin{tikzpicture}
	
	\begin{axis}[%
		width=5in,
		height=3.5in,
		at={(1.23in,0.788in)},
		scale only axis,
xmin=0.05,
xmax=4.45,
xlabel style={font=\color{white!15!black}},
xlabel={{Runtime} in [sec]},
ymin=1.5,
ymax=7.5,
ylabel style={font=\color{white!15!black}},
		ylabel={$\mar^\ssr$},
		axis background/.style={fill=white},
		title style={font=\bfseries},
		legend style={at={(axis cs:4.35,1.7)},anchor=south east,legend cell align=left, align=left, draw=white!15!black}
		]
		\addplot [color=red, line width=1.0pt, mark size=4.0pt, mark=o, mark options={solid, red}]
		table[row sep=crcr]{%
0	2.11331130475329\\
0.3542477536	4.58095335784947\\
0.59805547975	5.61231832863635\\
0.8094281414	5.92577059939003\\
1.01538856545	6.07914215638967\\
1.214494474	6.23801324522637\\
1.39503503195	6.3353707629698\\
1.55146685165	6.36621517488945\\
1.6950346805	6.38653639741482\\
1.8385540748	6.41832095290117\\
1.9808332225	6.47315513725057\\
2.1260642542	6.49293724551668\\
2.26240487855	6.49900476619802\\
2.41527203035	6.50313203498932\\
2.57710240835	6.50650562537983\\
2.73270627205	6.50934668383735\\
2.880245986	6.51177273394998\\
3.009911331	6.51386076190504\\
3.14210466245	6.51566626850017\\
3.2774896014	6.51723163904855\\
3.39827961025	6.51901611283648\\
3.5265785305	6.5222998341499\\
3.6580792923	6.53021531226277\\
3.7821603765	6.54680802124206\\
3.89382125315	6.57064721217899\\
4.0089571471	6.5808360081241\\
4.12091272645	6.58233730325819\\
4.2374459656	6.58298666600155\\
4.34966470985	6.58344586770941\\
4.45762585255	6.58379931077801\\
};
		\addlegendentry{MM-based Approach}
		
		\addplot [color=mycolor1, line width=1.0pt, mark size=4.0pt, mark=x, mark options={solid, mycolor1}]
		table[row sep=crcr]{%
0	2.11331130475329\\
0.28010623865	4.59020075632412\\
0.47359311345	5.62255397101956\\
0.66192386545	5.93065399742702\\
0.8650416809	6.08672663995888\\
1.0602307252	6.24149196728257\\
1.22580244045	6.33479348171546\\
1.40544512415	6.36440741302403\\
1.58107760405	6.38410292145224\\
1.73733115175	6.42081651694658\\
1.88736708485	6.47491483709608\\
2.0285839759	6.49272629951664\\
2.17246076085	6.49971649364238\\
2.3024051364	6.50458220459505\\
2.4287402881	6.50828689865525\\
2.5591480588	6.51124537470364\\
2.6867721715	6.51370399425111\\
2.8185811863	6.51668016500105\\
2.9414962476	6.52225220941517\\
3.0639656724	6.53419706528051\\
3.1933784338	6.55712730941939\\
3.3307594321	6.57689087019316\\
3.45747560665	6.58085130929768\\
3.5840415868	6.58201490655863\\
3.70973013165	6.58279701488112\\
3.8203145755	6.58338843544051\\
3.93112812355	6.58385021384363\\
4.04320410915	6.58421691572448\\
4.16077342655	6.58451426107853\\
4.27302587535	6.58476255684789\\
};
		\addlegendentry{BCD-based Approach}
	\end{axis}
\end{tikzpicture}%
 	\caption{Comparing the complexity of the  BCD-based and MM-based phase tuning approaches.}
 	\label{fig:MMvsBCD2}
\end{figure}
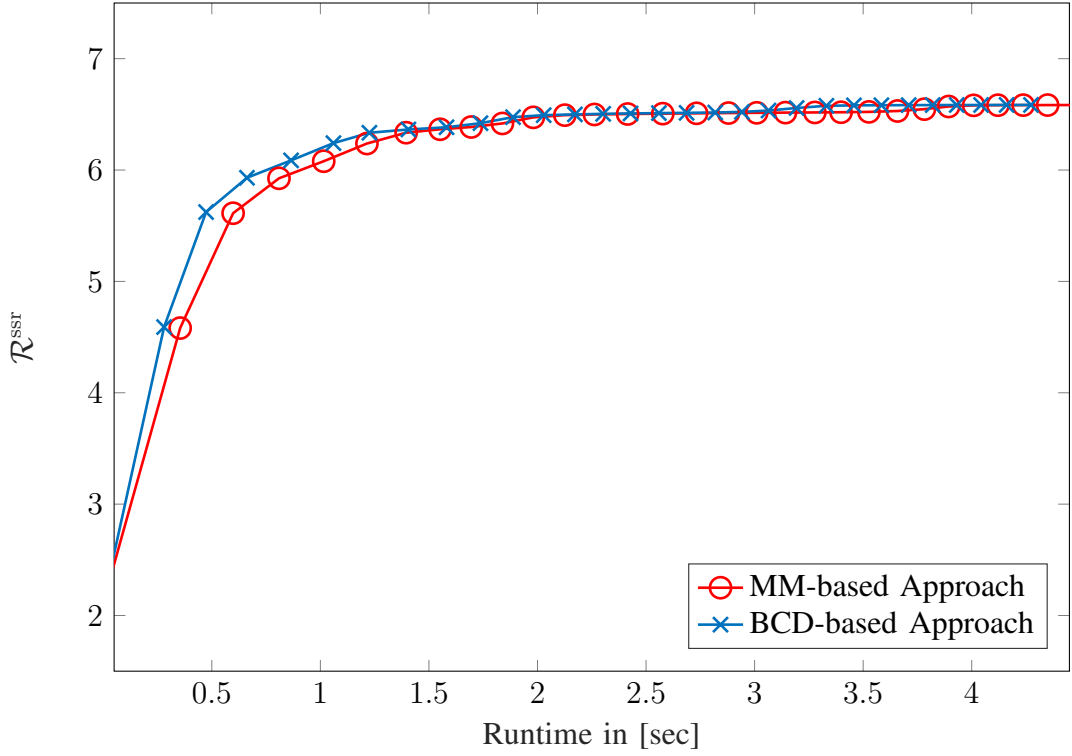

\section{Proof of Lemma~\ref{lem:Conv_1}}
\label{app:Proof_Conv_1}
Following the notation of Algorithm~\ref{Alg:1}, let $\bphi_0 = \bar{\bphi}$ and $\mb_0 = \bar{\mb}$ denote the vector of phase-shifts and auxiliary variables fixed as inputs of the algorithm, and let $\bar{\mW}$ and $\mW_{\rm new}$ represent the precoding matrix updated at the end of iteration $t$ and $t+1$, respectively. We define the following functions:
\begin{subequations}
	\begin{align}
		f_Q\brc{ \mbeta , \mgamma, \mW } &= Q_1^\rmm \brc{ \mW , \mbeta } + Q_1^\ee \brc{ \mW , \mgamma } + \log B_{1U}^\ee\\
		\mar_{\rm M}^{\ssr} \brc{\mW} &= \mar_\rmq^{\ssr} \brc{ \mW , \bar{\bphi} , \bar{\mb} }
	\end{align}
\end{subequations}
with $Q_1^\rmm \brc{ \mW , \mbeta }$,  $Q_1^\ee \brc{ \mW , \mgamma }$, and $B_{1U}^\ee$ as defined in Section~\ref{sec:First_Mar}. Moreover, let us define vectors 
\begin{subequations}
	\begin{align}
		\mbeta^{\rmF} \brc{\mW} &= \dbc{ \beta_1^\rmF \brc{\mW} , \ldots ,  \beta_K^\rmF \brc{\mW} },\\
		\mgamma^{\rmF} \brc{\mW} &= \dbc{ \gamma_1^\rmF \brc{\mW} , \ldots ,  \gamma_K^\rmF \brc{\mW} }
	\end{align}
\end{subequations}
for a given precoding matrix $\mW$, where 
\begin{subequations}
	\begin{align}
		\beta_k^\rmF \brc\mW &= \frac{ \displaystyle \sqrt{\omega_k \bar{b}_{k} \brc{1+ \frac{A_{1k}^{ \rm m } \brc{{\mW}}}{B_{1k}^{ \rm m } \brc{ {\mW} }} } } \left. \tilde{\mh}^\her_k \brc{\bar\bphi} \bw_k \right. }{
		A_{1k}^{ \rm m } \brc{\mW} + B_{1k}^{ \rm m }  \brc{\mW}
	},\\
		\gamma_k^\rmF \brc\mW &= \sqrt{ \frac{\displaystyle\omega_k \bar{b}_{k} \brc{1+ \frac{A_{1k}^{ \rm e } \brc{ {\mW} }}{B_{1k}^{ \rm e } \brc{ {\mW} }} } }{ B_{1U}^{\rm e} } } \sqrt{A_{1k}^{ \rm e } \brc{\mW}} .
	\end{align}
\end{subequations}
%

By substitution, it is straightforward to show that in iteration $t$ of Algorithm~\ref{Alg:1}, $\bar{\mbeta} = \mbeta^{\rmF} \brc{\bar\mW}$ and $\bar{\mgamma} = \mgamma^{\rmF} \brc{\bar\mW}$ and that 
\begin{align}
\mar_{\rm M}^{\ssr} \brc{\mW} = f_Q\brc{ \mbeta^\rmF \brc\mW , \mgamma^\rmF \brc\mW, \mW } . \label{Proom_Lem_1_0}
\end{align}

To start the proof, we note that for a fixed $\mW$, the marginal function $f_Q\brc{ \mbeta , \mgamma, \mW }$ is a concave function in terms of $\mbeta$ and $\mgamma$ whose maximum is at ${\mbeta} = \mbeta^{\rmF} \brc{\mW}$ and ${\mgamma} = \mgamma^{\rmF} \brc{\mW}$. As a result, for a given $\mW$, we have
\begin{subequations}
	\begin{align}
		f_Q\brc{ \bar\mbeta , \bar\mgamma, \mW } &\leq f_L\brc{ \mbeta^\rmF \brc\mW , \mgamma^\rmF \brc\mW, \mW }\\
		&= \mar_{\rm M}^{\ssr} \brc{\mW} . \label{Proom_Lem_1_2}
	\end{align}
\end{subequations}
Furthermore, the marginal function $f_Q\brc{ \bar\mbeta , \bar\mgamma, \mW }$ is a concave function in $\mW$. The quadratic problem $\maq_1^B$ indicates that the maximum of this marginal function is at $\mW =\mW_{\rm new}$, when we set $\mbeta=\bar{\mbeta}$ and $\mgamma=\bar{\mgamma}$. Thus, for a given $\mW$, we have
\begin{align}
	f_Q\brc{ \bar\mbeta , \bar\mgamma, \mW } &\leq f_Q\brc{ \bar\mbeta , \bar\mgamma, \mW_{\rm new} }. \label{Proom_Lem_1_3}
\end{align}

Considering the above inequalities, we can write
\begin{subequations}
	\begin{align}
	\mar_\rmq^{\ssr} \brc{ \bar\mW , \bar{\bphi} , \bar{\mb} } &= \mar_{\rm M}^{\ssr} \brc{\bar\mW}\\
	&\stackrel{\dagger}{=} f_Q\brc{ \bar\mbeta , \bar\mgamma, \bar\mW }\\
	&\stackrel{\star}{\leq} f_Q\brc{ \bar\mbeta , \bar\mgamma, \mW_{\rm new} }\\
	&\stackrel{\clubsuit}{\leq} f_Q\brc{  \mbeta^\rmF \brc{\mW_{\rm new} }, \mgamma^\rmF \brc{\mW_{\rm new} }, \mW_{\rm new} }\\
	&= \mar_{\rm M}^{\ssr} \brc{\bar\mW_{\rm new}}\\
	&= \mar_\rmq^{\ssr} \brc{ \bar\mW_{\rm new} , \bar{\bphi} , \bar{\mb} }
\end{align}
\end{subequations}
where $\dagger$ follows \eqref{Proom_Lem_1_0}, $\star$ comes from \eqref{Proom_Lem_1_3}, and $\clubsuit$ is concluded by \eqref{Proom_Lem_1_2}. 

\section{Proof of Lemma~\ref{lem:Conv_2}}
\label{app:Proof_Conv_2}
The proof follows the same approach as in Appendix~\ref{app:Proof_Conv_1}. Assuming the initialization $\mW_0 = \bar{\mW}$ and $\mb_0 = \bar{\mb}$, we start the proof by defining 
\begin{subequations}
	\begin{align}
		f_Q\brc{ \bff , \bpi, \bphi } &= Q_2^\rmm \brc{ \bphi , \bff } + Q_2^\ee \brc{ \bphi , \bpi } + \sum_{k=1}^K \omega_k \bar{b}_k \log B_{2Uk}^\ee ,\\
		\mar_{\rm M}^{\ssr} \brc{\bphi} &= \mar_\rmq^{\ssr} \brc{ \bar{\mW} , \bphi , \bar{\mb} }
	\end{align}
\end{subequations}
with $Q_2^\rmm \brc{ \bphi , \bff }$,  $Q_2^\ee \brc{ \bphi , \bpi }$ and $B_{2Uk}^\ee$ defined in Section~\ref{sec:Second_Mar}. We further define 
\begin{subequations}
	\begin{align}
		\bff^{\rmF} \brc{\bphi} &= \dbc{ f_1^\rmF \brc{\bphi} , \ldots ,  f_K^\rmF \brc{\bphi}  },\\
		\bpi^{\rmF} \brc{\bphi} &= \dbc{ \varpi_1^\rmF \brc{\bphi}, \ldots ,  \varpi_K^\rmF \brc{\bphi}  }
	\end{align}
\end{subequations}
for a given $\bphi$, where 
\begin{subequations}
	\begin{align}
		f_k^\rmF \brc\bphi &= \frac{ 
		\displaystyle
		\sqrt{\omega_k \bar{b}_{k} \brc{1 + 
				\frac{A_{2k}^{ \rm m } \brc{ \bphi }}{B_{2k}^{ \rm m } \brc{  \bphi }}
			} 
		} \brc{ \mh^\her_{\rd, k}\bar{\bw}_{k}  + \bphi^\her \mH^\her_k \bar{\bw}_{k} }
	}{
	\norm{ \mh^\her_{\rd, k}\bar{\mW} + \bphi^\her \mH^\her_k \bar{\mW} }^2 + \sigma_k^2
		},\\
		\varpi_k^\rmF \brc\bphi &= 
		\sqrt{
			\frac{ \displaystyle \omega_k \bar{b}_{k} \brc{1+
					\frac{A_{2k}^{ \rm e } \brc{ \bphi }}{B_{2k}^{ \rm e } \brc{ \bphi }}
				} }{ B_{2Uk}^{\rm e} } } 
			\sqrt{A_{2k}^{ \rm e } \brc{\bphi}
		}.
	\end{align}
\end{subequations}
By substitution, it is easily shown that for a given $\bphi$
\begin{align}
\mar_{\rm M}^{\ssr} \brc{\bphi}  =	f_Q\brc{ \bff^{\rmF} \brc{{\bphi}}, \bpi^{\rmF} \brc{{\bphi}} , \bphi }.
\end{align}

Let $\bar{\bphi}$ and $\bphi_{\rm new}$ indicate the phase-shift vectors updated at the end of iteration $t$ and $t+1$, respectively. The convergence of the \ac{mm} algorithm, i.e., Algorithm~\ref{alg:MM}, guarantees
\begin{align}
	f_Q\brc{ \bff^{\rmF} \brc{\bar{\bphi}}, \bpi^{\rmF} \brc{\bar{\bphi}} , \bar{\bphi} } \leq f_Q\brc{ \bff^{\rmF} \brc{{\bphi}}, \bpi^{\rmF} \brc{{\bphi}} , \bphi_{\rm new} }. \label{Proof_Lem_2_1}
\end{align}
This inequality holds as an identity, if Algorithm~\ref{alg:MM} iterates only for one iteration. Following concavity of the marginal function $f_Q\brc{ \bff , \bpi, \bphi }$ in terms of $\bff$ and $\bpi$ with maximum being at $\bff^{\rmF} \brc{\bphi}$ and $\bpi^{\rmF} \brc{\bphi}$, for a given $\bphi$, we can write that
\begin{align}
	f_Q\brc{ \bff^{\rmF} \brc{\bar{\bphi}}, \bpi^{\rmF} \brc{\bar{\bphi}} , \bphi } \leq f_Q\brc{ \bff^{\rmF} \brc{{\bphi}}, \bpi^{\rmF} \brc{{\bphi}} , \bphi }
\end{align}
for any feasible $\bphi$.

We hence can use the above inequalities and write
\begin{subequations}
	\begin{align}
		\mar_\rmq^{\ssr} \brc{ \bar{\mW} , \bar{\bphi }, \bar{\mb} }
		&= \mar_{\rm M}^{\ssr} \brc{\bar\bphi}  \\
		&=	f_Q\brc{ \bff^{\rmF} \brc{\bar{\bphi}}, \bpi^{\rmF} \brc{\bar{\bphi}} , \bar\bphi }\\
		&\leq	f_Q\brc{ \bff^{\rmF} \brc{\bar{\bphi}}, \bpi^{\rmF} \brc{\bar{\bphi}} , \bphi_{\rm new} }\\
		&\leq	f_Q\brc{ \bff^{\rmF} \brc{{\bphi}_{\rm new}}, \bpi^{\rmF} \brc{{\bphi}_{\rm new}} , \bphi_{\rm new} }\\
		&= \mar_{\rm M}^{\ssr} \brc{\bphi_{\rm new}}\\
		&= \mar_\rmq^{\ssr} \brc{ \bar{\mW} , \bphi_{\rm new} , \bar{\mb} }.
	\end{align}
\end{subequations}
Noting that \eqref{Proof_Lem_2_1} is the only constraint we need for the convergence proof, in terms of phase-tuning update algorithm, we conclude that the proof is valid for any alternative of Algorithm~\ref{alg:MM} which converges to its fixed-point in a \textit{non-decreasing fashion}.

\begin{acronym}
	\acro{mimo}[MIMO]{multiple-input multiple-output}
	\acro{mmwave}[mmW]{millimeter wave}
	\acro{tdd}[TDD]{time division duplexing}
	\acro{sinr}[SINR]{signal-to-interference-plus-noise ratio}
	\acro{csi}[CSI]{channel state information}
		\acro{ao}[AO]{alternating optimization}
	\acro{rhs}[r.h.s.]{right hand side}
	\acro{lhs}[l.h.s.]{left hand side}
	\acro{awgn}[AWGN]{additive white Gaussian noise}
	\acro{iid}[i.i.d.]{independent and identically distributed}
	\acro{ut}[UT]{user terminal}
	\acro{tas}[TAS]{transmit antenna selection}
	\acro{rf}[RF]{radio frequency}
	\acro{srzf}[SRZF]{secure regularized zero-forcing}
	\acro{irs}[IRS]{intelligent reflecting surface}
	\acro{mm}[MM]{majorization-maximization}
	\acro{mmse}[MMSE]{minimum mean square error}
	\acro{fp}[FP]{fractional programming}
	\acro{bs}[BS]{base station}
	\acro{bcd}[BCD]{block coordinate descent}
	\acro{qos}[QoS]{quality-of-service}
	\acro{miso}[MISO]{multiple-input single-output}
\end{acronym}

\bibliographystyle{IEEEtran}	
\bibliography{IEEEabrv,Bib}
\end{document}